\newif\ifshort
\shorttrue
\newif\ifpagelimit
\pagelimittrue

\documentclass[journal]{IEEEtran}
\usepackage{graphicx}
\ifCLASSINFOpdf
\else
\fi
\usepackage{url}
\usepackage{adjustbox}
\usepackage{graphicx} % Required for inserting images
\usepackage[colorlinks=false]{hyperref}
\hypersetup{colorlinks=true, unicode=true, linkcolor=[rgb]{0.10,0.05,0.67}, citecolor=[rgb]{0.10,0.05,0.67}, filecolor=[rgb]{0.10,0.05,0.67}, urlcolor=[rgb]{0.10,0.05,0.67}}
\usepackage{physics}
\usepackage{setspace}
\usepackage{amsmath}
\usepackage{amssymb}
\usepackage{stfloats}
\usepackage{amsthm}
\usepackage{multirow}
\usepackage{amsfonts}
\usepackage{xcolor}
\usepackage{graphicx}
\usepackage{subcaption}
\allowdisplaybreaks
\usepackage{cuted}

\usepackage[]{algorithmicx}
\usepackage{algpseudocode,algorithm}
\usepackage{cite}
\usepackage{mathtools}
\usepackage{stmaryrd}
\usepackage[mathscr]{euscript}
\usepackage{array}
\usepackage{mathrsfs}
\newtheorem{prop}{Proposition}
\usepackage{soul}
\usepackage{booktabs}  % For better-looking tables
\usepackage{comment}

\newtheorem{remark}{Remark}
\newtheorem{theorem}{Theorem}
\newtheorem{lemma}{Lemma}
\newtheorem{cor}{Corollary}
\newtheorem{definition}{Definition}
\usepackage{authblk}
\usepackage{tabularx}
\newcommand{\off}[1]{}
\newcommand{\blue}[1]{\textcolor{black}{#1}}
\newcommand{\ale}[1]{\textcolor{red}{[ #1 -- Alejandro ] \normalsize}}
\renewcommand{\blue}[1]{#1}
\newcommand{\newblue}[1]{\textcolor{black}{#1}}

\begin{document}
\title{Adaptive Integrate-and-Fire Time Encoding Machine with Quantization\vspace{-0.1cm}}

%\author{Aseel~Omar,
    %    Alejandro~Cohen,
      % <-this % stops a space
%\thanks{M. Shell was with the Department
%of Electrical and Computer Engineering, Georgia Institute of Technology, Atlanta,

\author{Aseel Omar and Alejandro Cohen\vspace{-0.1cm}\\
Faculty of Electrical and Computer Engineering, Technion—Israel Institute of Technology, Haifa, Israel,\\Emails: aseel.omar@campus.technion.ac.il and alecohen@technion.ac.il\vspace{-0.8cm}
\thanks{Parts of this work were presented at the 32nd European Signal Processing Conference, EUSIPCO 2024 \cite{omar2024adaptive}.}

}

\maketitle

\begin{abstract}
An integrate-and-fire time-encoding machine (IF-TEM) is a power-effective asynchronous sampler that translates amplitude information into non-uniform time sequences. In this work, we propose a novel Adaptive IF-TEM (AIF-TEM) approach, \off{which dynamically adjusts the TEM’s sensitivity in real time to variations in the input signal’s amplitude and frequency.}\newblue{which dynamically adapts the TEM bias and the induced Nyquist ratio in response to temporal amplitude and frequency variations of the input signal.} %\newblue{which dynamically adjusts the effective sampling behavior in response to variations in the input signal’s amplitude and frequency in real time.} 
We provide a comprehensive analysis of AIF-TEM's oversampling and distortion properties.
We also investigate the quantization process for AIF-TEM and analyze the corresponding mean squared error (MSE) bound. Our results show that AIF-TEM achieves significant improvements in rate-distortion performance compared to classical IF-TEM and traditional Nyquist (i.e., periodic) sampling methods for band-limited signals. In particular, AIF-TEM achieves at least a 12 dB reduction in reconstruction\off{sampling} MSE under a fixed oversampling rate. When quantization is considered, AIF-TEM provides at least a 14 dB improvement in quantization MSE compared to IF-TEM. Furthermore, AIF-TEM achieves the same reconstruction accuracy using less than 30\% of the total bits required by IF-TEM, highlighting the superior efficiency of our adaptive approach. Additionally, we introduce a dynamic quantization technique for AIF-TEM, which further improves performance by at least 10 dB compared to its classical quantization baseline.
\end{abstract}
\off{
\begin{abstract}
An integrate-and-fire time-encoding machine (IF-TEM) is an effective asynchronous sampler that translates amplitude information into non-uniform time sequences. In this work, we propose a novel Adaptive IF-TEM (AIF-TEM) approach. This design dynamically adjusts the TEM's sensitivity to changes in the input signal's amplitude and frequency in real-time. We provide a comprehensive analysis of AIF-TEM's oversampling and distortion properties. By the adaptive adjustments, AIF-TEM as we show can achieve significant performance improvements in terms of sampling rate-distortion in a practical finite regime. We demonstrate empirically that in the scenarios tested AIF-TEM outperforms classical IF-TEM and traditional Nyquist (i.e., periodic) sampling methods for band-limited signals. In terms of Mean Square Error (MSE), the reduction reaches at least 12dB (fixing the oversampling rate). Additionally, we investigate the quantization process for AIF-TEM and analyze the quantization MSE bound. Empirical results show that classic quantization for AIF-TEM improves performance by at least 14 dB compared to IF-TEM. We introduce a dynamic quantization technique for AIF-TEM, which further improves performance compared to classic quantization. Empirically, this reduction reaches at least 10 dB compared to classic quantization for AIF-TEM.
\end{abstract}}


% \fi
\IEEEpeerreviewmaketitle

\section{Introduction}
\IEEEPARstart{C}{onventional} analog-to-digital converters (ADCs) transform continuous analog signals into discrete digital values \cite{antoniou2006digital}. These ADCs perform two primary operations: periodic sampling and quantization. As depicted in Fig.~\ref{fig:TEM}(a), periodic sampling captures the amplitude of the signal at uniform intervals, whereas quantization converts the discrete values of these outcomes into bits \cite{antoniou2006digital}.

\begin{figure}
\centering
\includegraphics[width=0.45\textwidth]{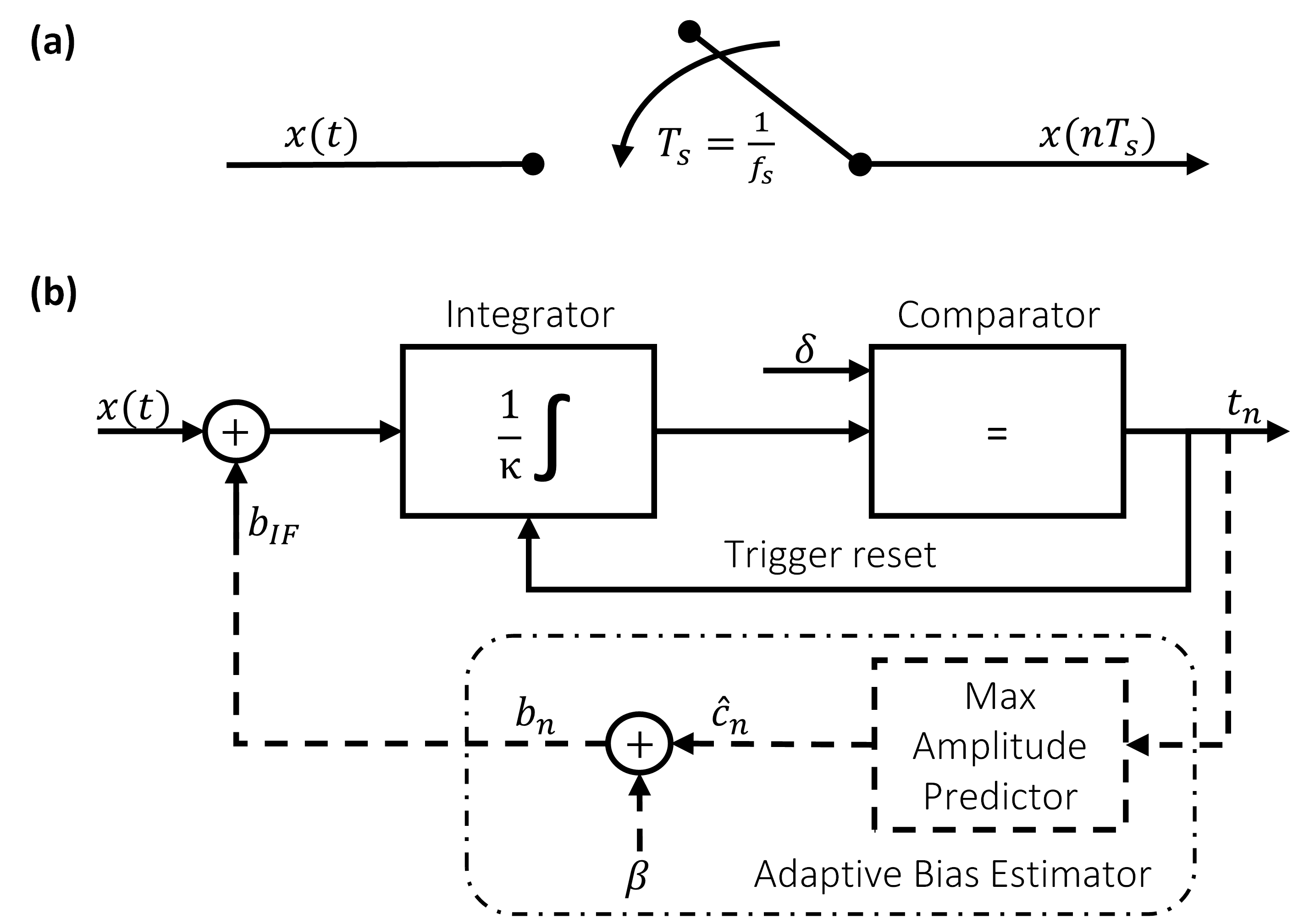}
\caption{\small\label{fig:TEM}(a) Periodic Sampler. (b) IF-TEM model (solid lines) and its adaptive design, AIF-TEM (solid with dashed lines).}
\vspace{-0.5cm}
\end{figure}

Asynchronous ADCs (AADCs) offer an intriguing alternative due to their energy-efficient operation without the sensitive global clock typically required by traditional ADCs \cite{miskowicz2018event,wei2006asynchronous,chen2006asynchronous}. In contrast to classic ADCs, AADCs sample signals non-uniformly, only when detecting specific events. \blue{In level-crossing asynchronous ADCs (LC-ADCs), sampling occurs when the input signal crosses predefined amplitude thresholds \cite{kozmin2008level, akopyan2006level,saeed2021evaluation}. While LC-ADCs are widely studied and implemented at the circuit level, particularly for low-power applications, time-encoding approaches differ by triggering sampling when the integral of the input signal reaches a fixed threshold \cite{lazar2004perfect, lazar2005time, rudresh2020time, koscielnik2007designing}}. Notably, the density of such time samples is directly related to variations in the signal amplitude \cite{lazar2004perfect}.
\off{like amplitude changes \cite{lazar2004perfect,lazar2005time, rudresh2020time ,koscielnik2007designing}. This method, known as "time encoding", generates time samples that provide a discrete representation of the analog signal. Notably, the density of these samples is directly proportional to the variations in signal amplitude \cite{lazar2004perfect}.}

\blue{Time encoding offers advantages such as low supply voltage (entirely processed in the time domain), ultra-low power consumption, and a straightforward architecture \cite{koscielnik2007designing, florescu2022time,kinget2005robustness,gontier2014sampling,naaman2022fri,rastogi2011integrate}.} Several implementations of time encoding are available, including the integrate-and-fire time encoding machine (IF-TEM) \cite{lazar2004time,rastogi2011integrate,ryu2021time,adam2020sampling,tarnopolsky2022compressed}, and the asynchronous sigma-delta modulator (ASDM) \cite{koscielnik2015sample, lazar2004perfect, guo2024ed, florescu2022time}. \blue{In practical applications, the output of time encoding machines must be quantized to enable digital processing \cite{lazar2004perfect}. However, quantization introduces distortion between the original signal and the signal reconstructed from its quantized samples. Several works have analyzed quantization in time encoding and proposed frameworks to reduce its impact \cite{lazar2004perfect, naaman2021time, florescu2023model, koscielnik2011natural}.}

In this work, our focus is on the IF-TEM sampler.\off{, which operates analogously to the functioning of human brain neurons \cite{andrew2003spiking}. Specifically} As depicted in Fig.~\ref{fig:TEM}(b) (solid lines), IF-TEM first biases the input analog signal, \emph{with a fixed bias}, chosen to exceed a constant determined by the signal's maximal amplitude and frequency. Following this, it integrates and contrasts the result against a threshold and records the instances when this threshold is crossed. %first biases (with fixed bias larger than a constant function based on the maximum amplitude and frequency input signal) and integrates the input analog signal. Following this, it contrasts the result against a threshold and records the instances when this threshold is crossed. 
\off{The time differences from the IF-TEM output are then quantized. Previous studies have investigated the quantization process using time encoding for band-limited signals, proposing an upper bound for the Mean Squared Error (MSE) in quantization distortion \cite{lazar2004perfect}. Additionally, other studies have compared IF-TEM with conventional ADCs, using the relationship between the signal’s energy, frequency, and maximum amplitude to study the MSE upper bound \cite{naaman2021time, papoulis1967limits}.However, a primary limitation of the IF-TEM is its \emph{unchanging sensitivity} to variations in signal amplitude and frequency, setting an average Nyquist ratio and oversampling rate }%However, a primary limitation of the IF-TEM is \newblue{its fixed bias and corresponding fixed sampling behavior,} which does not adapt to variations in signal amplitude and frequency, thereby imposing a constant average Nyquist sampling ratio and oversampling rate.
However, a primary limitation of the IF-TEM is \newblue{its fixed bias and resulting constant Nyquist ratio, which do not adapt to variations in signal amplitude and frequency, thereby imposing an oversampling rate
\cite{lazar2004time,naaman2021time} that significantly limits its performance. Specifically, its fixed bias is determined by a global amplitude bound, resulting in persistent oversampling when the temporal input signal amplitude is low.}

To address this limitation and further optimize time encoding schemes, we introduce a new adaptive design of IF-TEM, termed AIF-TEM, as illustrated in Fig.~\ref{fig:TEM}(b). The proposed approach \emph{dynamically adjusts its bias} in response to variations in the input amplitude and frequency, enabling \emph{adaptive adjustments to the actual Nyquist ratio and oversampling rate}. \blue{The primary objective of this design is to improve efficiency, either by reducing the reconstruction error for a given number of transmitted bits or by reducing the total number of bits required to achieve a target reconstruction quality. \newblue{The core engineering principle of AIF-TEM is to maintain the adaptive bias close to the temporal maximum amplitude of the signal, while ensuring the reconstruction condition remains satisfied at all times.} Our study focuses on analog band-limited (BL) signals to demonstrate the performance benefits of AIF-TEM in both sampling-reconstruction and quantization accuracy.}

We thoroughly investigate AIF-TEM's oversampling characteristics and \newblue{reconstruction} distortion \newblue{from the time-encoded samples}, establishing the reconstruction\off{sampling} distortion upper bound as a function of the sampling rate in a practical finite regime. This analysis provides deep insights into the effectiveness of our proposed adaptive design.

Additionally, we investigate the AIF-TEM sampler with quantization using a uniform quantizer. We present a tighter upper bound on the quantization MSE for both IF-TEM and its adaptive design. Our results show that adaptive adjustments in AIF-TEM achieve significant performance improvements in a practical finite regime in terms of reconstruction\off{sampling} and quantization distortion. We also introduce a dynamic quantization scheme that tracks the step size in a practical finite regime based on the estimated amplitude, thereby reducing the overall distortion of the recovered signals.

Finally, we conducted numerical evaluations using synthetic randomized BL signals and real audio signals to test the sampling process, adopting MSE as our primary evaluation metric. We further evaluated both classic and dynamic quantization schemes. The results demonstrate that AIF-TEM significantly outperforms IF-TEM and periodic sampling methods, both in reconstruction\off{sampling} and quantization MSE.

The remainder of this paper is organized as follows: Section~\ref{PRELIMINARIES} provides essential background information and formulates the problem. In Section~\ref{AIF-TEM Algorithm}, we detail our proposed AIF-TEM encoding and decoding algorithm and analyze its oversampling characteristics and distortion. We then present numerical simulations testing the sampling \newblue{and reconstruction} process. Section~\ref{Quantization for AIF-TEM} analyzes the quantization process, introduces the dynamic quantization strategy, and discusses the distortion of the quantization, followed by numerical simulations comparing the quantization performance of AIF-TEM, IF-TEM, and dynamic quantization for AIF-TEM. %We conclude the paper in Section~\ref{ses:con}.
%\IEEEpubidadjcol

\begin{figure}
\centering
\includegraphics[width=0.45\textwidth]{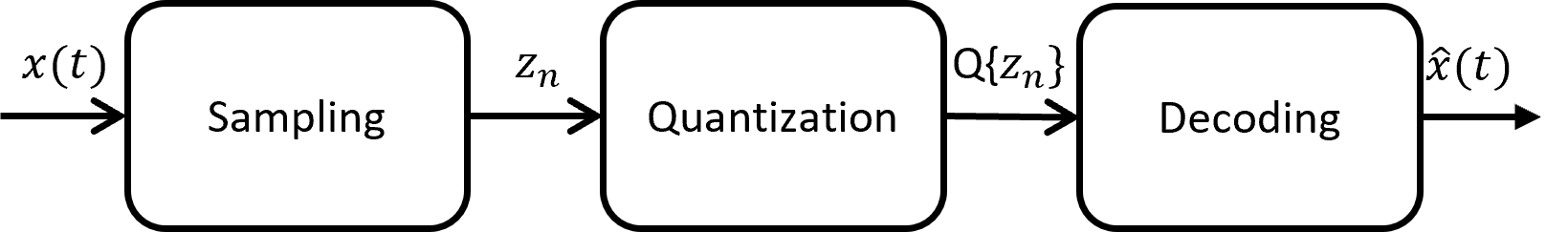}
\caption{\small \label{Encoding_decoding}\off{Generalized scheme for sampling, quantization, and decoding.} \newblue{Generalized signal processing chain comprising time encoding (sampling), quantization, and reconstruction (recovery/decoding).}}
\vspace{-0.5cm}
\end{figure}

%\subsubsection{Subsubsection Heading Here}
%Subsubsection text here.

\section{Problem Formulation And Preliminaries}\label{PRELIMINARIES}
This section introduces the problem formulation and provides background on IF-TEM and periodic samplers.
\subsection{Problem Formulation}
We address the problem of sampling an analog signal $x(t)$ and then reconstructing it. 
We characterize $x(t)$ as follows.
\off{\begin{definition}
    A signal $x(t)$ is termed $c_{\text{max}}$-bounded and $2\Omega \text{-BL}$ signal, if its amplitude is confined within $|x(t)| \leq c_{\text{max}}$ and its Fourier transform is zero for frequencies outside the closed interval $[- \Omega, \Omega]$. 
\end{definition}
\ale{Optional only for Alejandro: Let $x(t), t \in \mathbb{R}$, be a real-valued signal band-limited to the frequency support $\mathcal{B} = [-\Omega, \Omega]$, where $2\Omega$ denotes the total bandwidth and $\Omega$ is the maximum frequency component.}}
\begin{definition}
    A signal $x(t), t \in \mathbb{R}$ is termed $c_{\text{max}}$-bounded and $2\Omega \text{-BL}$ signal, i.e., band-limited to the frequency support $\mathcal{B} = [-\Omega, \Omega]$, where $2\Omega$ denotes the total bandwidth and $\Omega$ is the maximum frequency component, if its amplitude is confined within $|x(t)| \leq c_{\text{max}}$ and its Fourier transform is zero for frequencies outside the closed interval $\mathcal{B}$. 
\end{definition}
We assume that $x(t)$ has finite energy $E$ as follows.
\begin{definition}
A signal $x(t)$ has finite energy $E \in \mathbb{R}$ if
\[
E = \int_{-\infty}^{\infty} \abs{x(t)}^2 \ dt < \infty.
\]
\end{definition}
Moreover, we consider that for $c_{\text{max}}$-bounded and $2\Omega \text{-BL}$ signal with finite energy $E$, the amplitude upper bound $c_{\text{max}}$ is related to the bandwidth $\Omega$ as derived in \cite{papoulis1967limits}
\begin{equation}\label{c_max and frequncy relation}
    c_{\text{max}} = \sqrt{E\frac{\Omega}{\pi}}.
\end{equation}
\off{
We address the problem of sampling an analog signal $x(t)$ and then reconstructing it. 
We characterize $x(t)$ as follows
\begin{definition}
    A signal $x(t)$ is termed $c_{\text{max}}$-bounded and $2\Omega \text{-BL}$ signal, if its amplitude is confined within $|x(t)| \leq c_{\text{max}}$ and its Fourier transform is zero for frequencies outside the closed interval $[- \Omega, \Omega]$. 
\end{definition}
\begin{definition}
    A signal $x(t)$ is termed $2\Omega \text{-BL}$  if its Fourier transform is zero for frequencies outside the closed interval $[- \Omega, \Omega]$.
\end{definition}
\begin{definition}
    A signal $x(t)$ is $c_{\text{max}}$-bounded if its amplitude is confined within $|x(t)| \leq c_{\text{max}}$
\end{definition}

\begin{definition}
A signal $x(t)$ has finite energy $E \in \mathbb{R}$ if
\[
E = \int_{-\infty}^{\infty} \abs{x(t)}^2 \ dt < \infty
\]
\end{definition}

Given a $2\Omega \text{-BL}$ and $c_{\text{max}}$-bounded signal with finite energy $E$, the amplitude's upper bound $c_{\text{max}}$ is related to the bandwidth $\Omega$ as derived in \cite{papoulis1967limits}
\begin{equation}\label{c_max and frequncy relation}
    c_{\text{max}} = \sqrt{E\frac{\Omega}{\pi}}.
\end{equation}
}
Fig.~\ref{Encoding_decoding}\off{outlines a generalized scheme that encompasses sampling, quantization, and recovery processes.} \newblue{outlines a generalized signal processing chain that encompasses time encoding (sampling), quantization, and reconstruction (recovery/decoding)\footnote{\newblue{We note that in this work, the terms \emph{time encoding} and \emph{sampling} are used interchangeably to describe the non-uniform event-based sampling operation of IF-TEM and AIF-TEM. Likewise, the terms \emph{reconstruction}, \emph{recovery}, and \emph{decoding} are used synonymously to denote the inverse operation that reconstructs the input signal from its time-encoded samples.}}.} The input signal $x(t)$ undergoes a sampling process that results in discrete measurements $\{z_n\}_{n \in \mathcal{Z}}$. 
Sampling may be performed using a conventional sampling method as shown in Fig. \ref{fig:TEM}(a) or through IF-TEM or AIF-TEM methods as shown in Fig. \ref{fig:TEM} (b).
After the sampling process, discrete measurements from the sampling stage are quantized to produce bit representations $Q(z_n)=\tilde{z}_n$.
The quantized measurements are then decoded to reconstruct the signal $\hat{x}(t)$.
In practical scenarios involving finite regimes and quantization, a distortion arises between the original input signal $x(t)$ and the recovered signal $\hat{x}(t)$. This distortion is quantified using the Mean Squared Error (MSE), expressed in decibels (dB) as follows
\begin{equation}\label{MSE}
\text{MSE} \triangleq 20 \log_{10} \left( (1/\sqrt{T})\|x(t) - \hat{x}(t)\|_{L_2[0,T]} \right)\quad\text{[dB]}.
\end{equation}
%We aim to refine the sampling and quantization processes by reducing \textcolor{purple}{the overall reconstruction distortion arising from both the sampling and quantization stages, for a given number of bits used to quantize the signal, or by reducing the total number of bits required to achieve a target reconstruction distortion.} 
\newblue{We aim to refine the chain process (see Fig.~\ref{Encoding_decoding}) by reducing the overall reconstruction distortion arising from the entire stages in the chain, for a given number of bits used to quantize the signal, or by reducing the total number of bits required to achieve a target reconstruction distortion.}

\subsection{Related Work}
\blue{Integrate-and-Fire (IF) models are widely used in neuroscience to describe brain activity and neuron spiking behavior \cite{andrew2003spiking}. To better mimic neuronal spiking, Leaky IF (LIF) model introduces a leakage term, while the Exponential IF (EIF) model incorporates exponential spike generation \cite{fourcaud2003spike}. To further capture neural adaptation behavior, the Adaptive Exponential IF (AdEx) model \cite{brette2005adaptive} adds an adaptation current that dynamically influences the neuron's firing rate. In addition, \cite{liu2001spike} proposed a general adaptation mechanism for LIF neurons. These models primarily aim to prevent excessive neuronal firing. Beyond theoretical neuroscience, several studies \cite{kwon2021low,millner2010vlsi} have implemented IF models in VLSI circuits, enabling real-time, low-power neuromorphic processing. These hardware-based models emulate biological neurons for neuromorphic computing applications and often include bias currents as part of the neuron dynamics.}

\newblue{While the aforementioned works focus on biological neuron modeling and neuromorphic implementations, where bias terms can be employed (e.g.,~\cite{davies2018loihi}), their objectives differ from the present work. In neuromorphic systems, the bias typically serves as a neuronal parameter governing firing behavior, whereas in the proposed AIF-TEM (see Section~\ref{AIF-TEM Algorithm}), it is adaptively controlled to regulate the effective Nyquist ratio and support reconstruction guarantees for bandlimited signals.}
\blue{
%While the aforementioned works focus on biological neuron modeling, \off{without incorporating a bias term and primarily aim to regulate neuron spiking, }they do not provide signal reconstruction guarantees in the context of bandlimited signal processing. \off{grounded in an adaptive Nyquist condition. } In contrast, as presented in Section~\ref{AIF-TEM Algorithm}, this work proposes an adaptive bias mechanism for IF-TEM sampler in BL signal processing. 
Unlike AdEx, which adapts spiking dynamics, the proposed method herein derives and adapts the Nyquist condition to ensure perfect signal recovery. It dynamically adjusts event-based sampling to improve the efficiency of analog signal encoding. \off{This is achieved by reducing both \newblue{reconstruction}\off{sampling} and quantization distortion, as well as minimizing bit usage.}\textcolor{black}{This is achieved by reducing the overall reconstruction distortion arising from the entire stages in the chain (see Fig.~\ref{Encoding_decoding}), while minimizing bit usage.} Our results demonstrate that this adaptive bias enhances signal reconstruction efficiency compared to non-adaptive methods, making it fundamentally distinct from neuroscience-inspired adaptive IF models that focus on capturing neural dynamics, rather than linking their results to the derived adaptive Nyquist requirement suggested herein for BL signals.}

\blue{Asynchronous Sigma-Delta Modulator (ASDM) was first introduced by \cite{kikkert1975asynchronous} as a time encoding scheme based on threshold-triggered integration.\off{Later, Lazar introduced the Integrate and Fire Time Encoding Machine (IF-TEM) in \cite{lazar2004time}, following a similar event-based encoding approach.} \newblue{Building on earlier integrate-and-fire neuron models (e.g.,~\cite{stein1965theoretical}), Lazar et al. in \cite{lazar2004time} formalized the Integrate-and-Fire Time Encoding Machine (IF-TEM) framework for bandlimited signal processing and established associated reconstruction guarantees.} Both ASDM and IF-TEM share a common structure, consisting of an adder, an integrator, and a comparator.} \blue{The key difference between these models lies in that ASDM alternates between positive and negative fixed biases and compares against two thresholds (positive and negative). In contrast, IF-TEM applies only a positive bias and uses a single fixed threshold for comparison. Amplitude Adaptive ASDM (A-ASDM) was introduced in \cite{ozols2013amplitude} as an extension of ASDM with an adaptive bias mechanism that adjusts based on the time-varying envelope of the signal to reduce power consumption. Later, \cite{ozols2016amplitude} proposed an A-ASDM design that eliminates explicit envelope encoding while still adapting the bias through a signal-derived function. Further developments in \cite{shavelis2017amplitude} focused on hardware circuit implementations, demonstrating the feasibility of A-ASDM for real-time signal processing. }

\blue{While A-ASDM primarily focuses on power efficiency and relies on amplitude envelope estimation using the input signal, \newblue{A-ASDM requires direct access to the analog input waveform. This approach typically involves additional analog circuitry and continuous observation of the signal amplitude. } In contrast, the proposed Adaptive IF-TEM \newblue{detailed in Section~\ref{AIF-TEM Algorithm}, performs bias adaptation based on amplitude estimates inferred from output event timing without direct access to the input waveform. This enables a fully event-driven architecture that avoids the need for additional analog components.} \newblue{In scenarios where only the event stream is available, such as distributed sensing systems, neuromorphic hardware, or ultra-low-power time-domain implementations, adaptations based on output timings remain applicable and preserve the asynchronous nature of the system.} In addition, AIF-TEM introduces a novel adaptive sampling condition based on an analytically derived adaptive Nyquist requirement. The sampling rate is related to the amplitude and frequency variations of the input signal, enabling dynamic control of oversampling while preserving the signal reconstruction guarantees. These contributions result in improved sampling efficiency, reduced quantization error, and a lower number of transmitted bits.}

\subsection{IF-TEM vs. Periodic Sampler}\label{IF-TEM vs. Periodic Sampler}

Conventional sampling techniques, such as periodic sampling depicted in Fig.~\ref{fig:TEM}(a), involve measuring the amplitude of a signal at uniform time intervals. For an input signal $x(t)$, this approach yields discrete samples $x(nT_s)$ with a consistent sampling interval $T_s$. In contrast, the IF-TEM technique samples $x(t)$ non-uniformly\footnote{\newblue{Although IF-TEM sampling is inherently non-uniform, some reconstruction approaches map the resulting event times onto a uniform or structured \emph{computational} grid for recovery; see, e.g., \cite{florescu2015novel}.}}, focusing on capturing time instances rather than amplitude values.

An IF-TEM is characterized by three parameters: a fixed bias $b_{\text{IF}}$, a scaling factor $\kappa$, and a threshold $\delta$, as depicted in Fig.~\ref{fig:TEM}(b) (solid lines). The input to the IF-TEM, $x(t)$, is a $c_{\text{max}}$-bounded signal. The time-encoding process begins by adding the bias $b_{\text{IF}}$ to $x(t)$. This augmented signal, $x(t) + b_{\text{IF}}$, is subsequently scaled by $1/\kappa$ and integrated. To ensure that the integrator's output continuously rises, it is essential that $b_{\text{IF}} > c_{\text{max}}$. The moments, or firing times, denoted as ${\{t_n\}_{n\in \mathcal{Z}}}$, are recorded when the integral surpasses the threshold $\delta$. After each recording, the integrator is reset to zero.
For $t_n > t_{n-1}$, the integrator's output is given by
\begin{equation}\label{Integ out-IF-TEM}
%y_n(t) = \int_{t_{n-1}}^{t} \frac{1}{\kappa}(x(s)+b_{\text{IF}}) \ ds.
 \frac{1}{\kappa}\int_{t_{n-1}}^{t_n} (x(s)+b_{\text{IF}}) \ ds = \delta.
\end{equation}

Thus, the relationship between the input $x(t)$ and its time output ${{t_n},{n\in \mathcal{Z}}}$ of the IF-TEM is given by 
\begin{equation}\label{IF relation times ampl}
P_n \cong \int_{t_{n-1}}^{t_n} x(s) \ ds = -b_{\text{IF}}(t_n-t_{n-1}) +\kappa\delta.
\end{equation}
\off{
As a result of applying the mean theorem to the term on the left side of the previous Eq, the following result is obtained:
\begin{equation}\label{Tn Equall}
    T_n = \frac{\kappa\delta}{b_{\text{IF}}+x(\zeta_n)}
\end{equation}
Where $ \zeta_n \in [t_{n+1},t_n]$,}
Subsequently, based on \cite{lazar2004time}, the time differences between firing times, denoted as $T_n=t_n-t_{n-1}$, are bounded by
\begin{equation}\label{Tn Bound -IF-TEM}
\Delta t_{c_{\text{min}}} \triangleq \frac{\kappa \delta}{b_{\text{IF}}+c_{\text{max}}} \leq T_n \leq \frac{\kappa \delta}{b_{\text{IF}}-c_{\text{max}}} \triangleq 
\Delta t_{c_{\text{max}}}.
\end{equation}

\off{\subsection{Recovery of Signal from its Samples}}
\subsection{\newblue{Reconstruction from IF-TEM and Periodic Sampling}}
For the periodic sampling method, the Shannon-Nyquist theorem dictates that a $2\Omega$-BL signal, $x(t)$, can be perfectly reconstructed from its discrete samples $x(nT_s)$ when sampled at a rate no less than the Nyquist rate, $\frac{\Omega}{\pi}$ \cite{nyquist1928certain}.

\blue{
The recovery of an analog signal from its time-encoded output in IF-TEM has been extensively studied in previous works. Lazar and Tóth \cite{lazar2004perfect} proposed an algorithm for the perfect reconstruction of BL and 
$c_{\text{max}}$-bounded signals from an ASDM sampler. This approach was later extended in \cite{lazar2004time} to adapt the method for IF-TEM with a refractory period. Subsequently, Lazar et al.\ in \cite{lazar2006real,lazar2008overcomplete} further developed the work in \cite{lazar2004perfect} by introducing a real-time recovery algorithm for ASDM output that utilizes a segment-based recovery approach. Moreover, Lazar et al.\ in \cite{lazar2005fast} proposed a fast recovery algorithm for ASDM.}

\blue{While these works focused on bandlimited signals, Thao et al.\ \cite{thao2023bandlimited} examined various methods for decoding bandlimited signals from leaky IF-TEM output, where the leaky mechanism introduces an additional decay term that affects the timing of threshold crossings. They extended the Projection Onto Convex Sets (POCS) method and compared it to previous decoding approaches. In contrast, Alexandru et al.\ \cite{alexandru2019reconstructing} analyzed the reconstruction of non-bandlimited signals from IF-TEM output, addressing a different class of signals. Recent works such as \cite{florescu2023time} considered recovery of sparse signals using flexible filter-based time encoding frameworks.}

\off{
With IF-TEM, the reconstruction of a $2\Omega \text{-BL}$ signal from its time output has been extensively studied for input signals that are $c_{\text{max}}$-bounded with finite energy $E$ \cite{lazar2004time, lazar2005time, lazar2003time}.}

\blue{In this work, we} adopt the IF-TEM sampling and reconstruction mechanism as outlined in \cite{lazar2004time} with zero refractory period, which demonstrates that such signals can be perfectly reconstructed using an IF-TEM with parameters $\{b_{\text{IF}}, \kappa, \delta\}$, if $b_{\text{IF}} > c_{\text{max}}$ and the Nyquist ratio, $r_c$, is given by
\begin{equation}\label{recovery}
r_c\triangleq\frac{\kappa \delta}{b_{\text{IF}}-c_{\text{max}}}\frac{\Omega}{\pi} < 1.
\end{equation}
This constraint stipulates that the interval between two successive trigger times must not exceed the inverse of the Nyquist rate. We note that the IF-TEM employs a fixed upper-bound bias, setting an average sampling and Nyquist ratio that remains unaffected by variations in the input signal’s amplitude and frequency, which significantly limits its performance.

\subsection{Quantization}

In classical sampling, each amplitude measurement $ z_n = x(nT_s) $ is quantized. Given that these measurements fall within the amplitude range $ [-c_{\text{max}}, c_{\text{max}}] $, the quantization step size $ \Delta_{\text{per}} $ for a uniform $ K $-level quantizer is given by
\begin{equation*}
    \Delta_{\text{per}} = \frac{2c_{\text{max}}}{K}.
\end{equation*}

\blue{In contrast, for time encoding machines (TEMs), the sampled output consists of trigger times $ z_n = t_n $, which are then quantized. Several works have explored the quantization of time-encoded signals. The authors of \cite{lazar2004perfect} studied uniform quantization of time samples from an ASDM and provided a bound on quantization-induced reconstruction distortion. The authors of \cite{koscielnik2011natural} further investigated ASDM quantization, showing that the sampling density varies with the input signal amplitude, resulting in natural compression or expansion effects. The authors of \cite{florescu2023model} introduced a model-driven quantization framework for ASDM, proposing the QTEM (Quantized Time Encoding Machine) approach, in which sampling and quantization are combined into a single operation. Their results show that QTEM achieves lower quantization reconstruction distortion compared to conventional uniform quantization.}

For IF-TEM, uniform quantization of time samples has been studied in \cite{naaman2021time}, where it is compared against uniform amplitude quantization in classical sampling. The authors derived an upper bound on the quantization mean squared error (MSE) and analyzed how IF-TEM quantization is influenced by signal frequency and amplitude dynamics.

The quantization step size $ \Delta_{\text{IF}} $ for a $ K $-level uniform quantizer is determined based on the dynamic range of the time differences $ T_n $, as follows \cite{naaman2021time}

\begin{equation}\label{step size for IF}
    \Delta_{\text{IF}} = \frac{\kappa \delta}{(b_{\text{IF}}+c_{\text{max}})(b_{\text{IF}}-c_{\text{max}})}\frac{2c_{\text{max}}}{K}.
\end{equation} 

Considering any $ \gamma > 1 $, let $ b_{\text{IF}} = \gamma c_{\text{max}} $, with $ b_{\text{IF}} > c_{\text{max}} $. Unlike classical sampling, where increasing the sampling frequency increases the quantization step size, IF-TEM exhibits an inverse relationship between signal frequency and quantization resolution \cite{naaman2021time}. Specifically, an increase in signal frequency leads to a decrease in the quantization step size $ \Delta_{\text{IF}} $, since higher frequency signals exhibit larger amplitude variations, affecting inter-event times. The upper bound for the quantization MSE in IF-TEM is given by

\begin{equation}\label{quantization mse bound IF-TEM}
    \mathbb{E}[ \mathcal{E}^2 ]_{IF} \leq \frac{\Omega(b_{\text{IF}}+c_{\text{max}})}{\pi\kappa \delta}\left(\frac{b_{\text{IF}}+c_{\text{max}}}{1-r_c}\right)^2\frac{\Delta^2_{IF}}{12}.
\end{equation}

\section{AIF-TEM Algorithm}\label{AIF-TEM Algorithm}

In this section, we introduce the Adaptive Integrate-and-Fire Time Encoding Machine (AIF-TEM), a novel machine that dynamically adapts its operation to the amplitude and frequency variations of the input. \blue{The core idea of AIF-TEM is to adaptively control the sampling rate by modifying the bias $b_n$ based on recent amplitude behavior, allowing better alignment with signal dynamics, unlike classical IF-TEM, which maintains a fixed rate. A key feature of AIF-TEM is the Adaptive Bias Estimator, depicted by dashed lines in Fig.~\ref{fig:TEM}(b). This estimator dynamically adjusts the bias $b_n$ to account for variations in signal amplitude. The estimator includes the Max Amplitude Predictor (MAP) block, which estimates amplitude variations and assists in bias adaptation.}

Consider a $c_{\text{max}}$-bounded input signal $x(t)$. For AIF-TEM, the time output is represented as $\{t_n\}, \forall n \in \mathcal{Z}$. In this context, each iteration, represented by $n$, captures the duration between two successive trigger events, $t_{n-1}$ and $t_n$. We aim to determine the maximum amplitude value, $c_n$, within each iteration $n$. This is achieved by examining a time window that spans from the current trigger time $t_n$ back to the preceding $w$ trigger times. Here, $w$ denotes the number of past trigger intervals in the window. Therefore, the maximum amplitude value $c_n$ within this window is given by

\begin{equation}\label{c_n}
    c_n \triangleq \max_{t_{n-w} \leq t \leq t_n} (|x(t)|).
\end{equation}

The MAP block estimates $c_n$, which is then used to adapt the bias $b_n$ accordingly. To ensure that the integrator output increases consistently, we introduce the following definition.
\blue{\begin{definition} [Correct MAP Operation] \label{map_success} 
The MAP block operates \textit{correctly}, ensuring the \textit{successful} functioning of AIF-TEM, if and only if
\begin{equation*} 
b_n \geq c_n  \quad \forall n \in \mathcal{Z}. 
\end{equation*} 
\end{definition}}
Ideally, a desired operational mode ensures that the estimated amplitude $\hat{c}_n$ \newblue{tracks the temporal maximum amplitude $c_n$ while preserving the\off{ Nyquist reconstruction condition} adaptive Nyquist condition required for reconstruction (see \eqref{r_a general}). This alignment reduces unnecessary oversampling when the signal amplitude is low, while maintaining the sampling density required for accurate recovery when the amplitude increases. Encoding stability is ensured by enforcing $b_n \ge c_n$ (Definition~\ref{map_success}), which guarantees monotonic integrator growth. In addition, the bias is constrained within predefined limits, $b_{\min} \le b_n \le b_{\max}$, to prevent divergence of the adaptive update and to maintain robustness against estimation errors.}

%\newblue{closely tracks the local maximum $c_n$,  of the input signal. This allows the adaptive bias $b_n$ to remain aligned with the signal amplitude, avoiding unnecessary firing events when the signal is weak while preserving the sampling conditions required for accurate reconstruction when the amplitude increases. Additionally, the bias is constrained within predefined limits, $b_{\min} \le b_n \le b_{\max}$, to prevent excessive firing activity and to ensure controlled adaptation to amplitude variations.}\off{allowing accurate adaptation of the bias and thereby maximizing the performance of AIF-TEM. Additionally, the bias should be constrained within predefined limits, $b_{\min} \leq b_n \leq b_{\max}$, to ensure the system remains stable and responsive to amplitude changes.}

\blue{The AIF-TEM algorithm, with a particular emphasis on the bias estimator, supports various operational modes, offering flexibility in how the bias is updated. This adaptability allows AIF-TEM to be optimized for various signal characteristics, opening avenues for further research and improved performance in real-world applications.} The subsequent sections detail the encoding process, the implementation of the bias estimator, and alternative operational modes.

\subsection{Encoding Process}

The input signal is first biased by $b_n$, resulting in $x(t) + b_n$. For all $t \in \mathbb{R}$ and $n \in \mathcal{Z}$ with $t \geq t_{n-1}$, the output of the integrator is given by
\begin{equation}\label{Integ out-AIF-TEM}
y_n(t) \triangleq  \frac{1}{\kappa}\int_{t_{n-1}}^{t} (x(s)+b_n) \ ds.
\end{equation}
Assuming that the MAP block operates correctly as defined in Definition~\ref{map_success}, the output $y_n(t)$ will increase monotonically. \blue{At time $t = t_n$ (with $t_n > t_{n-1}$), the output $y_n(t)$ reaches the threshold $\delta$, satisfying the relation in \eqref{IF relation times ampl} with the adaptive bias $b_n$ replacing the fixed bias $b_{\text{IF}}$.}

\off{
\begin{equation}\label{integ}
P_n \triangleq \int_{t_{n-1}}^{t_{n}} x(s) \ ds = -b_n(t_{n}-t_{n-1}) + \kappa\delta.
\end{equation}
\blue{Equation~\eqref{integ} follows directly from \eqref{IF relation times ampl}, where the bias is now adaptive rather than fixed.}

\blue{Since the input signal is bandlimited and continuous, applying the mean value theorem to the left-hand side of \eqref{IF relation times ampl} guarantees the existence of a point $\zeta_n \in [t_{n-1}, t_n]$ such that
\begin{equation}\label{Tn equall AIF}
   x(\zeta_n) =  -b_n + \frac{\kappa\delta}{T_n}.
\end{equation}}

\blue{This equation establishes a direct link between time differences and signal amplitude, allowing us to estimate amplitude variations.}}
Once $t_n$ is recorded, the integrator resets, and the algorithm uses the time difference $T_n = t_n - t_{n-1}$ to predict the next bias value $b_{n+1}$.\off{ In this context, the MAP block is crucial, estimating $\hat{c}_n$ and forecasting $\hat{c}_{n+1}$ from the previous $w$ estimated values $\hat{c}_k$, for which $n-w \leq k \leq n-1$.}

During the interval $t_{n-1} < t \leq t_{n}$, the signal's amplitude is constrained by $|x(t)| \leq c_n < b_n$. By leveraging this inequality and substituting it into \eqref{IF relation times ampl} with the adaptive bias $b_n$, we determine a bound for the duration between successive trigger times, $T_n$. This bound is given by
\begin{equation}\label{Tn bound aif}
\Delta t_{\text{min}}[n] \triangleq \frac{\kappa \delta}{b_n+c_n} \leq T_n \leq \frac{\kappa \delta}{b_n-c_n} \triangleq \Delta t_{\text{max}}[n].
\end{equation}

Initially, the bias is set such that $b_{n=1} > c_{\text{max}}$. This configuration ensures that the integrator's output, $y_1(t)$, increases monotonically for all $t \geq t_0$. As a result, $y_1(t)$ surpasses the threshold $\delta$ at a time $t_1 > t_0$. Therefore, the equalities in \eqref{IF relation times ampl} and \eqref{Tn bound aif} hold for $n=1$.

%%here is the end of the section
\off{
However, since $b_n = \hat{c}_n +\beta = c_n+\epsilon_n+\beta$, where $\epsilon_n$ is the estimation error. Let $\epsilon_{\text{max}}$ be the max estimation error that could occur, 
\[
\epsilon_{\text{max}} = \max_n{\abs{\epsilon_n}}
\]
We get 
\begin{equation}\label{right Tn bound}
\Delta t_{\text{max}}[n]  \leq  \frac{\kappa \delta}{c_n +\beta + \epsilon_n-c_n} = \frac{\kappa \delta}{\beta - \epsilon_{\text{max}}} = T_{\text{max}}
\end{equation}
And the other side
\begin{equation}\label{left T_n bound}
\Delta t_{\text{min}}[n] \geq  \frac{\kappa \delta}{c_n +\beta + \epsilon_n+c_n} = \frac{\kappa \delta}{2c_{\text{max}}+ \beta+ \epsilon_{\text{max}}} = T_{\text{min}}
\end{equation}
%%here is the end of the section
}

\blue{Next, we detail the operation of the Adaptive Bias Estimator, describing our primary mode of operation while briefly outlining alternative approaches.}

\subsection{Adaptive Bias Estimator Operation} \label{MAP Modes}

In this work, we focus on a specific practical and computationally efficient MAP operation mode that demonstrates the effectiveness of the adaptive mechanism in AIF-TEM.
In this mode, the bias is updated using the relation
$b_n = \hat{c}_n + \beta$,
where $\beta>0$ serves as a safety margin \newblue{that is chosen to help maintain the condition $b_n \ge c_n$ in Definition~\ref{map_success} in practice},\off{to ensure encoding stability} while preserving the Nyquist ratio (discussed in Section~\ref{decoding process}).

For amplitude estimation, we adopt an Exponentially Weighted Moving Average (EWMA) filter due to its simplicity, computational efficiency, and ability to track smooth amplitude variations while filtering out noise. These properties make it well-suited for the bandlimited input signals considered in our application. 
\newblue{Dividing both sides of \eqref{IF relation times ampl} by $T_n$ and with the adaptive bias $b_n$ (replacing the fixed bias $b_{\text{IF}}$) yields $\bar{x}_n \triangleq \frac{1}{T_n} \int_{t_{n-1}}^{t_n} x(s)\,ds 
= -b_n + \frac{\kappa\delta}{T_n}$, where $\bar{x}_n$ represents the average value of the input signal over the interval $(t_{n-1}, t_n)$ and serves as the amplitude-related quantity used by the MAP block.}
The estimated amplitude is then computed as
\begin{equation*} 
\hat{c}_n = \alpha_1  \newblue{\bar{x}_n}   + (1-\alpha_1)\hat{c}_{n-1},
\end{equation*}
where  $0\leq\alpha_1\leq 1$  is a weighting  factor\footnote{\label{note:iir_p}\newblue{The parameters $\alpha_1$ and $\alpha_2$ control the adaptation dynamics of the estimator. In the EWMA recursion, values of $\alpha_1$ close to $1$ yield faster tracking with reduced smoothing, whereas smaller values provide stronger smoothing but slower adaptation. The effective memory length of the filter is approximately $1/(1-\alpha_1)$ samples, which provides a practical guideline for selecting $\alpha_1$ according to the expected rate of amplitude variation. The parameter $\alpha_2$ scales the standard-deviation-based correction term and determines the responsiveness of the bias update to amplitude fluctuations.}}. To enhance prediction accuracy, we iteratively compute the standard deviation of past estimates, $s_n$, via Welford's method \cite{welford1962note}
\begin{equation*} 
s_n^2 = \frac{1}{n}\sum_{i=1}^{n}(\hat{c}_i-\mu_n)^2, \quad \text{where} \quad \mu_n =  \frac{1}{n}\sum_{i=1}^{n}\hat{c}_i.
\end{equation*}
Using this estimate, the next amplitude value is predicted as 
\[
\hat{c}_{n+1} = \hat{c}_n + \alpha_2 s_n,
\]
where $\alpha_2$ is a weighting factor decided by the user$^{\ref{note:iir_p}}$. Finally, we compute the maximum value among the previous $w$ predicted amplitude values.

\newblue{The performance of the MAP block in practice is defined as its ability to satisfy the correct MAP operation condition $b_n \ge c_n$ (Definition~3), which guarantees monotonic growth of the integrator and preserves the adaptive Nyquist condition\off{ required for stable reconstruction, while maintaining the bias close to the \textcolor{red}{temporal} amplitude envelope of the signal}. Since the amplitude estimate $\hat{c}_n$ is obtained herein via causal EWMA filtering and statistical prediction, transient underestimation may occur, potentially yielding $b_n < c_n$. In this case, the integrator may momentarily violate the strict monotonicity condition; however, this event is mitigated in practice by incorporating the design margin parameter $\beta$, constraining $b_n \in [b_{\min}, b_{\max}]$, and including the variance-based correction term $\alpha_2 s_n$ in the prediction step. These mechanisms significantly reduce the probability and duration of such events and restore the condition $b_n \ge c_n$ in subsequent iterations.}
%
%\newblue{
%The performance of the MAP block is defined in terms of its ability to maintain the adaptive sampling condition while avoiding excessive oversampling. Since the amplitude estimate is obtained via EWMA filtering, transient underestimation may occur, potentially yielding $b_n < c_n$. In practice, this effect is mitigated by incorporating the design margin parameter $\beta$, constraining $b_n \in [b_{\min}, b_{\max}]$, and including the standard deviation term in the prediction step, which improves robustness of the adaptive update.}
%
%\newblue{The parameters $\alpha_1$ and $\alpha_2$ control the adaptation dynamics of the estimator. Values of $\alpha_1$ close to 1 provide stronger smoothing and are suitable for slowly varying signals, while smaller values allow faster adaptation to rapid amplitude changes. The parameter $\alpha_2$ controls how strongly the bias reacts to changes in the signal amplitude.}
%
While the considered mode herein may not be optimal in practice, its simplicity enables efficient processing of AIF-TEM, which significantly outperforms existing approaches, as demonstrated in Sections~\ref{EVALUATION RESULTS} and~\ref{Evaluation results for Quantization Process}. We note that alternative bias estimation approaches exist in the literature, which may offer opportunities for further improvement. Potential alternatives that can meet the requirements in Definition~\ref{map_success} include: 1) \blue{Autoregressive with Exogenous Input (ARX) Model – predicts future values based on a linear combination of past values \cite{xie2021robust}.}  
2) \blue{Kalman Filter – uses a series of time-based measurements to produce optimal estimates of unknown variables \cite{chui2017kalman}.}  
3) \blue{Recursive Least Squares (RLS) – an adaptive filtering algorithm that efficiently updates estimates in real time \cite{cioffi1984fast}.} \blue{We leave the exploration of these methods for future work.}

\subsection{Decoding Process}\label{decoding process}
This section outlines the decoding process for a $2\Omega$ BL input signal, $x = x(t), t\in \mathbb{R}$, using the output of AIF-TEM, denoted by $\{t_n\}_{n = 0}^{N}$, with $N$ denotes the total number of samples.
The decoding algorithm proposed herein extends the framework introduced by Lazar \cite{lazar2004time,lazar2004perfect} and builds upon classical reconstruction techniques employed in non-adaptive TEM schemes \cite{benedetto1994theory,duffin1952class}. In particular, we adopt a segment-based reconstruction approach, where the decoder dynamically adjusts the bias during reconstruction in accordance with the estimated maximum amplitude observed over the preceding $w$ sampling intervals, as defined in \eqref{c_n}. 
\newblue{In practical implementation, synchronization between the encoder and decoder is achieved by sharing the adaptive bias update rule and its initialization parameters. In the sampling framework, the adaptive bias sequence ${b_n}$ is deterministically generated from the time differences ${T_n}$ via the MAP block described in Section~\ref{MAP Modes}. Since both encoder and decoder implement the same update mechanism with identical initialization, the decoder can regenerate the corresponding bias values.}

%\newblue{In the quantized setting discussed in Section~\ref{Quantization for AIF-TEM}, bias values are selected in a segment-wise manner from a predefined discrete set. The corresponding bias indices are encoded using $R_b$ bits and transmitted to the decoder. Since the bias range and step size $\Delta_b$ are shared between encoder and decoder, each received index uniquely determines the active bias value in the corresponding segment, thereby ensuring synchronization during reconstruction.}
%Lazar's seminal research introduced a decoding algorithm for TEM \cite{lazar2003time} that is predicated on the principles of irregular sampling decoding \cite{benedetto1994theory} and grounded in frame theory \cite{duffin1952class}, with further elaboration for IF-TEM decoding process in \cite{lazar2004time} and for Sigma-Delta modulator in \cite{lazar2004perfect}. Building upon Lazar's framework, our approach introduces an adaptation by dynamically adjusting the bias based on the maximum amplitude value observed in the preceding $w$ sampling points, as detailed in \eqref{c_n} and Definition~\ref{map_success}.
\off{, particularly for scenarios where the bias changes with each sampling iteration.}
\off{He established that perfect recovery of BL signals is feasible from a noise-free output.}

Thus, the decoding process utilizes $\mathcal{S}$ segments $W_i, i \in [1:\mathcal{S}]$, each characterized by a continuous interval, $W_i =[t_{a},t_{b}]$. Here, $a=\sum_{k=0}^{i-1}L_{k}+1$ and $b=\sum_{k=1}^{i}L_{k}$, where $L_i$ denotes the number of discrete sampling times within the $i-$th segment, starting with $L_0=0$. 
The duration of each window $W_i$ in time is $|W_i| = t_b - t_a$. Let $S_{W_i} = \{a,..., b\}$ denote the specific samples within each $W_i$ segment. %To identify specific samples within each segment, we employ the set $S_{W_i} = \{a,..., b\}$, which enumerates the sample indices within $W_i$. Such that, each $i$-th segment contains $L_i$ samples and 
The total number of samples in the recovery process is given by $N = \sum_iL_i$.

For each $i$-th segment with $\{b_n\}_{n\in S_{W_i}}$, to decode the signal within each segment $W_i$ from AIF-TEM output, represented by $\{t_n\}_{n\in S_{W_i}}$, we employ the operator $\mathcal{A}$, defined as
\begin{equation}\label{Operator A}
\mathcal{A}x =\sum_{n\in S_{W_i}} \int_{t_{n-1}}^{t_n} x(u)du\ g(t-\theta_n)
= \sum_{n\in S_{W_i}} P_n\ g(t-\theta_n),
\end{equation}
where $g(t) = sin(\Omega t)/\pi t$ is the sinc function, and $\theta_n = (t_{n-1} + t_n)/2$ denotes the midpoints of each pair of consecutive sampling points. This operator distinguishes itself from Lazar's methodology by implementing an adaptive bias. The coefficients $P_n, n \in  S_{W_i}$, derived from the sequences $t_n, b_n, n \in  S_{W_i}$, as given in \eqref{IF relation times ampl}.
The operator $\mathcal{A}$ effectively generates Dirac-delta pulses generated at times $\theta_n$ with corresponding weights $P_n$, and then employs a low-pass filter to smooth these pulses, facilitating the recovery of the bandlimited signal.

Focusing on samples within $S_{W_i}$ for each $i$-th segment, we define the Nyquist ratio for AIF-TEM as
\begin{equation}\label{r_a general}
r_{a_n} \triangleq \frac{\kappa \delta}{b_n - c_n} \frac{\Omega}{\pi} < 1, \quad \forall n \in  S_{W_i},
\end{equation}
and the maximum Nyquist ratio as $r_{w_i} \triangleq \max_{n \in  S_{W_i}}\{r_{a_n}\}$.

To recover the signal for $t \in W_i$, we define the sequence $x_l = x_l(t)$ via the recursion $x_{l+1} = x_l +\mathcal{A}(x-x_l)$, starting with $x_0 =\mathcal{A}x$. This recursive approach allows us to refine the signal approximation incrementally. By induction, we deduce that 
\begin{equation}\label{x_l}
 x_l = \sum_{n=0}^{l} (I-\mathcal{A})^{n}\mathcal{A}x,   
\end{equation}
where $I$ denotes the identity operator. The recovery process in \eqref{x_l} can be redefined by practical matrices formulation \cite{lazar2004time}. %To get matrix formulation for $x_l$, we define the vectors $\textbf{g} = [g(t-\theta_n)]$, $\textbf{P} = [P_n] = [-b_n(T_n) +\kappa\delta]$, matrix $\textbf{G}=[G_{kn}] = [\int_{t_{k-1}}^{t_k} g(u-\theta_n)du]$. The signal $x_l$ can be represented by $x_l = \textbf{g}\textbf{V}_l\textbf{P}$, where $\textbf{V}_l = \sum_{n=0}^{l} (I-\textbf{G})^n$.
%\textcolor{red}{In particular, in AIF-TEM, the bias $b_n$ is updated as $b_n = \hat{c}_n + \beta$, where $\beta > 0$ is a safety margin that is chosen to ensure that the condition $b_n \ge c_n$ in Definition~\ref{map_success} holds in practice. This guarantees that the operator norm in Lemma~\ref{key_lemma2} satisfies $\|I - \mathcal{A}\| < 1$, leading to a convergent reconstruction.}
\newblue{In particular, in AIF-TEM, the bias $b_n$ is updated as $b_n = \hat{c}_n + \beta$, where $\beta > 0$ is a safety margin that is chosen to help maintain the condition $b_n \ge c_n$ in Definition~\ref{map_success} in practice.\off{ When this condition holds, the operator norm in Lemma~\ref{key_lemma2} satisfies $\|I - \mathcal{A}\| < 1$, leading to a convergent reconstruction.}} \textcolor{black}{When this condition holds, Lemma~\ref{key_lemma} implies that the operator $I-\mathcal{A}$ is a contraction on the segment $W_i$, i.e., $\|(I-\mathcal{A})x\|_{W_i} \le r_{w_i}\|x\|_{W_i}$ and  $r_{w_i}<1$, which leads to a convergent reconstruction.}
%\textcolor{purple}{The effect of finite-window recovery and the resulting leakage are analyzed later in the distortion analysis in Section~\ref{Analytical Results} and in the supplementary material.}

%\newblue{In particular, in AIF-TEM the bias $b_n$ is updated \textcolor{purple}{such that $b_n = \hat{c}_n + \beta$, where $\hat{c}_n$ is the estimated amplitude and $\beta > 0$ is a safety margin. This margin compensates This margin compensates for estimation errors and is chosen so that the condition $b_n \ge c_n$ in
%Definition~\ref{map_success} is satisfied in practice,}\off{$b_n \ge c_n + \beta$, where $\beta > 0$ is a safety margin. This condition ensures the "Correct MAP Operation" in Definition~\ref{map_success},} which guarantees that the operator norm in Lemma~\ref{key_lemma2} satisfies $\|I - \mathcal{A}\| < 1$, leading to a convergent reconstruction.} \ale{Something in the sentence above is not ok: 1) There is repetition, 2) Why do we need to introduce $ c_n$ again here, and 3) the connection with the blue text after it looks weird}s

%\vspace{-0.2cm}
\subsection{Analytical Results}\label{Analytical Results}
This section delves into the performance analysis of the proposed AIF-TEM, focusing first on its oversampling characteristics, followed by an examination of reconstruction distortion as a function of the sampling rate in a practical, finite regime.  
\blue{To maintain clarity in notation, we use the overline symbol $\overline{x}$ to denote the arithmetic mean (average) of a sequence.}

\begin{definition}\label{TEM Oversampling}[TEM Oversampling]
For a $2\Omega$-BL signal, the average oversampling in TEM is defined as $OS \triangleq f_s \times \frac{\pi}{\Omega}$, where the average sampling frequency, $f_s$, is given by $f_s \triangleq 1/\blue{\overline{T_n}}$.
\end{definition}
%\ale{I am confused by the notation. $T$ is not defined in any place in the paper. You want $T_n$? Please be precise with the notation. Now it is confusing with the changes you made.}\ale{I included in red how I think it should be, here and in all the places}

%The average oversampling and sampling frequency of AIF-TEM are denoted as $OS_a$ and $f_{sa}$
\blue{With this definition in place, we now analyze the average oversampling rate of the AIF-TEM and provide a practical upper bound that reflects the influence of its adaptive parameter $b_n$.
The oversampling rate is inversely related to the average time difference between consecutive events, $\overline{T_n}_{\text{AIF}}$\off{ \textcolor{red}{for $n\in S_{W_i}$}}, which in the AIF-TEM model depends on both the input amplitude $c_n$ and the adaptive bias $b_n$. Since both $b_n$ and $c_n$ are strictly positive by design, the minimum time difference $\Delta t_{\text{min}}[n]$, defined in Equation~\eqref{Tn bound aif}, is a convex function of these variables.}

\blue{Applying Jensen’s inequality to the convex function $\Delta t_{\text{min}}[n] = \frac{\kappa \delta}{b_n + c_n}$ yields the following lower bound}

 \begin{equation}\label{mean T-AIF}
%\begin{split}
\overline{T_n}_{\text{AIF}} \geq\frac{\kappa \delta}{\overline{b_n+c_n}} = \frac{\kappa \delta}{\overline{b_n} + \overline{c_n}}.
%\end{split}
\end{equation}
\blue{Substituting this result into the oversampling rate expression from Definition~\ref{TEM Oversampling}, we obtain the following upper bound on the average oversampling rate}
\begin{equation}\label{OS upper bound AIF-TEM}
    OS_{a} \leq \frac{\overline{b}_n+\overline{c}_n}{\kappa \delta}\;\frac{\pi}{\Omega} \leq \frac{\overline{b}_n+c_{\text{max}}}{\kappa \delta}\;\frac{\pi}{\Omega} \triangleq OSU_{a}.
\end{equation}
\blue{This bound provides insight into how the adaptive parameter $b_n$ directly impacts the sampling density of the AIF-TEM. As the input amplitude decreases, both $b_n$ and $c_n$ decrease, leading to sparser sampling, which improves efficiency and reduces oversampling.}

\newblue{Reconstruction distortion refers to the discrepancy between the input signal and its reconstructed estimate obtained from the time-encoded samples. Since the time-encoding stage is modeled as ideal in this work \cite{lazar2004time,sankar2007analysis,mekel2025self}, the analyzed distortion originates from the reconstruction (decoding) process.} This distortion in AIF-TEM is analyzed here as a function of the minimum sampling rate, $f_{s_{\text{min}}} \triangleq \frac{\Omega}{\pi \max_{i}\{r_{w_i}\}}$.
\off{ We define the maximum Nyquist ratio for AIF-TEM, for any segment $W_i$ within $t_{(i-1)w+1}\leq t_n \leq t_{iw}$\off{with size $|t_{n-w}-t_n|$}, by $r_{w_i} = \max_{n\in S_{W_i}}\{r_{a_n}\}$. Hence, the minimum sampling rate is}\off{We do note that when the signal amplitude within this segment is positive, the determination of $r_{w_i}$ is refined. That is, for a positive segment, the right terms in \eqref{IF relation times ampl} are positive, i.e., $0\leq-b_nT_n + \kappa\delta$. Thus, $r_{w_i}$ is more precisely given by $\max_{n\in S_{W_i}}\left\{\frac{\kappa \delta}{b_n} \frac{\Omega}{\pi}\right\}\leq \max_{n\in S_{W_i}}\left\{\frac{\kappa \delta}{b_n-c_n} \frac{\Omega}{\pi}\right\}$.}First, we introduce the following two lemmas, which are key steps in upper-bounding this distortion. %\newblue{Before stating with the lemmas, we define a tapering function (i.e., finite-duration window) used to localize the analysis to a finite segment\footnote{\label{sec:sup_m}\newblue{In the supplementary material of this work, a discussion is provided on tapering functions in the finite regime in signal processing literature.}}.
%For each segment $W_i=[t_s^{(i)},t_e^{(i)}]$, let $\varphi_i\in C^1(\mathbb R)$ be a real-valued taper such that $0\le \varphi_i(t)\le 1$, $\varphi_i(t)=1$ for $t\in W_i$, and $\mathrm{supp}(\varphi_i)\subseteq [t_s^{(i)}-\Delta,\,t_e^{(i)}+\Delta]$ for some $\Delta>0$. The taper function enables a smooth transition between the segment $W_i$ and the surrounding time region \cite{oppenheim1999discrete,harris2005use,lazar2006real}.}
\newblue{Before stating the lemmas, we define a tapering function (i.e., a finite-duration window) used to localize the analysis to a finite segment\footnote{\label{sec:sup_m}\newblue{In the supplementary material of this work, a discussion is provided on tapering functions in the finite regime in signal processing literature.}}.
For each segment $W_i=[t_s^{(i)},t_e^{(i)}]$, let $\varphi_i\in C^1(\mathbb R)$ be a real-valued taper such that $0\le \varphi_i(t)\le 1$, $\varphi_i(t)=1$ for $t\in W_i$, and $\mathrm{supp}(\varphi_i)\subseteq [t_s^{(i)}-\Delta,\,t_e^{(i)}+\Delta]$ for some $\Delta>0$. Hence, the transition region has total width at most $2\Delta$, and the associated derivative $\varphi_i'(t)$ is supported only near the segment boundaries. The taper function enables a smooth transition between the segment $W_i$ \cite{oppenheim1999discrete,harris2005use,lazar2006real}. The leakage term introduced in Lemma~\ref{lemma:Bernstein} depends on these taper properties through both the boundary-transition contribution and the spectral-leakage contribution induced by windowing. These taper-dependent quantities are collected into $\mathcal{J}_{\varphi_i}(x)$ and subsequently propagated into the distortion bound of Theorem~\ref{Distortion AIF-TEM} via $\mathcal{L}_{\varphi_i}^{\max}(x)$.}
\blue{\begin{lemma}\label{key_lemma} Assume a $2\Omega$-BL, $c_{\text{max}}$-bounded signal with finite energy $E$, sampled using an AIF-TEM. Then, the norm of the discrepancy between $x$ and $\mathcal{A}x$ within $W_i$ is bounded by \off{$\| x-\mathcal{A}x\|_{W_i} \ \leq r_{w_i}\newblue{\|{x}\|_{2}}$.}
\newblue{
\begin{equation*}
\| x-\mathcal{A}x\|_{W_i}^2
\le
r_{w_i}^2\Big(\|x\|_{W_i}^2 + \tfrac{1}{\Omega^2}\mathcal{J}_{\varphi_i}(x)\Big).
\end{equation*}}
\newblue{where $\mathcal{J}_{\varphi_i}(u)$ 
is the leakage term given in Lemma~\ref{lemma:Bernstein} associated with the taper $\varphi_i$ on $W_i$.}
\end{lemma}}
\blue{\begin{lemma}\label{key_lemma2}
 Under the same setting as Lemma~\ref{key_lemma}, the difference between $x$ and $x_{L_i}$ for $t \in \mathbb{R}$ is given by $\| x-x_{L_i}\|_{W_i} = \| (I-\mathcal{A})^{L_i+1}x \|_{W_i}$.
\end{lemma}}

\blue{{\em Proof:} The proofs of Lemmas~\ref{key_lemma} and~\ref{key_lemma2}
are deferred to Appendix~\ref{proof:lem1} and Appendix~\ref{proof:lem2}, respectively.}

\blue{Using Lemmas~\ref{key_lemma} and~\ref{key_lemma2},} the following Theorem provides an upper bound for the \newblue{reconstruction} distortion in the proposed AIF-TEM as a function of the minimum sampling rate. \newblue{To obtain this bound, we first define the residual sequence 
$z_k \triangleq (I-\mathcal A)^k x$ and
$\mathcal{J}_i^{\max}(x)
\triangleq
\max_{0\le k\le L_i}
\mathcal{J}_{\varphi_i}(z_k)$ for $k\ge0$.}

\begin{theorem}\label{Distortion AIF-TEM}
 Consider a $2\Omega$-BL, $c_{\text{max}}$-bounded signal with finite energy $E$, sampled using an AIF-TEM. Then the \newblue{reconstruction} distortion as a function of the minimum sampling rate $f_{s_{\text{min}}}$, for maximum Nyquist ratio $r_{w_i} < 1$ in each segment $W_i$, is upper bounded by
\begin{multline}\label{sampling distortion bound}
\newblue{
%\begin{aligned}
\textcolor{black}{D_R}(f_{s_{\min}})
\le
\frac{1}{\mathcal S}\sum_{i=1}^{\mathcal S}
\Big(
r_{w_i}^{2(L_i+1)}\|x\|_{W_i}^2
+
\mathcal{L}_{\varphi_i}^{\max}(x)
\Big)} \\
\newblue{\le
\frac{1}{\mathcal S}\sum_{i=1}^{\mathcal S}
\Big(
\left(\tfrac{\Omega}{\pi f_{s_{\min}}}\right)^{2(L_i+1)}\|x\|_{W_i}^2
+
\mathcal {L}_{\varphi_i}^{\max}(x)
\Big).}
%\end{aligned}
\end{multline}
where $\mathcal{S}$ is the number of segments, $L_i$ is the number of samples in the $i$-th segment\newblue{, and the leakage term defined as
\[
\mathcal {L}_{\varphi_i}^{\max}(x)
\triangleq
\frac{r_{w_i}^2(1-r_{w_i}^{2(L_i+1)})}
{\Omega^2(1-r_{w_i}^2)}
\,\mathcal{J}_i^{\max}(x).
\]}
\end{theorem}
\begin{proof}[Proof of Theorem~\ref{Distortion AIF-TEM}]
%\sloppy
%For $k\ge0$, let $z_k=(I-\mathcal A)^k x$. 
\newblue{By Lemma~\ref{key_lemma2}, 
\[
\mathcal{E}_{s_i}^{2}
\triangleq \|x-x_{L_i}\|_{W_i}^{2}
=
\|(I-\mathcal A)^{L_i+1}x\|_{W_i}^{2}
=\|z_{L_i+1}\|_{W_i}^{2}.
\]
Applying Lemma~\ref{key_lemma} to $z_k=(I-\mathcal A)^k x$ yields
\begin{equation}\label{eq:recursion_thm1}
\|z_{k+1}\|_{W_i}^2
=
\|z_k-\mathcal A z_k\|_{W_i}^2
\le
r_{w_i}^2\Big(\|z_k\|_{W_i}^2 + \tfrac{1}{\Omega^2}\mathcal{J}_{\varphi_i}(z_k)\Big).
\end{equation}
Using $\eqref{eq:recursion_thm1}$ recursively, we obtain
{\small
\[
\|z_{k}\|_{W_i}^2
\le
r_{w_i}^{2k}\|x\|_{W_i}^2
+
\frac{r_{w_i}^2}{\Omega^2}
\sum_{m=0}^{k-1}
r_{w_i}^{2m}
\mathcal{J}_i^{\max}(x).
\]
}
Now, setting $k=L_i+1$, we have
\[
\|z_{L_i+1}\|_{W_i}^2
\le
r_{w_i}^{2(L_i+1)}\|x\|_{W_i}^2
+
\frac{r_{w_i}^2}{\Omega^2}
\sum_{m=0}^{L_i}
r_{w_i}^{2m}
\mathcal{J}_i^{\max}(x),
\]}
\newblue{and evaluating the geometric sum yields
\[
\|z_{L_i+1}\|_{W_i}^2
\le
r_{w_i}^{2(L_i+1)}\|x\|_{W_i}^2
+
\frac{r_{w_i}^2(1-r_{w_i}^{2(L_i+1)})}
{\Omega^2(1-r_{w_i}^2)}
\mathcal{J}_i^{\max}(x).
\]
Finally, averaging over all segments establishes the first bound in \eqref{sampling distortion bound}. 
The second bound follows directly since $r_{w_i}\le\max_j r_{w_j}=\Omega/(\pi f_{s_{\min}})$.}
\end{proof}

\newblue{As detailed in Lemma~\ref{lemma:Bernstein} and discussed in the supplementary materials of this work, the leakage term $\mathcal{L}_{\varphi_i}^{\max}(x)$ in Theorem~\ref{Distortion AIF-TEM} arises from the finite windowing operation and consists of two main components. The first component is directly related to the boundary transitions of the taper function $\varphi_i(t)$ and depends on the region where $\varphi_i'(t)\neq 0$. The second component captures spectral leakage introduced by tapering.} %, namely $\frac{1}{2\pi}\int_{|\omega|>\Omega}(\omega^2-\Omega^2)|G(\omega)|^2\,d\omega$, where $g=x\varphi_i$ and $G$ denotes its Fourier transform.}
\newblue{In the special case where $\varphi_i(t)\equiv 1$ on $\mathbb{R}$ (i.e., $\Delta=0$), we have $\varphi_i'(t)=0$ and the transition region is empty. Consequently, all time-domain leakage terms vanish. %Moreover, since $g=x$, the signal remains $\Omega$-bandlimited, implying that $\int_{|\omega|>\Omega}(\omega^2-\Omega^2)|G(\omega)|^2\,d\omega=0$. Hence, the total leakage contribution satisfies $\mathcal{L}_{\varphi_i}^{\max}(x)=0$. 
In particular, leakage becomes small when the plateau region of the window is large relative to the transition width $\Delta$ and when the taper varies smoothly. In this regime, the boundary region occupies only a small fraction of the segment, and the inequality in~\eqref{eq:winBern_leak} approaches the classical Bernstein bound \cite{oppenheim1999discrete,harris2005use,lazar2006real}.} %\textcolor{red}{We note that for each segment $W_i$, if the correct MAP operation condition $b_n \ge c_n + \beta$ holds for all $n \in S_{W_i}$ and $r_{w_i} < 1$, then (in the ideal case without tapering leakage, e.g., $\phi_i \equiv 1$) the reconstruction operator $\mathcal{A}$ satisfies $\|I - \mathcal{A}\| \le r_{w_i} < 1$, and therefore constitutes a contraction. This guarantees the convergence of the reconstruction over the segment.}
\newblue{In particular, the dependence on the taper regularity, transition width $\Delta$, and boundary behavior is fully absorbed into $\mathcal{J}_{\varphi_i}(x)$ and therefore into $\mathcal{L}_{\varphi_i}^{\max}(x)$ in Theorem~\ref{Distortion AIF-TEM}.}

In the corollary below, we express the distortion as a function of the maximum Nyquist ratio $r_{w_i}$ by deriving a bound on $L_i$ based on the average time interval ${\overline{T}}_{W_i}$. Formally, we have
\begin{equation}\label{samples num}
%\[ 
 L_i = \frac{|W_i|}{{\overline{T}}_{W_i}} 
\geq \frac{|W_i|}{\max_{n\in S_{W_i}} (T_n)} = \frac{|W_i|}{r_{w_i}\frac{\pi}{\Omega}}.
%\]
\end{equation}
 
\begin{cor}\label{cor:dis}
Assume the setting in Theorem~\ref{Distortion AIF-TEM}. Then, the reconstruction distortion as a function of the maximum Nyquist ratio $r_{w_i}$ in AIF-TEM is upper bounded by
%By substituting the derived bound on $L_i$ into the overall distortion presented in \eqref{distortion_formula}, we establish an upper bound on the sampling distortion as follows
 \[
 \textcolor{black}{D_R}(r_{w_i})\leq  \frac{1}{\mathcal{S}}\sum_{i=1}^\mathcal{S} \Big(r_{w_i}^{2\left(\frac{|W_i|}{r_{w_i}}\frac{\Omega}{\pi}+1\right)}\newblue{\|x\|_{W_i}^2 +\mathcal {L}_{\varphi_i}^{\max}(x)\Big).}
 \]
\end{cor}
{\em Proof:} This corollary follows by substituting the bound on $L_i$ from Equation~\eqref{samples num} into the distortion formula in Equation~\eqref{sampling distortion bound}.
\begin{cor}\label{cor: perfect recovery}
   \newblue{Assume the setting of Corollary~\ref{cor:dis} and that $r_{w_i}<1$ for all $i\in[1:\mathcal S]$.\off{and sufficiently large segments $W_i, i \in [1:\mathcal{S}]$ with size $|W_i|$,} If the leakage term $\mathcal{L}_{\varphi_i}^{\max}(x) = 0$ for all $i\in[1:\mathcal S]$, the reconstruction distortion of AIF-TEM approaches zero as $L_i\to\infty$, i.e., $D_R\rightarrow 0$, with convergence rate $\mathcal{O}\!\left(r_{\max}^{2L_{\min}}\right)$ for $r_{\max} \triangleq \max_{1\le i\le \mathcal S} r_{w_i} < 1$ and $L_{\min}\triangleq \min_{1\le i\le \mathcal S} L_i$.}
\end{cor}
\newblue{Corollary~\ref{cor: perfect recovery} characterizes the ideal no-leakage case, which is attained, for example, when $\varphi_i(t)\equiv 1$. In finite-window settings with tapering, the leakage term $\mathcal{L}_{\varphi_i}^{\max}(x)$ is generally not exactly zero. Nevertheless, if $\mathcal{L}_{\varphi_i}^{\max}(x)\to 0$ as $L_i\to\infty$, then Theorem~\ref{Distortion AIF-TEM} still implies $D_R\to 0$. In this more general case, however, the exponential rate $\mathcal{O}\!\left(r_{\max}^{2L_{\min}}\right)$ is not guaranteed unless the leakage term decays at least at the same order.}
%
%\ale{Why $\mathcal{L}_{\varphi_i}^{\max}(x) = 0$ and not $\mathcal{L}_{\varphi_i}^{\max}(x) \to 0$? Is $\mathcal{L}_{\varphi_i}^{\max}(x)$ can be equal to zero in practice in AIF-TEM? Or can it approach zero under some parameters of the problem?}
%\textcolor{purple}{If we say $\mathcal{L}_{\varphi_i}^{\max}(x) \to 0$, then we cant claim directly that the convergence rate is $\mathcal{O}\!\left(r_{\max}^{2L_{\min}}\right)$ unless the leakage term $\mathcal{L}_{\varphi_i}^{\max}(x) \to 0$ decays or converge at least the same rate, or it is exactly zero. right? However, we still get $D_R\rightarrow 0$ but not necessarily with convergence rate $\mathcal{O}\!\left(r_{\max}^{2L_{\min}}\right)$. In practice not necessarily we get $\mathcal{L}_{\varphi_i}^{\max}(x) = 0$ because we work with finite segments.}
\begin{proof}
\newblue{Under the stated condition $\mathcal{L}_{\varphi_i}^{\max}(x)= 0$\off{ (see the discussion after Theorem~\ref{Distortion AIF-TEM})},
the bound in \eqref{sampling distortion bound} reduces to
\begin{equation*}%\label{eq:dr_b_c}
D_R \le \frac{1}{\mathcal S}\sum_{i=1}^{\mathcal S} r_{w_i}^{2(L_i+1)}\|x\|_{W_i}^2.
\end{equation*}
Let $r_{\max} \triangleq \max_{1\le i\le \mathcal S} r_{w_i} < 1$.
Then, for all $i$, $r_{w_i}^{2(L_i+1)} \le r_{\max}^{2(L_i+1)}$. 
Hence, we obtain
\[
D_R \le \frac{1}{\mathcal S}\sum_{i=1}^{\mathcal S} r_{\max}^{2(L_i+1)}\|x\|_{W_i}^2
\le r_{\max}^{2(L_{\min}+1)}\cdot \frac{1}{\mathcal S}\sum_{i=1}^{\mathcal S}\|x\|_{W_i}^2.
\]}
\newblue{Since $x$ has finite energy and the segments $\{W_i\}_{i=1}^{\mathcal S}$ are disjoint, we have $\frac{1}{\mathcal S}\sum_{i=1}^{\mathcal S}\|x\|_{W_i}^2 \le \|x\|^2 < \infty$. Then, because $r_{\max}<1$, we have $r_{\max}^{2(L_{\min}+1)}\to 0$ as $L_{\min}\to\infty$, which implies $D_R\to 0$. Moreover, the above inequality yields the exponential convergence rate $D_R=\mathcal{O}\!\left(r_{\max}^{2L_{\min}}\right)$.}
\end{proof}
\newblue{The corollary is stated in terms of $L_i$, the number of samples within segment $W_i$. Based on~\eqref{samples num}, we have $L_i ={|W_i|}/{\overline{T}_{W_i}}$,
where $\overline{T}_{W_i}$ denotes the average time interval within the segment. Applying Definition~\ref{TEM Oversampling}, this can be expressed as
$L_i = f_{s,i}\,|W_i|$,
where $f_{s,i}$ denotes the average sampling rate in segment $i$. Therefore, achieving a target $L_i$ requires selecting the segment duration $|W_i|$ sufficiently large for the expected sampling rate $f_{s,i}$.}
\newblue{Furthermore, since the distortion decays exponentially as $O(r_{w_i}^{2L_i})$, a target distortion level $\varepsilon$ can be achieved by enforcing
$r_{w_i}^{2L_i} \le \varepsilon$, which yields the condition $L_i \ge{\log(1/\varepsilon)}/{(2|\log r_{w_i}|)}$. Thus, the required number of samples per segment grows only logarithmically with respect to the desired reconstruction accuracy.}

\begin{figure}
%\vspace{-0.2cm}
\includegraphics[width=0.48\textwidth]{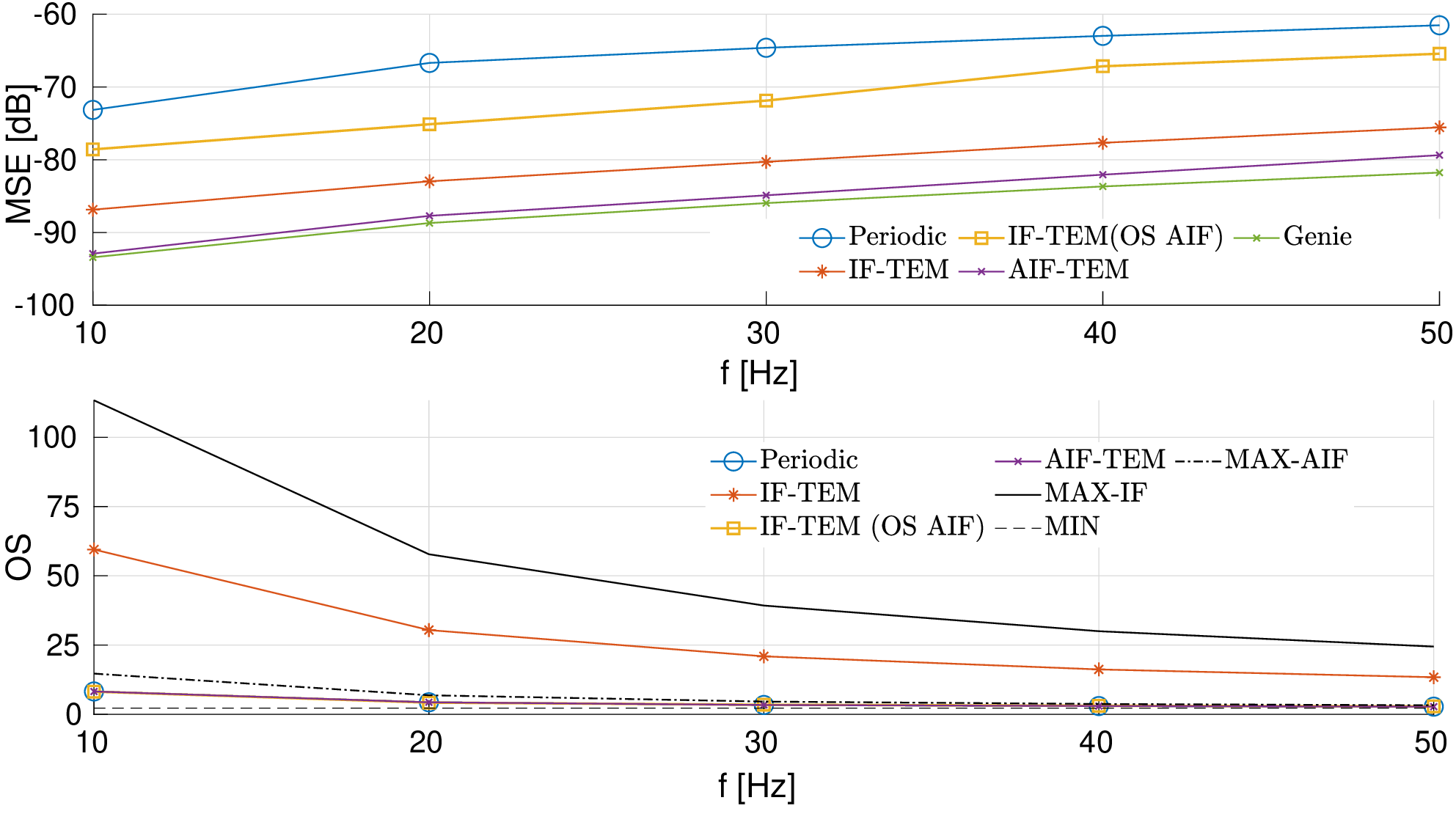}
\vspace{-0.3cm}
%\captionsetup{width=0.5\textwidth}
\caption{\small Performance comparison of Periodic sampler, IF-TEM, AIF-TEM, and 'Genie' sampler in terms of (a) MSE and (b) Oversampling factor (OS).}\label{MSE fig}\vspace{-0.5cm}
\end{figure}

%\vspace{-0.2cm}
\subsection{Evaluation Results for sampling process}\label{EVALUATION RESULTS}
%\vspace{-0.1cm}
This section provides a numerical evaluation of the proposed AIF-TEM, classical IF-TEM, and periodic sampling. The evaluation is conducted using both synthetic signals generated in MATLAB and a real audio signal.  \blue{For amplitude value estimation, the MAP block model is employed as proposed in Section~\ref{MAP Modes}. The weighting factors are configured as $\alpha_1=0.98$ and $\alpha_2=0.17$.} We begin by highlighting the advantages of the proposed AIF-TEM using a $2\Omega$-BL, $c_{\text{max}}$-bounded signal with finite energy $E$. Here, $c_{\text{max}} = \sqrt{(E\Omega)/\pi}$ \cite{papoulis1967limits}, and $\Omega$ varies between $2\pi \cdot (10{-}50)$ Hz. The input signal is defined as
\begin{equation}\label{input signal}
   x(t) =  \sum_{n=-M}^{M}a[n]\frac{\sin\left(\Omega \left(t-n\frac{\pi}{\Omega}\right)\right)}{\Omega \left(t-n\frac{\pi}{\Omega}\right)},
\end{equation}
where $M=2$, $t \in \mathbb{R}$, and coefficients $a[n]$ are randomly selected 100 times from [-1,1]. The signal energy $E$ is varied within the range [0.25,2.5]. For the classical IF-TEM, the bias $b_{\text{IF}}$ is set to achieve a Nyquist ratio of $r_c = 0.45$. In the case of AIF-TEM, the safety margin parameter $\beta$ is chosen such that $r_{a_n} \leq 0.45$ for all $n$, with a window size $w = 1$.

\begin{figure}
    \vspace{-0.0cm}
    \begin{subfigure}{0.155\textwidth}
        \includegraphics[width=\textwidth]{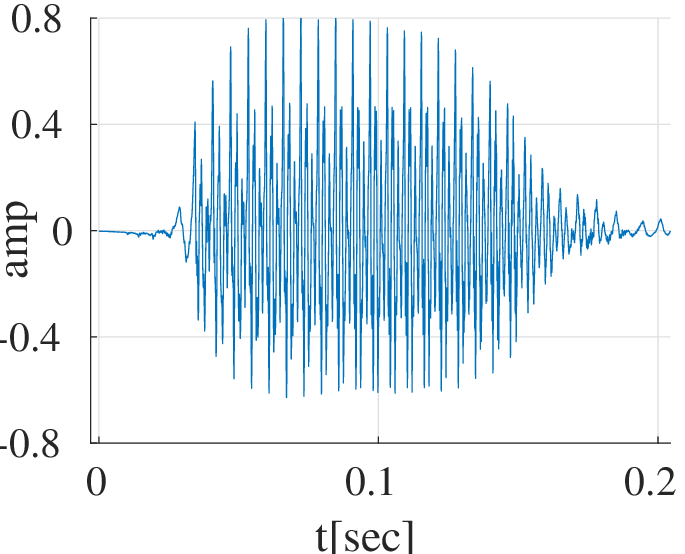}
        \caption{\vspace{-0.1cm}Segment A}
    \end{subfigure}
    \begin{subfigure}{0.155\textwidth}
        \includegraphics[width=\textwidth]{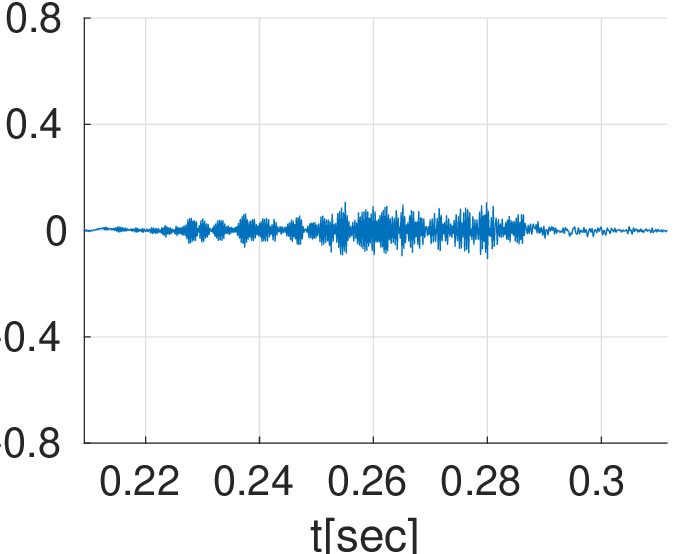}
        \caption{\vspace{-0.1cm}Segment B}
    \end{subfigure}
    \begin{subfigure}{0.155\textwidth}
        \includegraphics[width=\textwidth]{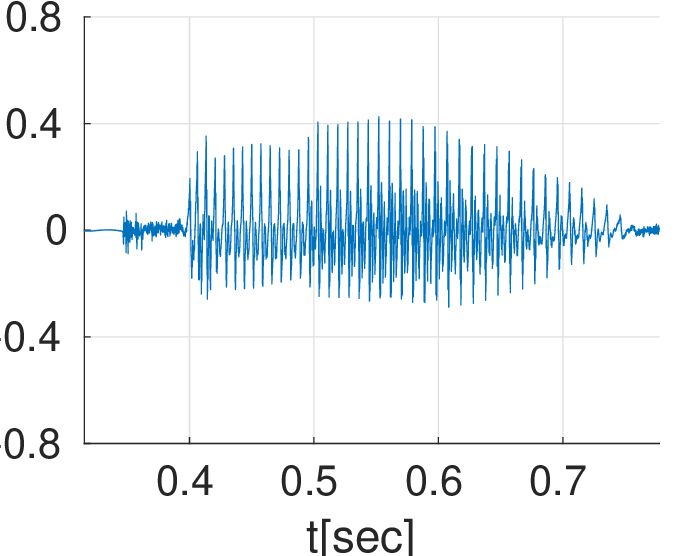} 
        \caption{\vspace{-0.1cm}Segment C}
    \end{subfigure}
        \begin{center}
            \footnotesize % Reduce the font size
            \begin{tabular}{|l||l|l||l|l||l|l|}
            \hline
            & \multicolumn{2}{c||}{\bf Segment A} & \multicolumn{2}{c||}{\bf Segment B} & \multicolumn{2}{c|}{\bf Segment C} \\
            \hline\hline
            \textbf{Sampler}        & MSE & OS & MSE & OS & MSE & OS \\ \hline
            \textbf{Periodic}       & -54.77      & 4.79 & -52.61 &  4.79  & -48.95      & 4.79\\ \hline
            \textbf{IF-TEM}         & -68.52      & 12.99 & -64.33      & 13.02 & -68.61      & 12.99\\ \hline
            \textbf{AIF-TEM}        & -71.67      & 5.75  & -74.71      & 4.08 & -76.13      & 4.54\\ \hline
            \textbf{IF-TEM2} & -68.72      & 4.79 & -53.5       & 4.8 & -65.28      & 4.79 \\ \hline
            \end{tabular}
        \end{center}
    \vspace{-0.3cm}
    \caption{\small\label{Audio Sig}Comparative performance\off{ of Periodic, IF-TEM, and AIF-TEM samplers with bandlimited} with audio signal. Each subfigure presents a segment of an audio signal, where under the subfigures, the table represents the MSE in dB of reconstruction for each sampler and the average oversampling factor (OS).}\vspace{-0.3cm}
\end{figure}
\off{
\begin{figure}
    \centering
    \includegraphics[width=1\linewidth]{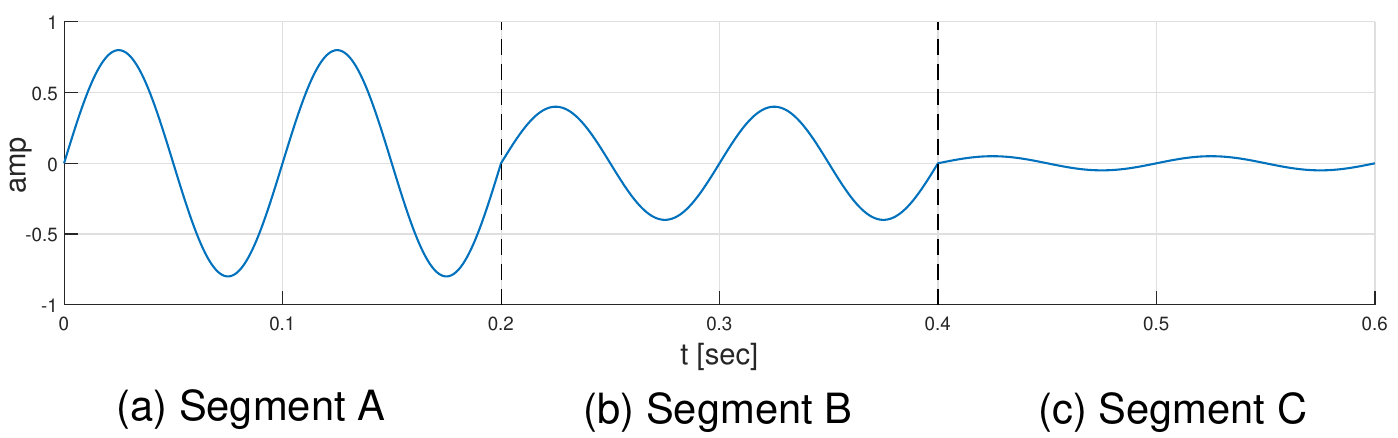}
    \vspace{-0.4cm}
    \begin{center}  
        \footnotesize % Reduce the font size
        \begin{tabular}{|l||l|l|l||l|l|l||l|l|l|}
        \hline
        & \multicolumn{3}{c||}{\bf Segment A} & \multicolumn{3}{c||}{\bf Segment B} & \multicolumn{3}{c|}{\bf Segment C} \\
        \hline\hline
        \textbf{Sam.}   & NM    & B   & OS  & NM    & B   & OS & NM   & B   & OS   \\ \hline
        \textbf{IF.}      & -81  & -19 & 9  & -73  & -19 & 9.2 & -62 & -19 & 9.5 \\ \hline
        \textbf{AIF.}     & -101 & -84 & 13 & -99  & -69 & 9.2 & -73 & -46 & 5.5 \\ \hline
        \end{tabular}
    \end{center}
    \vspace{-0.3cm}
    \caption{\small\label{fig:Sin Sig}Comparative performance with bandlimited signal. The table represents the NMSE in dB (NM) of recovery for each sampler, the NMSE bound in dB (B), and the average oversampling factor (OS) for each segment of the signal.}\vspace{-0.5cm}
\end{figure}
}

\begin{figure}
    \centering
    \includegraphics[width=1\linewidth]{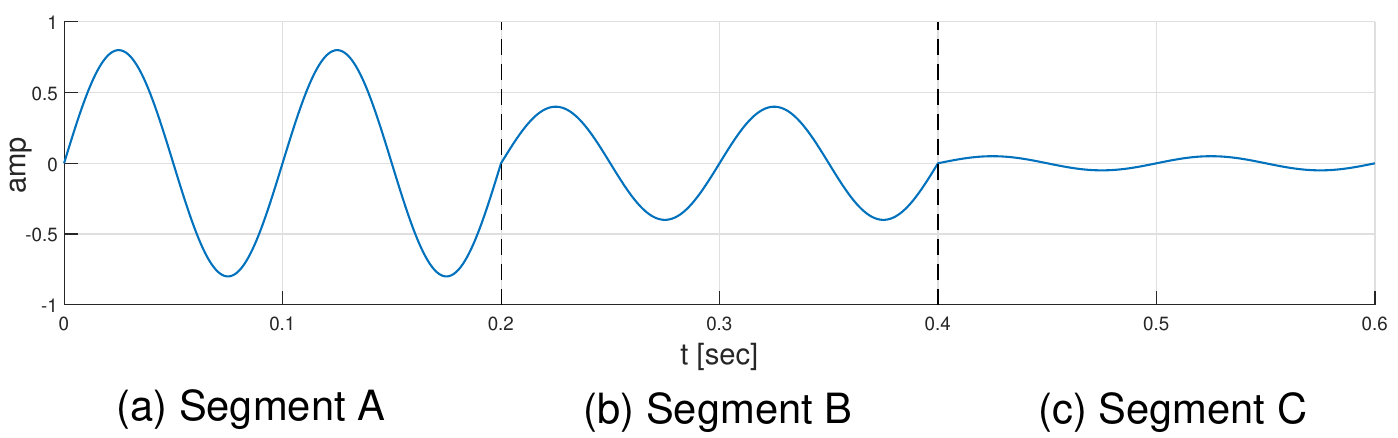}
    \vspace{-0.4cm}
    \begin{center}  
        \footnotesize
        \begin{tabular}{|l||l|l||l|l||l|l|}
        \hline
        & \multicolumn{2}{c||}{\bf Segment A} 
        & \multicolumn{2}{c||}{\bf Segment B} 
        & \multicolumn{2}{c|}{\bf Segment C} \\
        \hline\hline
        \textbf{Sampler} & MSE & OS & MSE & OS & MSE & OS \\ \hline
        \textbf{IF-TEM}  & -86  & 9   & -84  & 9.2 & -91  & 9.5 \\ \hline
        \textbf{AIF-TEM} & -106 & 13  & -110  & 9.2 & -102  & 5.5 \\ \hline
        \end{tabular}
    \end{center}
    \vspace{-0.3cm}
    \caption{\small\label{fig:Sin Sig}
    Comparative performance with a bandlimited signal. 
    The table reports the MSE in dB of recovery for each sampler and 
    the average oversampling factor (OS) for each segment.}
    \vspace{-0.5cm}
\end{figure}
Fig.~\ref{MSE fig}(a) presents the MSE results as defined in \eqref{MSE}. The performance of AIF-TEM is indicated by the purple line, the red line denotes IF-TEM, both are configured with parameters $\kappa = 0.5$ and $\delta = 0.02$. The blue line demonstrates the MSE for periodic sampling, matched to the average sampling frequency of AIF-TEM. The yellow line represents IF-TEM operating with a larger $\delta$, resulting in an average oversampling rate similar to AIF-TEM, which violates the perfect reconstruction condition in \eqref{recovery}. The green line demonstrates the MSE for a hypothetical scenario with a "Genie" that informs us of the temporal amplitude, thereby determining the accurate temporal bias. In Fig.~\ref{MSE fig}(b), the average oversampling for the compared samplers is shown. The black dash-dotted line shows the AIF-TEM maximum oversampling bound from \eqref{OS upper bound AIF-TEM}. The dark line shows the IF-TEM maximum oversampling, computed by $\frac{\pi/\Omega}{\Delta t_{c_{\text{min}}}}$. In contrast, the dashed line points to the minimum oversampling, based on $1/r_c = 2.22$. Importantly, with similar oversampling, AIF-TEM achieves an MSE reduction of at least 12 dB compared to both IF-TEM and the periodic sampler. Moreover, with the Genie sampler, we observe nearly comparable oversampling and MSE results to those obtained using the estimation block.

Fig.~\ref{Audio Sig} contrasts the performances of the periodic, IF-TEM, and AIF-TEM samplers on a BL audio signal. It's evident that IF-TEM consistently displays a higher oversampling factor, but its MSE performance lags behind AIF-TEM. Notably, AIF-TEM dynamically modulates its oversampling in response to the signal amplitude, resulting in increased oversampling for larger amplitudes. In contrast, the oversampling in IF-TEM remains relatively fixed. When the periodic sampler's rate is aligned with AIF-TEM's oversampling rate, it yields a higher MSE. For IF-TEM2, adjustments to $\delta$ ensure that the average oversampling in IF-TEM corresponds to that of AIF-TEM.
 %The results in Fig.~\ref{Audio Sig} show that AIF-TEM outperforms both IF-TEM and periodic samplers.

% \textcolor{purple}{Since the table~\ref{fig:Sin Sig} have Normalized MSE (Not MSE value) I converted the numbers to MSE, also I deleted the Bound column, since the previous bound from previous versions is not valid}
Fig.~\ref{fig:Sin Sig} contrasts the performances of IF-TEM and AIF-TEM on a $2\Omega$-BL signal, $x(t)$ (with $\Omega = 2\pi10$) that is divided into 3 segments, each with distinct maximum amplitude, $x(t) = a\sin(\Omega t)$, where $a=[0.8, 0.4, 0.05]$ for each segment. The table presents the MSE\off{, the NMSE bound as per \eqref{distortion_formula} (for IF-TEM with $r_{w_i}=r_c$),} and the oversampling for each segment.
\off{ The average oversampling is assigned to be equal for every sampler over the signal.}

\off{While both AIF-TEM and IF-TEM ensure perfect recovery, the former can achieve this with less oversampling. Simulations further highlight AIF-TEM's advantages over IF-TEM and traditional periodic samplers using synthetic and audio signals.}
%% Or you use manual references (pay attention to consistency and the
%% formatting style!):

Figs.~\ref{fig:myfig1} and \ref{fig:myfig2} contrast the performances of the IF-TEM and AIF-TEM samplers over time.
The input signals are defined by equation \eqref{input signal}, characterized by $M=5$, $\Omega = 2\pi20$ Hz, and coefficients $a[n]$ randomly selected from [-1,1]. Given the signal's maximum energy $E =0.0869$, we calculate the maximum amplitude as $c_{\text{max}} = \sqrt{(E\Omega)/\pi} = 3.16$.
For the AIF-TEM configuration, the parameters are set to $\Delta =0.01,\kappa = 0.24$, $\beta = 0.1$ and the MAP block utilizes $w=5$. For the IF-TEM, the setup involves $\kappa = 0.24$ with a fixed bias $b_{\text{IF}} = c_{\text{max}}+\beta$. To ensure a comparable average oversampling rate to that of AIF-TEM, the threshold $\delta_c$ for IF-TEM is adjusted to $0.0433$ for the signal presented in Fig.~\ref{fig:myfig1}(a) and to $0.0262$ for the signal presented in Fig.~\ref{fig:myfig2}(a).

Sub-figures (a) shows the input signals with blue lines, alongside the output times, $t_n$, of IF-TEM and AIF-TEM, illustrated by Dirac pulses with the same amplitude as the sampled input signal, in red and purple lines, respectively. 
Sub-figures (b) illustrates the Nyquist ratio, $r$, with the purple line indicating the ratio for AIF-TEM as per \eqref{r_a general} and the red line showing the actual ratio for each time sample of IF-TEM with the fixed bias, i.e., for $b_n = b_{\text{IF}}$ (as obtained in practice by the classical scheme in Section~\ref{IF-TEM vs. Periodic Sampler}). For the ideal case with a genie (such that $b_n = \beta +c_n$), the target Nyquist ratio $r_a = \kappa \delta/\beta$, is depicted by a blue dotted line. The yellow line represents the Nyquist ratio bound, $r_c$, for IF-TEM as given in \eqref{recovery}. 
Sub-figures (c) compares the oversampling factor of IF-TEM and AIF-TEM, represented by red and purple lines, respectively. These factors are given by $OS_n \triangleq \frac{1}{T_n} \times \frac{\pi}{\Omega}$, where the blue dotted line indicates the mean oversampling as in Definition~\ref{TEM Oversampling}. The oversampling results, presented for both samplers, are for $T_n = t_{n+1} - t_n$, using the output times illustrated in Sub-figures (a).
Sub-figures (d) presents the reconstruction error between the input signals and recovered signals from IF-TEM and AIF-TEM outputs, showcased by red and purple lines, respectively. The error metric shown is $\abs{x(t) -\hat{x}(t)}$, where $\hat{x}(t)$ denotes the recovered signal. In the scenario presented in Fig.~\ref{fig:myfig1}, the MSE values in dB are $-103.78$ for AIF-TEM and $-70.4$ for IF-TEM. Additionally, for sampling using a periodic sampler, with a sampling frequency and oversampling rate equivalent to the average determined by AIF-TEM, the MSE value in dB is $-48.27$. 
\off{ Adjusting IF-TEM's threshold yields $r_c =  4.16 > 1$, indicating that the criterion for perfect recovery is not met.} Similarly, for Fig.~\ref{fig:myfig2}, the MSE values are $-109.14$ for AIF-TEM, $-79.78$ for IF-TEM, and $-69.75$ for periodic sampler.\off{ The increased threshold for IF-TEM results in $r_c = 1.25 > 1$, that is, the condition for perfect recovery doesn't hold.}

\begin{figure}
  \centering
  \includegraphics[width=\linewidth]{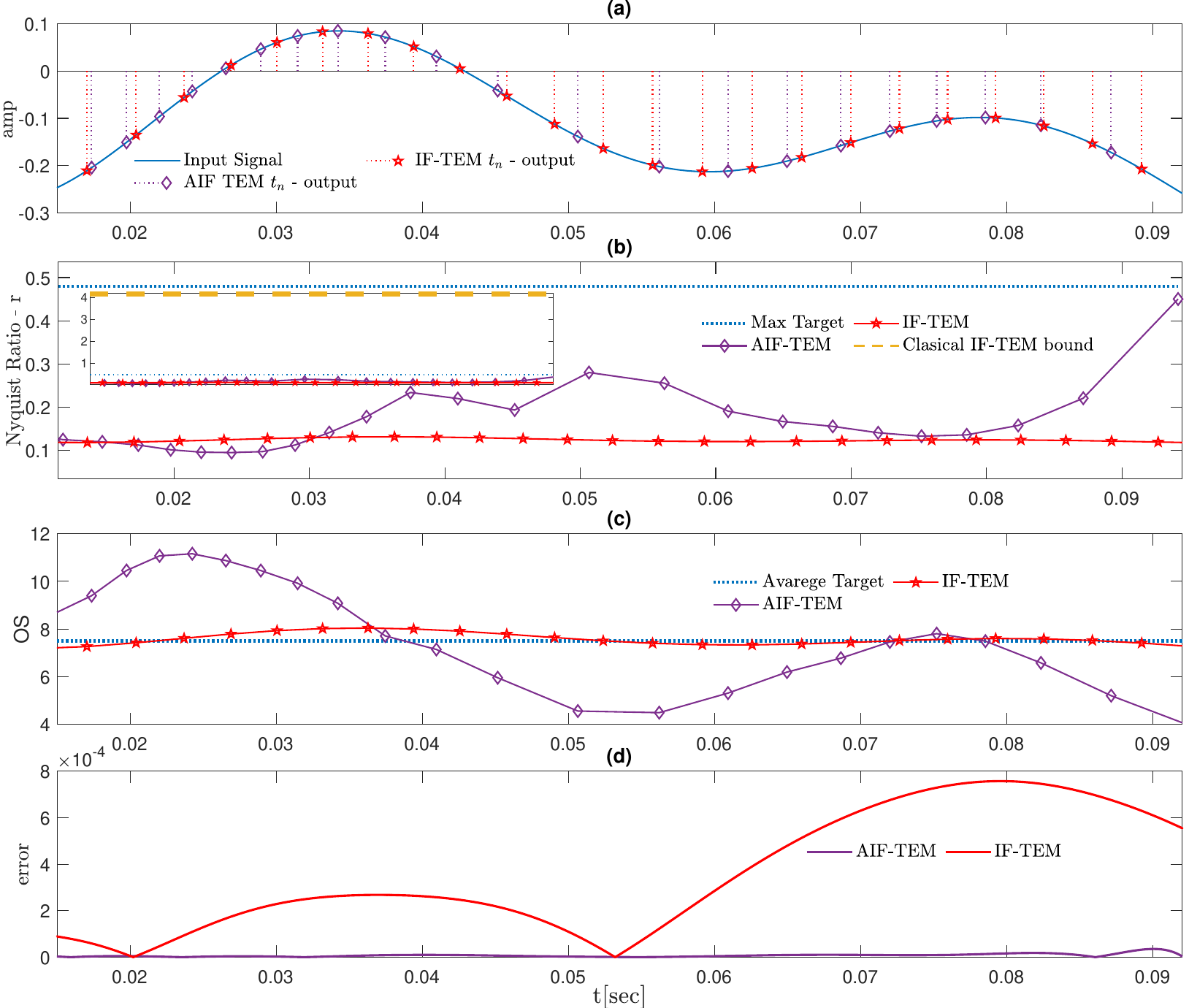}
 \caption{\small\off{Comparative performance of AIF-TEM versus IF-TEM samplers in time.}\newblue{Time-domain comparison of sampling density and event distribution for AIF-TEM and IF-TEM under varying signal conditions for signal realization 1. }}
  \label{fig:myfig1}
  \vspace{-0.5cm}
\end{figure}

\begin{figure}
  \centering
  \includegraphics[width=\linewidth]{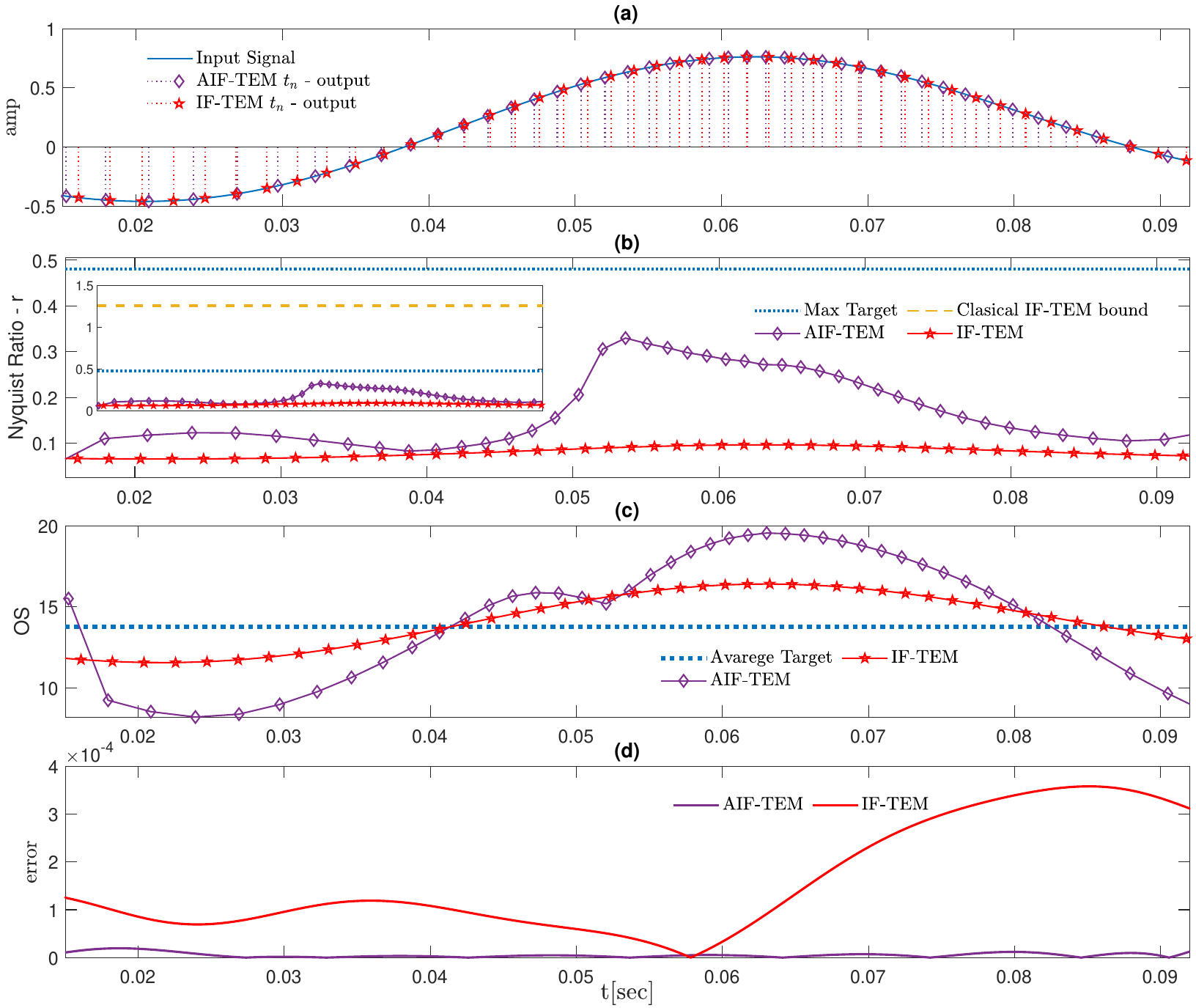}
  \caption{\small\newblue{Time-domain comparison of sampling density and event distribution for AIF-TEM and IF-TEM under varying signal conditions for signal realization 2.}}
  \label{fig:myfig2}
 \vspace{-0.5cm}
\end{figure}

We do note, that sub-figures (b) reveals that the Nyquist ratio for IF-TEM at each time sample significantly deviates from the classical IF-TEM bound (given in \eqref{recovery}). This deviation is attributed to the fixed bias, which exceeds the maximum amplitude value and lacks adaptation to temporal amplitude variations. Conversely, AIF-TEM dynamically adjusts its bias in response to temporal amplitude values, $c_n$, and the standard deviation of the preceding values, following the proposed adaptive scheme with the MAP block. This adjustment ensures that the Nyquist ratio is adapted to amplitude changes. Furthermore, Sub-figures (c) demonstrate that AIF-TEM's oversampling, at each time sample, is responsive to amplitude fluctuations compared to IF-TEM. For IF-TEM, the oversampling rates for each time sample remain near to the average oversampling, showing less adaptability. 

In summary, it is important to note that the observations from Figs.~\ref{fig:myfig1} and~\ref{fig:myfig2} illustrate that IF-TEM's \newblue{sampling behavior does not adapt to variations in the signal amplitude and frequency, resulting in a constant Nyquist ratio and oversampling rate.} In contrast, the proposed AIF-TEM's adaptive bias mechanism, by MAP block, allows for adjustments in the Nyquist ratio and oversampling rate in response to signal variations, showcasing its adaptability and results with significantly lower error in the scenarios evaluated.
%
%\ale{Aseel, please add below exactly the type and parameters you used for the tapering window.}
%
%\textcolor{black}{Moreover, we extend the evaluation presented in this section to support overlapping between encoded segments as suggested in the literature for real-time applications, by using tapering functions, i.e., finite-duration window \cite{oppenheim1999discrete,harris2005use} (see Section~\ref{Analytical Results} and the supplementary materials of this work). We employ the same tapering window and overlap-add strategy as in \cite{lazar2006real}. Specifically, each segment is multiplied by a cosine taper with plateau on $[t_s^{(i)}+\Delta,\,t_e^{(i)}-\Delta]$ with smooth transitions of total width $2\Delta$ at the segment boundaries. In our simulations, the MSE difference between non-tapered reconstruction and tapered overlap-add reconstruction remained within approximately 3--4 dB in all the samplers tested, confirming the gains of AIF-TEM compared to existing tested approaches.}
\textcolor{black}{Finally, we extend the evaluation presented in this section to support overlapping between encoded segments, as commonly adopted in real-time applications, by incorporating tapering functions, i.e., finite-duration windows \cite{oppenheim1999discrete,harris2005use} (see Section~\ref{Analytical Results} and the supplementary materials for the analytical justification).
\textcolor{black}{In the tests, the signal is divided into segments of length $0.2\,\mathrm{s}$ with overlap width $\Delta = 0.02\,\mathrm{s}$ at each boundary. Specifically, each segment $W_i=[t_s^{(i)},t_e^{(i)}]$ is multiplied by a cosine taper whose plateau equals one on $[t_s^{(i)}+\Delta,\,t_e^{(i)}-\Delta]$, with smooth transition regions of total width $2\Delta$ at the segment boundaries. This overlap-add strategy follows the same framework used in \cite{lazar2006real}.}
%We employ the same tapering window and overlap-add strategy as in \cite{lazar2006real} for time-encoding machines. Specifically, each segment $W_i=[t_s^{(i)},t_e^{(i)}]$ is multiplied by a cosine taper whose plateau equals one on $[t_s^{(i)}+\Delta,\,t_e^{(i)}-\Delta]$, with smooth transition regions of total width $2\Delta$ at the segment boundaries.
In our simulations, the MSE difference between the non-tapered reconstruction and the tapered overlap-add reconstruction remained within approximately 3--4~dB across all samplers tested. Importantly, the relative performance gains of AIF-TEM over IF-TEM were preserved under tapering, confirming that the adaptive advantages remain robust in finite-duration, overlap-add implementations.}
%\ale{Aseel, I updated the text above in red. Please check if this is fine and include all the information needed to reproduce the test you did with the tapering window.}\textcolor{purple}{I did not do the test for periodic sampling, I don't think we should do the tapering method for periodic sampling, because our analytical results are about IF-TEM/ AIF-TEM, But I agree that the comparison, AIF-TEM over IF-TEM and periodic sampling were preserved under tapering}

%\newblue{Although tapering is introduced for the analytical results (see Section~\ref{Analytical Results}), its effect on the practical reconstruction distortion is negligible. To verify this, we employ the same tapering window and overlap-add strategy as in \cite{lazar2006real}. Specifically, each segment is multiplied by a cosine taper with plateau on $[t_s^{(i)}+\Delta,\,t_e^{(i)}-\Delta]$ with smooth transitions of total width $2\Delta$ at the segment boundaries. In our simulations, the MSE difference between non-tapered reconstruction and tapered overlap-add reconstruction remained within approximately 3--4 dB, confirming that the previously reported performance results remain valid.}

In the following sections, we delve into the analysis of quantization for AIF-TEM.

\section{Quantization for AIF-TEM}\label{Quantization for AIF-TEM}

Following the sampling stage using AIF-TEM, we proceed to quantize the time differences, $T_n= t_{n+1}-t_n$. The quantized time differences are denoted by $\tilde{T}_{n} = \tilde{t}_{n+1} - \tilde{t}_n$.  
\off{While the biases $b_n$ are utilized in the signal recovery process, they are not subject to quantization.}
\newblue{While the biases $b_n$ are utilized in the signal recovery process, they are selected from a predefined discrete set with resolution $\Delta_b$ over the range $[b_{\min}, b_{\max}]$ and are not further quantized beyond this resolution. The MAP block updates $b_n$ according to the amplitude estimate $c_n$, as described in Section~\ref{MAP Modes}.}\off{ We assume that the MAP block is implemented such that it selects biases in a segmented fashion with a step size of $\Delta_b$ Given a uniform quantizer, the number of bits $R_b$ allocated for the biases is pre-determined during the sampling process and corresponds to $K_b = 2^{R_b} = \frac{b_{\text{max}} - b_{\text{min}}}{\Delta_b}$.} 
\newblue{We represent each bias value by its index on the known grid, $k_n \triangleq \frac{b_n - b_{\min}}{\Delta_b} \in \{0,1,\ldots,K_b-1\}$, for $K_b \triangleq \left\lceil \frac{b_{\max}-b_{\min}}{\Delta_b} \right\rceil$. Accordingly, the bias information can be conveyed to the decoder by transmitting the index sequence $\{k_n\}$ aligned with the event stream, which requires $R_b = \lceil \log_2 K_b \rceil$ bits per bias index. Since $b_{\min}$, $b_{\max}$, and $\Delta_b$ are shared between encoder and decoder, each received index uniquely determines the corresponding bias value via $b_n=b_{\min}+k_n\Delta_b$, thereby ensuring synchronization during reconstruction.} %\ale{The last sentence is true only for the classic approach, not for the dynamic, right? If so, it is needed to update this part}

This section delineates two quantization schemes for AIF-TEM: 1) a dynamic uniform quantizer that adjusts the step size $\Delta_{W_i}$, per segment $W_i$, and 2) a classic uniform quantizer with a fixed step size $\Delta$. These schemes are detailed in Subsections~\ref{Dynamic Quantization} and \ref{Classic Quantization} and illustrated in Fig.~\ref{fig:TEM with Quant}.

\begin{figure}[h!]
\centering
\includegraphics[width=0.48\textwidth]{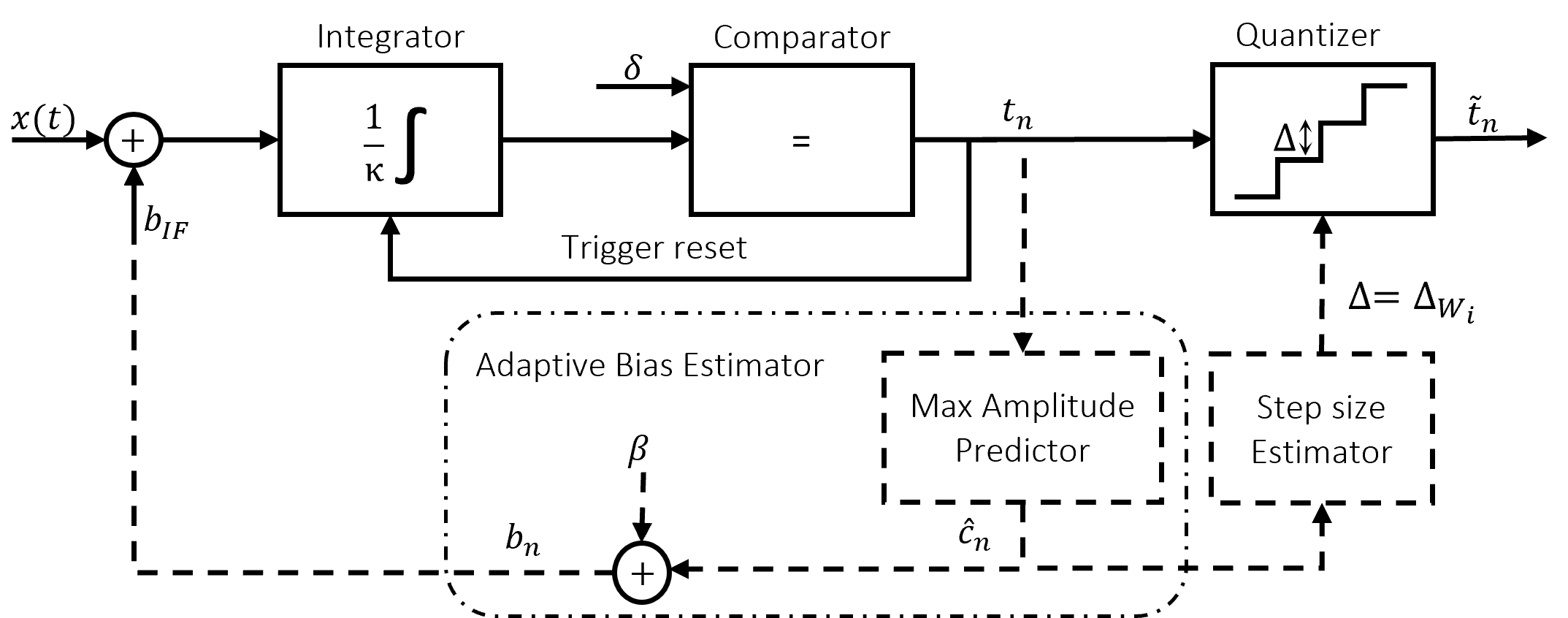}
\caption{\small\label{fig:TEM with Quant} IF-TEM (solid lines) and its adaptive design, AIF-TEM (solid with dashed lines), with classic and dynamic quantization schemes.}
\vspace{-0.3cm}
\end{figure}

In the following subsection, we elaborate on dynamic quantization.
\subsection{Dynamic Quantization}\label{Dynamic Quantization}

The dynamic approach uses a $K$-level uniform quantizer with dynamically adjusted quantization step sizes for each individual segment $i$. The step size varies for every segment $W_i$, based on the maximum bias and maximum amplitude, determined as $b_{\text{max}}^{(i)} = \max\{b_n \mid n \in S_{W_i}\}$ and $c_{\text{max}}^{(i)} = \max\{c_n \mid n \in S_{W_i}\}$, respectively. 
\newblue{The decoder identifies the segment boundaries according to the shared deterministic segmentation rule described in Section~\ref{decoding process}. Since the quantization step size $\Delta_{W_i}$ is computed at the encoder based on segment-specific statistics, its value is signaled once per segment as side information using a fixed predefined representation. This ensures that both encoder and decoder apply identical quantization parameters within each segment.}
The objective of dynamic quantization for AIF-TEM is to reduce the step size, thereby minimizing the overall distortion in the recovered signal. This approach is particularly beneficial for signals with amplitude and frequency variations.
The following lemma introduces the quantization step size for each segment.

\begin{lemma}\label{lemma: dynamic step size}
Consider an AIF-TEM with a correctly operating MAP block, followed by a $K$-level dynamic uniform quantizer. For a $c_{\text{max}}$-bounded signal, the quantization step size $\Delta_{W_i}$ in each segment $i$ is given by
     \begin{equation}\label{dynamic step size}
           \Delta_{W_i} = \kappa \delta \frac{b_{\text{max}}^{(i)} +c_{\text{max}}^{(i)}-\beta}{\beta(b_{\text{max}}^{(i)} +c_{\text{max}}^{(i)})K}.
     \end{equation}

\end{lemma}

\begin{proof}
 Using \eqref{Tn bound aif}, the lower bound for time differences in the $i$-th segment, $T_n, \forall n \in S_{W_i}$, is given by
       \begin{equation}\label{min Tn}
        \Delta t_{\text{min}}[n] = \frac{\kappa \delta}{b_n+c_n} \geq \frac{\kappa \delta}{b_{\text{max}}^{(i)}+c_{\text{max}}^{(i)}} = T_{\text{min}}^{(i)}, 
    \end{equation}
    when the inequality follows from the definition of $b_{\text{max}}^{(i)}$ and $c_{\text{max}}^{(i)}$.
    The upper bound for time differences in the $i$-th segment is given by 
    \begin{align}
         \Delta t_{\text{max}}[n] & = \frac{\kappa \delta}{b_n-c_n} \leq  \frac{\kappa \delta}{\min_n\{b_n -c_n\}} \label{max Tn} \\
         & \leq \frac{\kappa \delta}{\beta} = T_{\text{max}}^{(i)}, \notag 
    \end{align}
     where the first inequality holds directly, and the last inequality follows from Definition~\ref{map_success}, i.e., $b_n \geq c_n $. 
     Thus, in the segment $i$, the dynamic range of $T_n, \forall n \in S_{W_i}$ is given by
     \begin{equation}\label{dynamic range of Tn}
         T_{\text{max}}^{(i)} - T_{\text{min}}^{(i)} = \kappa \delta \frac{b_{\text{max}}^{(i)}+c_{\text{max}}^{(i)}-\beta}{\beta(b_{\text{max}}^{(i)}+c_{\text{max}}^{(i)})},
     \end{equation}
     and the quantization step is $\
     \Delta_{W_i} =  (T_{\text{max}}^{(i)} - T_{\text{min}}^{(i)}) /K$.
\end{proof}
\vspace{-0.1cm}
While the MAP block in practical scenarios may not estimate the amplitude perfectly, in the numerical evaluation we use the estimated maximum amplitude for the $i$-th segment, $\hat{c}_{\text{max}}^{(i)} = \max\{\hat{c}_n\}, n \in S_{W_i}$.

\blue{We note that in the suggested dynamic quantization approach herein, each segment $i$ includes side information containing either the segment’s maximum amplitude and bias or the computed step size $\Delta_{W_i}$. This side information informs the decoder how to interpret the quantized levels for that segment. Although $\Delta_{W_i}$ varies across segments, the total number of quantization levels $K$ (and thus $\log_2(K)$ bits per sample) can remain the same for all segments. Since $\Delta_{W_i}$ depends inversely on $c_{\text{max}}^{(i)}$, as seen in \eqref{dynamic step size}, segments with lower amplitudes result in smaller step sizes, thereby reducing quantization distortion in those regions. As shown in the simulation results in Section~\ref{Evaluation results for Quantization Process}, for three signal segments (and thus four distinct step sizes), the amount of side information is small, yet it yields a significantly smaller reconstruction error. Hardware deployment of this approach would require dedicated protocols to transmit or store side information, which we leave for future work.} %\ale{In the simulation section, you show that this side information is small, so we can gain a lot, right? If this is true, maybe we should also say that here and refer to the corresponding part in the simulation section.}

\off{Since decreasing the quantization step can decrease the quantization distortion, the following corollary shows that the approach of dynamic quantization is beneficial for signals with amplitude variations.

\begin{cor}\label{cor: cmaxi, step size relation}
    Assume the setting in Lemma~\ref{lemma: dynamic step size} with $\beta$ fixed. The quantization step size $\Delta_{W_i}$ decreases as the maximum amplitude for window $W_i$, $c_{\text{max}}^{(i)}$, decreases.
\end{cor}

\begin{proof}
 Since $\Delta = \frac{\kappa\delta}{K} \left(\frac{1}{\beta}-\frac{1}{b_{\text{max}}^{(i)} + c_{\text{max}}^{(i)}} \right)$, where the equality follows from \eqref{dynamic step size}. As $c_{\text{max}}^{(i)}$ decreases, $b_{\text{max}}^{(i)}$ decreases, thus the step size decreases.
\end{proof}}

\subsection{Classic Quantization}\label{Classic Quantization}
This traditional approach applies a uniform quantizer with a constant step size across all segments.
\blue{Since $c_{\text{max}}^{(i)} \leq c_{\text{max}}$ for all $i$, and the MAP block always assigns $b_n \leq b_{\text{max}}$ for all $n$, it follows that $b_{\text{max}}^{(i)} \leq b_{\text{max}}$ for every segment $i$. Substituting these inequalities into~\eqref{min Tn} and noting that $T_{\text{max}}^{(i)}$ is fixed across all segments, we obtain the quantization step size for the classic $K$-level uniform quantizer as
\begin{equation}\label{fixed step size}
\Delta = \kappa \delta \,\frac{b_{\text{max}} + c_{\text{max}} - \beta}{\beta \bigl(b_{\text{max}} + c_{\text{max}}\bigr)\,K}.
\end{equation}}

\off{The following lemma introduces the step size.
\begin{lemma}\label{lemma: fixed step size}
     Consider an AIF-TEM followed by a classic $K$-level uniform quantizer. For BL signals, the quantization step size $\Delta$ is given by
     \begin{equation}\label{fixed step size}
           \Delta = \kappa \delta \frac{b_{\text{max}}+c_{\text{max}}-\beta}{\beta(b_{\text{max}}+c_{\text{max}})K}.
     \end{equation}
\end{lemma}
\begin{proof}
    The proof follows directly using similar arguments as in Lemma~\ref{lemma: dynamic step size}. Since $c_{\text{max}}^{(i)} \leq c_{\text{max}}, \forall i$ and since the MAP block always assign $b_n \leq b_{\text{max}}, \forall n$, thus $b_{\text{max}}^{(i)} \leq b_{\text{max}} \forall i$. Substituting this to \eqref{min Tn} and since for all segments $T_{\text{max}}^{(i)}$ are fixed, the quantization step-size in \eqref{fixed step size} follows directly. 
\end{proof}}

\blue{As an example, if the maximum bias is set to $b_{\text{max}} = \beta + c_{\text{max}}$ with a fixed $\beta$, then the quantization step $\Delta$ increases with the maximum signal amplitude $c_{\text{max}}$. In contrast, if the bias is scaled as $b_{\text{max}} = \gamma c_{\text{max}}$ for some $\gamma > 1$, then $\Delta$ decreases as $c_{\text{max}}$ increases, thereby improving quantization resolution in high-amplitude regions. This behavior aligns with the result in \cite[Theorem 1]{naaman2021time}, where the authors showed that for IF-TEM, increasing the signal’s frequency or finite energy—both of which raise $c_{\text{max}}$—leads to a smaller quantization step size.}

\off{\begin{cor}
    Assume the setting in Lemma~\ref{lemma: fixed step size} and following Corollary~\ref{cor: cmaxi, step size relation}. For $b_{\text{max}} = \beta + c_{\text{max}}$ and fixed $\beta$, the quantization step $\Delta$ increases as the maximum amplitude $c_{\text{max}}$ increases.
\end{cor}
\begin{proof}
   This follows from Lemma~\ref{lemma: dynamic step size}. As $c_{\text{max}}$ increases, this leads to an increase in the quantization step $\Delta$.
\end{proof}
The next corollary is a direct result from \cite[Theorem 1]{naaman2021time}, where the authors demonstrated that for IF-TEM, as described by \eqref{step size for IF}, an increase in the frequency or the finite energy of the input signal leads to a reduction in the quantization step size.
\begin{cor}
   Assume the setting in Lemma~\ref{lemma: fixed step size} and $b_{\text{max}} = \beta +  c_{\text{max}} = \gamma c_{\text{max}}$,  for any $\gamma >1$. The quantization step $\Delta$ decreases as the maximum amplitude $c_{\text{max}}$ increases.
\end{cor}
\begin{proof}
   Since $\beta = \gamma c_{\text{max}} - c_{\text{max}} =  (\gamma-1)c_{\text{max}}$, substituting this into \eqref{fixed step size}, we get $ \Delta = \frac{\kappa \delta}{(\gamma+1)(\gamma-1)}\frac{2}{c_{\text{max}}K}$, which decreases as $c_{\text{max}}$ increases.
\end{proof}
}
\subsection{Upper Bound for Quantization Error}

We delve into the analysis of the MSE upper bound between the signal reconstructed from quantized time-encoded samples and the original input signal, $x(t)$. \newblue{The total reconstruction distortion is decomposed into two components: the reconstruction distortion $\textcolor{black}{D_R}$, corresponding to reconstruction from the non-quantized sampler output, and the quantization distortion $D_Q$, corresponding to the additional error introduced by quantization of the time intervals.}
\off{This analysis reveals that the total distortion is a cumulative effect of both the sampling distortion $\textcolor{black}{D_R}$ (the difference between the signal recovered directly from the sampler output and the original input) and the quantization distortion $D_Q$ (the difference between the signal recovered from quantized samples and the signal recovered from the non-noisy sampler output).}
\off{
\begin{theorem}\label{theo_total_dis}
    Consider an AIF-TEM, with the setting as in Theorem~\ref{Distortion AIF-TEM}, followed by a $K$-level dynamic uniform quantizer. For BL signals, the total distortion across all segments, as a function of the minimum sampling rate $f_{s_{\text{min}}}$, is 
    \begin{equation*}
        D_{T}(f_{s_{\text{min}}}) \leq \textcolor{black}{D_R}(f_{s_{\text{min}}}) + D_Q(f_{s_{\text{min}}}),
    \end{equation*}
    where $\textcolor{black}{D_R}(f_{s_{\text{min}}})$ and $D_Q(f_{s_{\text{min}}})$ are as given in \eqref{sampling distortion bound} and \eqref{Quantization rate distortion }, respectively.
\end{theorem}
To obtain the main result given in Theorem~\ref{theo_total_dis},} \off{ (the 
upper bound for $\textcolor{black}{D_R}(f_{s_{\text{min}}})$ is given in Theorem~\ref{Distortion AIF-TEM})}

In the following, we establish an upper bound for the quantization distortion $D_Q$. This analysis follows methodologies similar to those used in \cite{lazar2004perfect} and applies to both dynamic and classic (fixed step size) quantization schemes. We denote the quantization step size as $\Delta_i$ for each segment $i$. In the classical approach, $\Delta_i$ remains constant across all segments, i.e., $\Delta_i = \Delta$. In the dynamic approach, the step size $\Delta_i$ varies for each segment, i.e., $\Delta_i = \Delta_{W_i}$.

For each segment $W_i$, we consider the signal
\begin{equation}\label{practical ideal quantizted recovered signal}
     x_{d_i}(t) = \sum_{k =0}^{L_i} (I-\mathcal{\tilde{A}})^{k}\tilde{\mathcal{A}}x,
\end{equation}
where $\tilde{\mathcal{A}}$ is given by
\begin{equation*}
      \tilde{\mathcal{A}}x =\sum_{n\in S_{W_i}} \int_{\tilde{t}_{n-1}}^{\tilde{t}_{n}} x(u)du\  g(t-\tilde{\theta}_n),
\end{equation*}
and $\tilde{\theta}_n = (\tilde{t}_{n-1}+\tilde{t}_{n})/2$. 
The reconstructed signal from the quantizer output is given by
\begin{equation}\label{practical recovered sig from quant out}
   x_{q_i}(t) = \sum_{k=0}^{L_i} (I-\tilde{\mathcal{A}})^{k}\sum_{n \in S_{W_i}} 
   (\kappa \delta -b_n\tilde{T}_n) g(t-\tilde{\theta}_n).
\end{equation}

While the equality described in \eqref{IF relation times ampl} holds at the sampler output, the quantization process disrupts it. Specifically, $\int_{\tilde{t}_{n-1}}^{\tilde{t}_{n}} x(u),du \neq \kappa \delta - b_n\tilde{T}_n$, which introduces quantization distortion. We denote the signal reconstructed from the sampler output in segment $i$ as $x_{s_i}(t)$, which corresponds to $x_{L_i}$ as defined in \eqref{x_l}, where $l = L_i$, and $L_i$ is the number of samples in the $i$-th segment.

For a sufficiently large segment $W_i$, and assuming $\mathcal{S} = 1$, perfect reconstruction is achievable from the non-quantized AIF-TEM output, as suggested by Corollary~\ref{cor: perfect recovery}. \off{In such cases, $L_i$ is large enough and $x_{s_i}(t) = x(t)$. This observation is crucial for analyzing quantization distortion: when $L_i$ is large, we have $x_{d_i}(t) = \tilde{\mathcal{A}}^{-1} \tilde{\mathcal{A}}x = x$, and thus $x_{d_i}(t) = x_{s_i}(t)$.} In particular, when $L_i$ is sufficiently large, the reconstructed signal from the sampler output satisfies $x_{s_i}(t)\to x(t)$. In this regime, the signal $x_{d_i}(t)$, defined in~\eqref{practical ideal quantizted recovered signal}, provides a good approximation of $x_{s_i}(t)$. Consequently, the quantization distortion can be analyzed by studying the difference between the signals $x_{q_i}(t)$ and  $x_{d_i}(t)$.
However, for finite $W_i$ lengths, the signal reconstructed from the sampler output, $x_{s_i}(t)$, does not exactly match the original input $x(t)$ within the segment, giving rise to \newblue{reconstruction}\off{sampling} distortion, which was analyzed in Section~\ref{Analytical Results}. For these cases, we assume $x_{d_i}(t) \approx x_{s_i}(t)$, with the \newblue{reconstruction}\off{sampling} distortion relative to the input $x(t)$ elaborated in Theorem~\ref{Distortion AIF-TEM}.

The following lemma provides an upper bound on the quantization MSE for segment $i$, defined as the MSE between $x_{q_i}(t)$ and $x_{s_i}(t)$. \newblue{The analysis of the lemma adopts the classical high-resolution quantization model, in which the quantization error is modeled as an i.i.d. uniform sequence independent of the signal. While adjacent time samples generated from bandlimited signals exhibit correlation, this approximation is standard in high-resolution quantization theory and is widely used in time-encoding and sampling analyses (e.g., ~\cite{naaman2021time,lazar2004perfect}).}

\begin{lemma}\label{lemma: Quantization Distortion for segment i}
Assume the quantization error for some $n \in S_{W_i}$ by $d_n = \tilde{T}_n - T_n$ is a sequence of i.i.d. random variables in $[-\Delta_i/2, \Delta_i/2]$. Then, the MSE of the reconstructed signal from the AIF-TEM quantized times output in the $i$-th segment is bounded by
{\footnotesize
\begin{equation}\label{MSE Quant bound for segment i}
    \mathbb{E}[\mathcal{E}_q^2]_i \leq \frac{(1-r_{w_i}^{L_i+1})^2}{(1-r_{w_i})^2}\frac{(\Delta_i\kappa\delta)^2}{12}\frac{\Omega}{\pi}\overline{\left(\frac{1}{T_n}\right)_i^3}.
\end{equation}
}
\off{
\newblue{where $u(t)\triangleq\sum_{n\in S_{W_i}}\epsilon_n g(t-\tilde{\theta}_n)$ and $\mathcal{J}_i^{\max}(u)\triangleq\max_{0\le k\le L_i} \mathcal{J}_{\varphi_i}((I-\tilde{\mathcal A})^k u)$, with $\mathcal{J}_{\varphi_i}(\cdot)$ as in Lemma~\ref{lemma:Bernstein}.
}}
\end{lemma} 
\begin{proof}
 The MSE for segment $i$ is given by
 \begin{equation*}
     \mathbb{E}[\mathcal{E}_q^2]_i  =  \frac{1}{|W_i|} \mathbb{E}\left[\|  x_{q_i}(t) - x_{s_i}(t) \|_{W_i}^2\right].
 \end{equation*}
Substituting \eqref{practical ideal quantizted recovered signal} and \eqref{practical recovered sig from quant out} leads to
\begin{equation*}
     \mathbb{E}[\mathcal{E}_q^2]_i  = \frac{1}{|W_i|} \mathbb{E}\left[\| \sum_{k =0}^{L_i} (I-\tilde{\mathcal{A}})^{k}\sum_{n\in S_{W_i}} \epsilon_n g(t-\tilde{\theta}_n) \|_{W_i}^2\right],
\end{equation*}
where,
\begin{equation*}
     \epsilon_n  =   (\kappa \delta -b_n\tilde{T}_n) - \int_{\tilde{t}_{n-1}}^{\tilde{t}_{n}} x(u)du.
\end{equation*}  
Now, we upper bound the MSE. \newblue{Using $\|v\|_{W_i}\le \|v\|_2$,} we obtain 
{\small 
\begin{align*}
 & \newblue{\mathbb{E}[\mathcal{E}_q^2]_i \leq  \frac{1}{|W_i|} \mathbb{E}\left[\| \sum_{k =0}^{L_i} (I-\tilde{\mathcal{A}})^{k}\sum_{n\in S_{W_i}} \epsilon_n g(t-\tilde{\theta}_n) \|_{2}^2\right]}
 \\
 &\overset{(a)}{\leq}\|\sum_{k =0}^{L_i} (I-\tilde{\mathcal{A}})^{k}\|_{\newblue{2}}^2   \frac{1}{|W_i|}  \mathbb{E}\left[\|\sum_{n\in S_{W_i}} \epsilon_n g(t-\tilde{\theta}_n)\|_{\newblue{2}}^2\right] \\
& \overset{(b)}{\leq} \frac{(1-r_{w_i}^{L_i+1})^2}{(1-r_{w_i})^2} \frac{1}{|W_i|} \mathbb{E}\left[\|\sum_{n\in S_{W_i}} \epsilon_n g(t-\tilde{\theta}_n)\|_{2}^2\right] \\
& \overset{(c)}{=} \frac{(1-r_{w_i}^{L_i+1})^2}{(1-r_{w_i})^2} \frac{1}{|W_i|}\int_{-\newblue{\infty}}^{\newblue{\infty}}\!\sum_{\scriptstyle n,m\in  S_{W_i}}  g(t-\tilde{\theta}_n)g(t-\tilde{\theta}_m)\mathbb{E}[\epsilon_n\epsilon_m]dt \\
& \overset{(d)}{=} \frac{(1-r_{w_i}^{L_i+1})^2}{(1-r_{w_i})^2} \frac{1}{|W_i|}\sum_{n\in S_{W_i}}\left(\frac{\kappa\delta}{T_n}\right)^2 \frac{\Delta_i^2}{12} \int_{-\infty}^{\infty} g^2(t-\tilde{\theta}_n)dt \\
\off{& \overset{(e)}{\leq} \frac{(1-r_{w_i}^{L_i+1})^2}{(1-r_{w_i})^2}\frac{\Omega(\Delta_i\kappa\delta)^2}{12\pi} \frac{1}{L_i{\overline{T}}_i} \sum_{n\in S_{W_i}} \left(\frac{1}{T_n}\right)^2 \\}
& \overset{(e)}{\leq} \frac{(1-r_{w_i}^{L_i+1})^2}{(1-r_{w_i})^2}\frac{\Omega(\Delta_i\kappa\delta)^2}{12\pi} \frac{1}{L_i{\overline{T_n}}} \sum_{n\in S_{W_i}} \left(\frac{1}{T_n}\right)^2 \\
& \overset{(f)}{\leq}  \frac{(1-r_{w_i}^{L_i+1})^2}{(1-r_{w_i})^2}\frac{\Omega(\Delta_i\kappa\delta)^2}{12\pi}\overline{\left(\frac{1}{T_n}\right)_i^3},
\end{align*}
}
\newblue{where (a) follows from the norm properties, i.e., $\|Tv\|_2 \le \|T\|_2\|v\|_2$. (b)  follows from Lemma~\ref{key_lemmav2} applied in $L^2(\mathbb R)$, which yields
$\|(I-\tilde{\mathcal A})v\|_2 \le r_{w_i}\|v\|_2$ for signals generated from the sample set $S_{W_i}$. Therefore, $\|(I-\tilde{\mathcal A})^k\|_2 \le r_{w_i}^k$, and the operator-sum norm is bounded by the geometric series $\sum_{k=0}^{L_i} r_{w_i}^k = \frac{1-r_{w_i}^{L_i+1}}{1-r_{w_i}}$.}\off{is given from Lemma~\ref{key_lemma}, and since $r_{w_i} < 1$, the norm of a finite sum of the operator is bounded using the geometric series $\sum_{k =0}^{L_i} r_{w_i}^k$.} (c) follows directly by applying the norm. (d) follows since $\mathbb{E}[\epsilon_n\epsilon_m]  = (\frac{\kappa\delta}{T_n})^2 \frac{\Delta_i^2}{12}\delta_{n,m}$ (the proof is given in Appendix~\ref{math_obs1})), and by interchanging the sum and the integral. (e) follows from Parseval's theorem and the bandlimited spectrum of $g(t)$, $\int_{-\infty}^{\infty} g^2(t-\eta)dt = \frac{1}{2\pi} \int_{-\infty}^{\infty} 1_{[-\Omega,\Omega]}dw =\frac{\Omega}{\pi}$, and substituting $|W_i| = L_i\overline{T_n}$ for $n \in S_{W_i}$. 
%\ale{In (e) above you have $\overline{T}_i$ here you have $\overline{T}$. I am confused by the notation. Please check this in the entire paper. I think you want to say $|W_i| = L_i\overline{T_n}$ for $n \in S_{W_i}$ and then also in the Appendix~\ref{math_obs3} you need the subscript $n$ in $\overline{T}$} 
(f) By the inequality 
\[
   \left({1}/{\overline{T_n}}\right)\;\overline{\Bigl({1}/{T_n}\Bigr)^2}
   \;\;\le\;\;
   \overline{\Bigl({1}/{T_n}\Bigr)^3}
   \quad
   \text{(Proof in Appendix~\ref{math_obs3})},
\]
%\ale{Also here, I think you need to replace $\overline{T}$ with $\overline{T_n}$ for $n \in S_{W_i}$ }
we immediately obtain the desired bound.
\end{proof}
%\begin{remark}
%  \newblue{The assumption that the quantization error sequence is i.i.d. uniform follows the classical high-resolution quantization model widely adopted in sampling and time-encoding analyses, see, e.g.,~\cite{naaman2021time,lazar2004perfect}. \textcolor{pink}{Although the time intervals $T_n$ are generally correlated, this approximation is standard in the high-resolution regime and simplifies the MSE analysis.} \textcolor{red}{Although the time intervals $T_n$ generated from bandlimited signals are generally correlated, in the high-resolution regime, this modeling assumption is standard and provides an accurate first-order characterization of the MSE.}}
%\end{remark}
Next, we establish an upper bound for the quantization MSE by utilizing the temporal biases $b_n$ and the temporal amplitudes $c_n$.
\begin{cor}
  Assume the setting in Theorem~\ref{lemma: Quantization Distortion for segment i}. Then the MSE of the reconstructed signal from AIF-TEM quantized times output in the $i$-th segment is bounded by
    \begin{equation}\label{Quant bound 2}
           \mathbb{E}[\mathcal{E}_q^2]_i \leq \frac{(1-r_{w_i}^{L_i+1})^2}{(1-r_{w_i})^2}\frac{\Delta_i^2}{12}\frac{\Omega}{\pi}\frac{(\overline{b_{n
        }+c_n)_i^3}}{\kappa \delta}. 
    \end{equation}
\end{cor}
\begin{proof}
    By substituting the lower bound of \eqref{mean T-AIF} with ${\overline{T_n}}$ for $n \in S_{W_i}$ into equation \eqref{MSE Quant bound for segment i}, we obtain the desired result.
\end{proof}

\blue{The total quantization distortion, $D_Q$, for segments with maximum Nyquist ratio $r_{w_i} < 1$, is upper bounded by the average of the segment-wise quantization MSE values
\begin{equation}
D_Q \leq \frac{1}{\mathcal{S}} \sum_{i=1}^{\mathcal{S}} \mathbb{E}[\mathcal{E}_q^2]_i.
\end{equation}}

\blue{In summary, the total distortion, $D_T$, is upper bounded by the sum of the average \newblue{reconstruction}\off{sampling and quantization} distortion \newblue{across all segments (see Subsection~\ref{Analytical Results}) and the additional distortion introduced by quantization.} That is, 
\[
D_T \leq \textcolor{black}{D_R} + D_Q,
\]
where $\textcolor{black}{D_R}$ is given in \eqref{sampling distortion bound}.}

\off{\begin{cor}\label{Cor: Quant rate Distortion }
Assume the setting in Theorem~\ref{lemma: Quantization Distortion for segment i}. Then the total quantization distortion , for maximum Nyquist ratio $r_{w_i} < 1$ in each segment $W_i$, is upper bounded by
\begin{equation}\label{Quantization rate distortion }
     D_Q(f_{s_{\text{min}}}) \leq \frac{1}{\mathcal{S}}\sum_{i=1}^\mathcal{S} \frac{(1-r_{w_i}^{L_i+1})^2}{(1-r_{w_i})^2}\frac{\Delta_i^2}{12}\frac{\Omega}{\pi}\frac{(\kappa\delta)^2}{ \mathbb{E}^3[T_n]_{W_i}}
\end{equation}
\end{cor}
\begin{proof}
    The proof follows directly from averaging the MSE bound provided in \eqref{MSE Quant bound for segment i} across all segments.
\end{proof}}
\begin{remark}
The analytical results of this section also apply to IF-TEM with fixed parameters instead of the adaptive parameters for the AIF-TEM.
\end{remark}
\begin{remark}
    For a sufficiently large segment, applying the fixed parameters for IF-TEM ($b_n = b_{\text{IF}}$, $r_{w_i} = r_c$, $c_n = c_{\text{max}}$) to the bound in \eqref{Quant bound 2}, we obtain the bound provided in \eqref{quantization mse bound IF-TEM}, as presented by \cite{naaman2021time}.
\end{remark}

\subsection{Comparison to IF-TEM}

This section compares the oversampling of AIF-TEM and IF-TEM, and relates it to the quantization MSE of both samplers. 
Let $OS_{c}$ denote the oversampling of IF-TEM and $OSU_{c}$ denote its upper bound. By applying the fixed bias $b_{\text{IF}}$ to \eqref{OS upper bound AIF-TEM}, we obtain $OS_{c} \leq \frac{b_{\text{IF}}+c_{\text{max}}}{\kappa \delta}\frac{\pi}{\Omega} \triangleq OSU_{c}.$
Let $f_{s_{a}}$ and $f_{s_{c}}$ denote the average sampling frequencies for AIF-TEM and IF-TEM, respectively. From Definition~\ref{TEM Oversampling} and Eq.~\eqref{mean T-AIF}, we get $f_{s_{a}} \leq \frac{\overline{b}_n+\overline{c}_n}{\kappa \delta}$ and by applying the fixed bias and maximum amplitude $b_{\text{IF}}$ and $c_{\text{max}}$, we have $f_{s_{c}} \leq \frac{b_{\text{IF}}+c_{\text{max}}}{\kappa \delta}$.
The following proposition compares the oversampling rates of AIF-TEM and IF-TEM.
\begin{prop}\label{OS FOR aif-tem, if-tem}
Let $x = x(t)$, $t\in \mathbb{R}$, be a $2\Omega$-BL signal with $|x(t)| \leq c_{\text{max}} < \infty$. Assume $x(t)$ is sampled using an IF-TEM characterized by parameters $b_{\text{IF}} > c_{\text{max}}$, $\kappa$, and $\delta$, and an AIF-TEM with shared parameters $\kappa$ and $\delta$ and a correct operating MAP. For $b_n \leq b_{\text{IF}}, \forall n$, then $fs_{a} \leq fs_{c}$ and $OS_{a} \leq OS_{c}$.
\end{prop}
\begin{proof}
Given the successful operation of MAP and $b_n \leq b_{\text{IF}}$, we have, $x(t) + b_n \leq x(t) +b_{\text{IF}}$. With fixed $\kappa$ and $\delta$, and referencing \eqref{Integ out-AIF-TEM} and \eqref{Integ out-IF-TEM}, the output of the AIF-TEM integrator reaches the threshold $\delta$ more slowly than the IF-TEM. This results in larger time intervals, leading to a reduced average sampling frequency and oversampling for AIF-TEM.
\end{proof}
\begin{remark}
    Under the conditions of Proposition~\ref{OS FOR aif-tem, if-tem}, the quantization MSE bound given in \eqref{Quant bound 2} for AIF-TEM is smaller than that of IF-TEM. This improvement arises because the adaptive bias $b_n$ in AIF-TEM is typically smaller than the fixed worst-case bias $b_{\text{IF}}$ used in IF-TEM, resulting in a reduced quantization error bound.
\end{remark}

\blue{Let $\nu$ denote the oversampling ratio between IF-TEM and AIF-TEM, defined as $\nu = OS_c / OS_a$, and let $\rho$ be the number of bits used to encode the bias. The following proposition shows that, under certain conditions, AIF-TEM requires fewer total bits to encode the signal than IF-TEM.}

\blue{ \begin{prop}\label{theorem: total bits ratio} Consider IF-TEM and AIF-TEM configured as in Proposition~\ref{OS FOR aif-tem, if-tem}, followed by a $K$-level dynamic uniform quantizer. The total number of bits used to encode the signal is smaller for AIF-TEM if $\nu - 1 \geq \frac{\rho}{\log_2 K}$. In this case, the ratio of total bits between IF-TEM and AIF-TEM is given by \begin{equation}\label{bits ratio} r_{\text{bits}} = \frac{\nu \log_2 K}{\log_2 K + \rho}. \end{equation} \end{prop} }

\begin{proof}\blue{ Let $L_a$ denote the number of samples generated by AIF-TEM. Then IF-TEM produces $\nu L_a$ samples. With a $K$-level dynamic uniform quantizer, the total number of bits for IF-TEM is $\nu L_a \log_2 K$, where $\nu \geq 1$ as shown in Proposition~\ref{OS FOR aif-tem, if-tem}.}

\blue{For AIF-TEM, the total number of bits is $L_a \log_2 K + \rho L_a$. Hence, the ratio of total bits is
\begin{equation*} 
r_{\text{bits}} = \frac{\nu L_a \log_2 K}{L_a (\log_2 K + \rho)}. 
\end{equation*} 
Dividing the numerator and the denominator by $L_a$ yields the expression in \eqref{bits ratio}. For AIF-TEM to use fewer bits, we require $r_{\text{bits}} > 1$, which leads to 
%\begin{align*}  
$\nu \log_2 K > \log_2 K + \rho,$ for  $\nu > 1 + {\rho}/{\log_2 K}$.} 
%\end{align*}}
\end{proof} 
\blue{To summarize, we demonstrate that under the same parameters $\kappa$ and $\delta$ for both IF-TEM and AIF-TEM, and assuming the conditions in Proposition~\ref{OS FOR aif-tem, if-tem} hold, AIF-TEM not only achieves a lower oversampling rate but also results in smaller quantization error and fewer total bits, as established in Proposition~\ref{theorem: total bits ratio}.}

\subsection{Evaluation results for the Quantization Process}\label{Evaluation results for Quantization Process}
\ifpagelimit
%See Appendix~\ref{app Evaluation results for Quantization Process} due to space limitations.

%\else
This section presents the simulation results of the quantization process applied to the proposed AIF-TEM and the classical IF-TEM. For AIF-TEM, the evaluation includes both classical and dynamic quantization techniques.

We start by comparing the quantization performance of AIF-TEM and classical IF-TEM with a standard uniform quantizer. The input signals are defined by the equation \eqref{input signal}, with $\Omega$ varying in the range $2\pi \cdot [10, 50]$ Hz and $M=2$. The coefficients $a[n]$ are randomly selected 100 times from the interval [-1, 1]. The maximum amplitude is calculated as $c_{\text{max}} = \sqrt{(E\Omega)/\pi} = 2$.

\off{
\begin{figure}[!ht]
  \centering
  \begin{subfigure}[c]{0.45\textwidth}
    \includegraphics[width=0.9\textwidth]{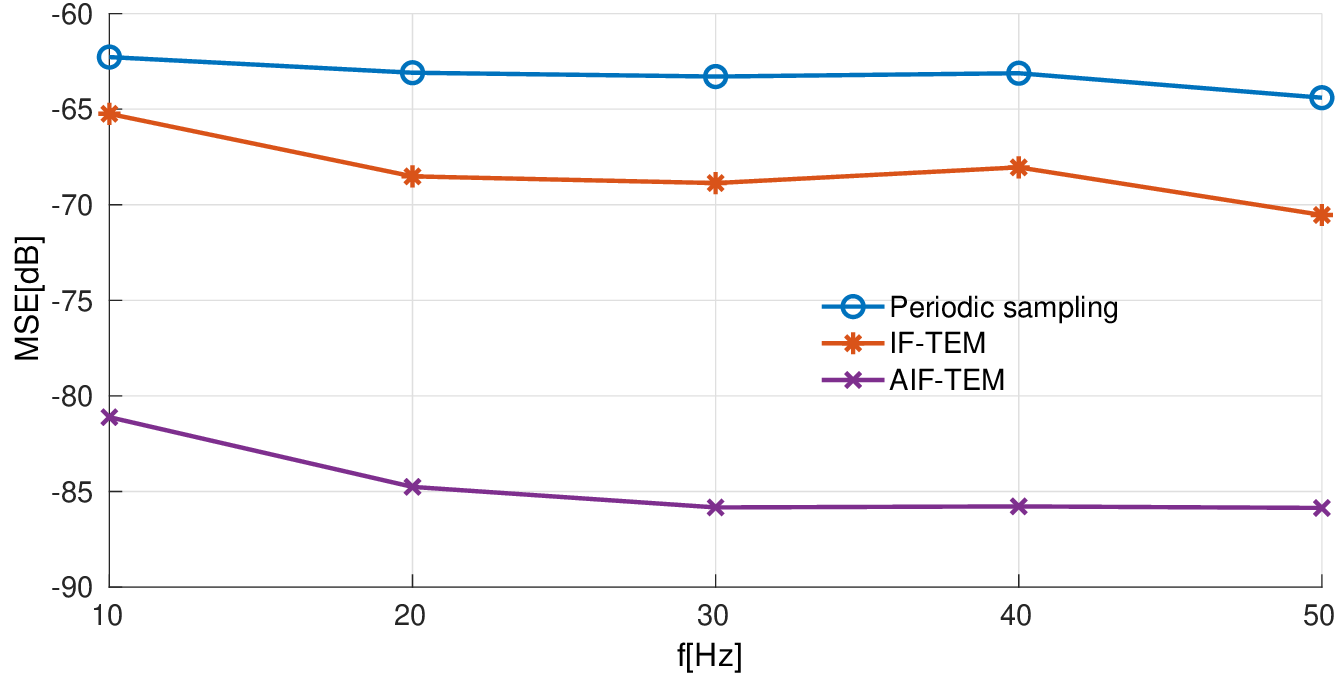}
    \caption{MSE of Quantization and \textcolor{red}{reconstruction}\off{sampling} for IF-TEM, AIF-TEM using 12 bits}
    \label{fig:MSE Quant1}
  \end{subfigure}
  \hfill % adds horizontal space between the two subfigures
  \begin{subfigure}[c]{0.45\textwidth}
    \includegraphics[width=0.9\textwidth]{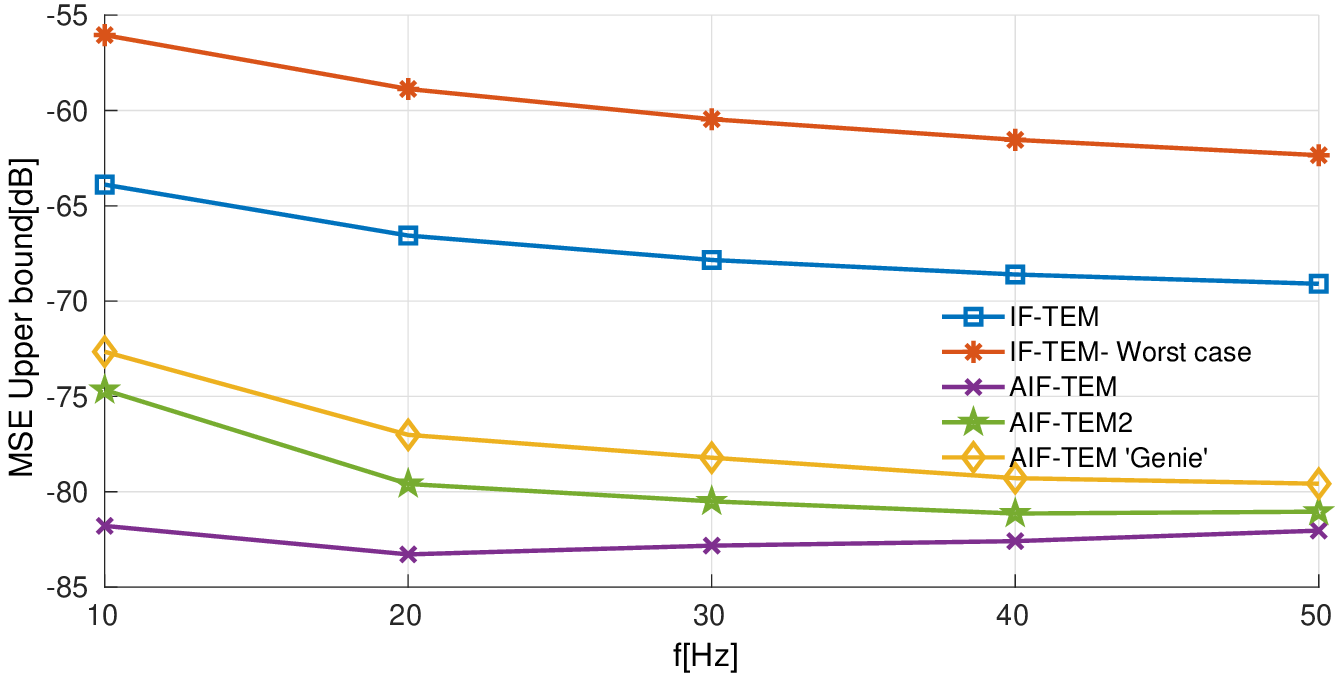}
    \caption{MSE upper bound, calculated based on \eqref{MSE Quant bound for segment i} in blue and red for IF-TEM and AIF-TEM, respectively. IF-Worst case is based on \eqref{Quant bound 2}, where bias and amplitude are fixed. AIF-TEM 'Genie' and AIF-TEM2 are based on \eqref{Quant bound 2} where the first is with the real amplitude provided and the second with estimated amplitude.}
    \label{fig:MSE Upper bound}
  \end{subfigure}
  \caption{Performance comparison of quantization for IF-TEM and AIF-TEM samplers using 12-bit quantizer.}
  \label{fig:quantization_results}
  %\vspace{-0.5cm}
\end{figure}
}
\begin{figure}[!ht]
\centering
\includegraphics[width=0.4\textwidth]{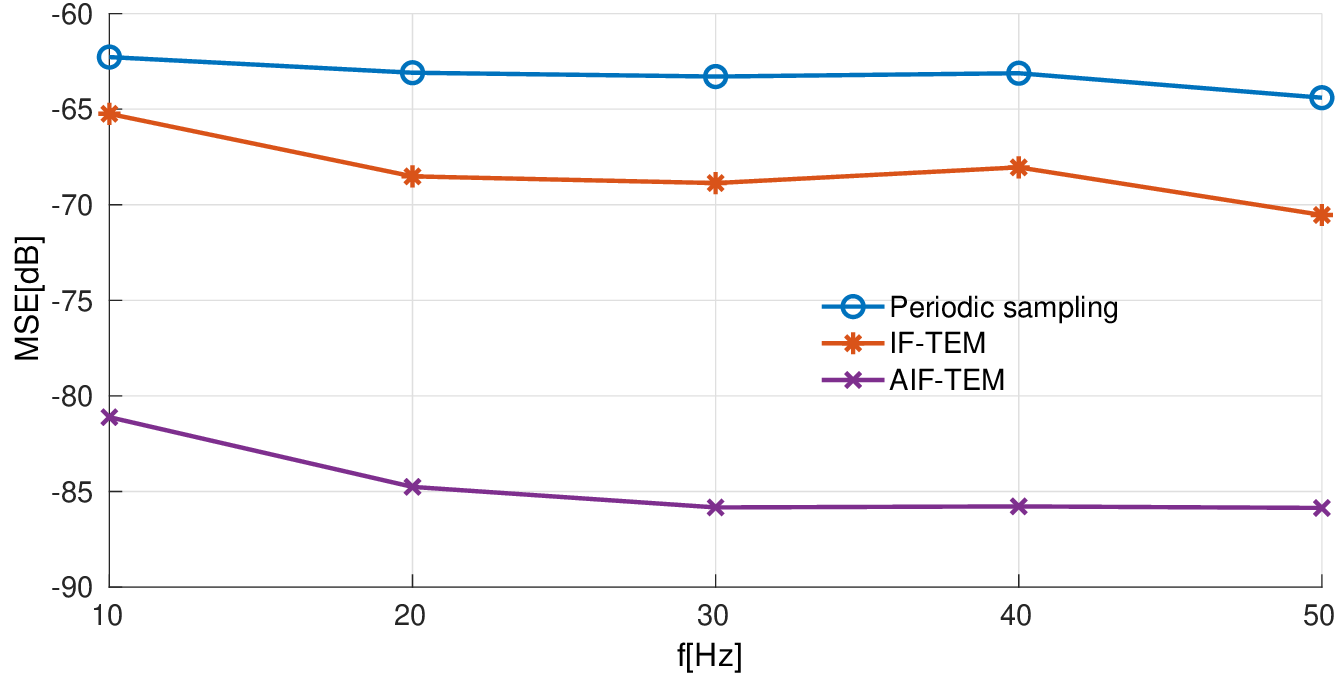}
\caption{\small MSE of Quantization and \newblue{reconstruction}\off{sampling} for IF-TEM, AIF-TEM using 12 bits}\label{fig:MSE Quant1}
\vspace{-0.2cm}
\end{figure}

Figure~\ref{fig:MSE Quant1} shows the MSE due to quantization and reconstruction\off{sampling}, as described in \eqref{MSE}. The performance of AIF-TEM is depicted by a red line, utilizing parameters $\delta =0.018$ and $\kappa = 0.21$. $\beta$ is set to achieve a Nyquist ratio of $r_a = 0.39$, \blue{and the MAP block operates (as proposed in Section~\ref{MAP Modes}) with a window of $w=1$. The weighting factors are configured as $\alpha_1=0.98$ and $\alpha_2=0.17$ as in Section~\ref{EVALUATION RESULTS}.} The blue line represents IF-TEM, configured with the same values of $\kappa$ and $\delta$. We adjusted $b_{\text{IF}}$ to $c_{\text{max}} + \beta$, such that $r_c =0.39$. The average oversampling rate for IF-TEM is higher than that of AIF-TEM, with ratios of $[4.1, 3.6, 3.2, 2.9, 2.7]$ across frequencies ranging from 10 Hz to 50 Hz. The biases $b_n$ are encoded using 4 bits. Additionally, the total number of bits required for IF-TEM is greater than that of AIF-TEM, with bit usage ratios of $[3.1, 2.7, 2.4, 2.2, 2.0]$.

AIF-TEM clearly outperforms IF-TEM in terms of MSE. It is important to note that in \cite{naaman2021time}, the authors compare IF-TEM with classical sampling under quantization and establish conditions under which the quantization error of IF-TEM is lower than that of the periodic sampler.

\begin{figure}[!ht]
\centering
\includegraphics[width=0.45\textwidth]{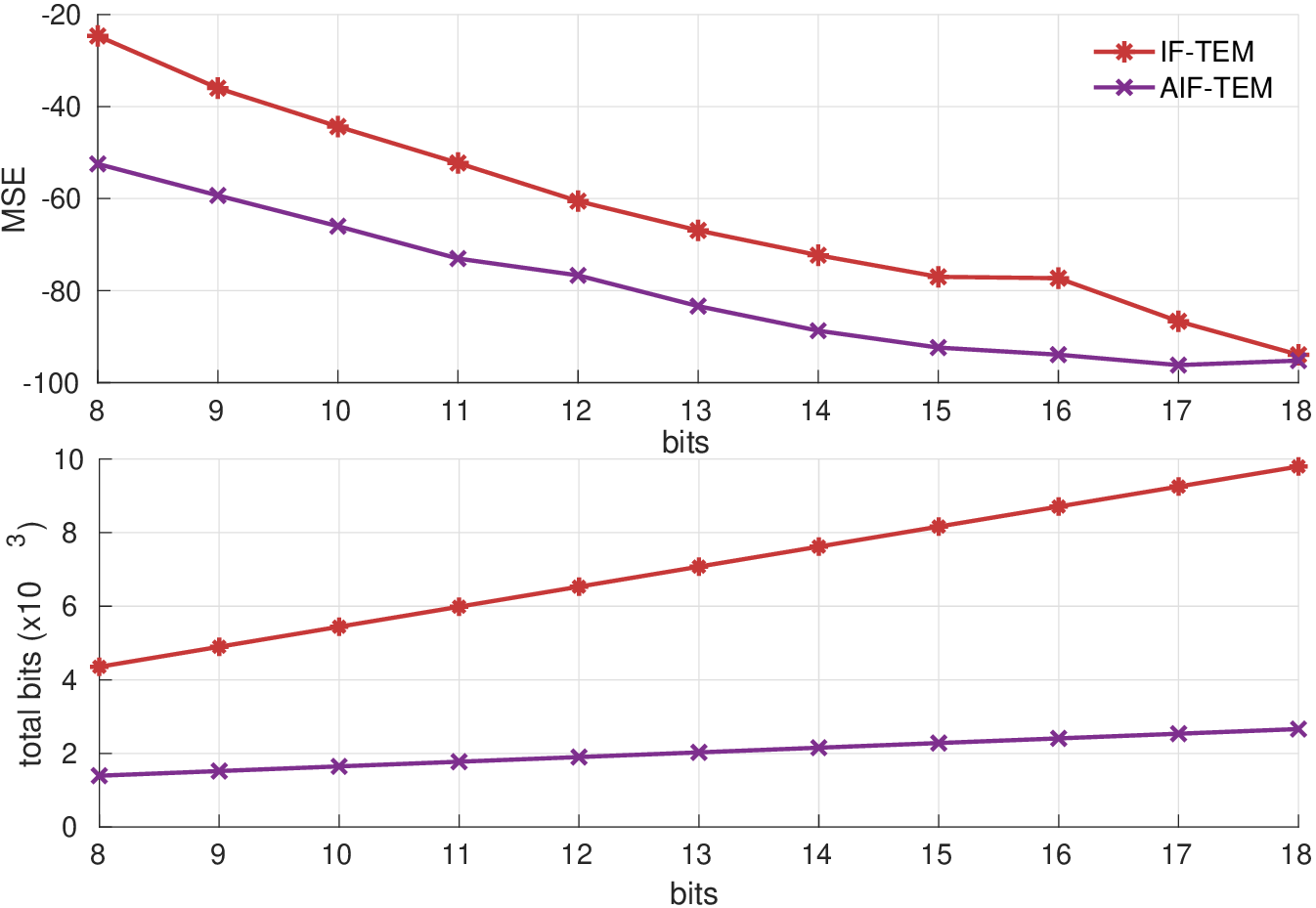}
\caption{\small \blue{(a) Comparison in terms of MSE (in dB) between IF-TEM and AIF-TEM. The evaluation includes quantization across different bit levels for $T_n$, with 3-bit quantization used for the biases in AIF-TEM. (b) Total number of bits required to encode the signal by the samplers.}}\label{fig:MSE vs Bits}
\vspace{-0.4cm}
\end{figure}

\blue{Figure~\ref{fig:MSE vs Bits}(a) showcases the MSE of quantization for both IF-TEM (red line) and AIF-TEM (purple line). For both samplers, the parameters used are $\delta = 0.0094$ and $\kappa = 0.24$. For AIF-TEM, $\beta=0.1$ is used, and for IF-TEM, the bias $b_{\text{IF}}$ is adjusted to $c_{\text{max}} + \beta$. The number of bits used to encode the biases $b_n$ in AIF-TEM is 3 bits. The oversampling ratio between IF-TEM and AIF-TEM is 4.3, ensuring that the condition of Proposition 2 for a smaller total number of bits in AIF-TEM holds. Figure~\ref{fig:MSE vs Bits}(b) shows the total number of bits required to encode the signal across different bit values used to encode the time differences $T_n$. The MSE performance of AIF-TEM clearly surpasses that of IF-TEM, and AIF-TEM also uses fewer total bits.}

Next, we explore the benefits of dynamic quantization for AIF-TEM as detailed in Subsection~\ref{Dynamic Quantization}, compared to classic quantization for both AIF-TEM and IF-TEM. For this analysis, we use a $2\Omega$-BL signal, $x(t)$ (where $\Omega = 2\pi[10:50]$ Hz), divided into three segments. Each segment $i$ consists of a signal $x_i(t) = a_i\sin(\Omega t)$, with coefficient $a_i$ randomly selected 50 times from the interval [0, 1] for each segment. Each segment is recovered independently using the decoding algorithm described in Section~\ref{decoding process}.

The AIF-TEM sampler configuration includes parameters $\kappa = 0.18$, $\beta = 0.1$, and a threshold $\delta$ set to achieve a Nyquist ratio $r_a = 0.67$ for each frequency, with the MAP block utilizing a window of $w = 15$. IF-TEM is configured with the same $\kappa$ and $\delta_c=\delta$, but with a fixed bias $b_{\text{IF}} = c_{\text{max}} + \beta$, ensuring $r_c = r_a = 0.67$. With these settings, the average oversampling rate is $16.25$ for IF-TEM and $11.2$ for AIF-TEM.

\begin{figure}[!ht]
  \centering
  \begin{subfigure}[c]{0.45\textwidth}
    \includegraphics[width=0.9\textwidth]{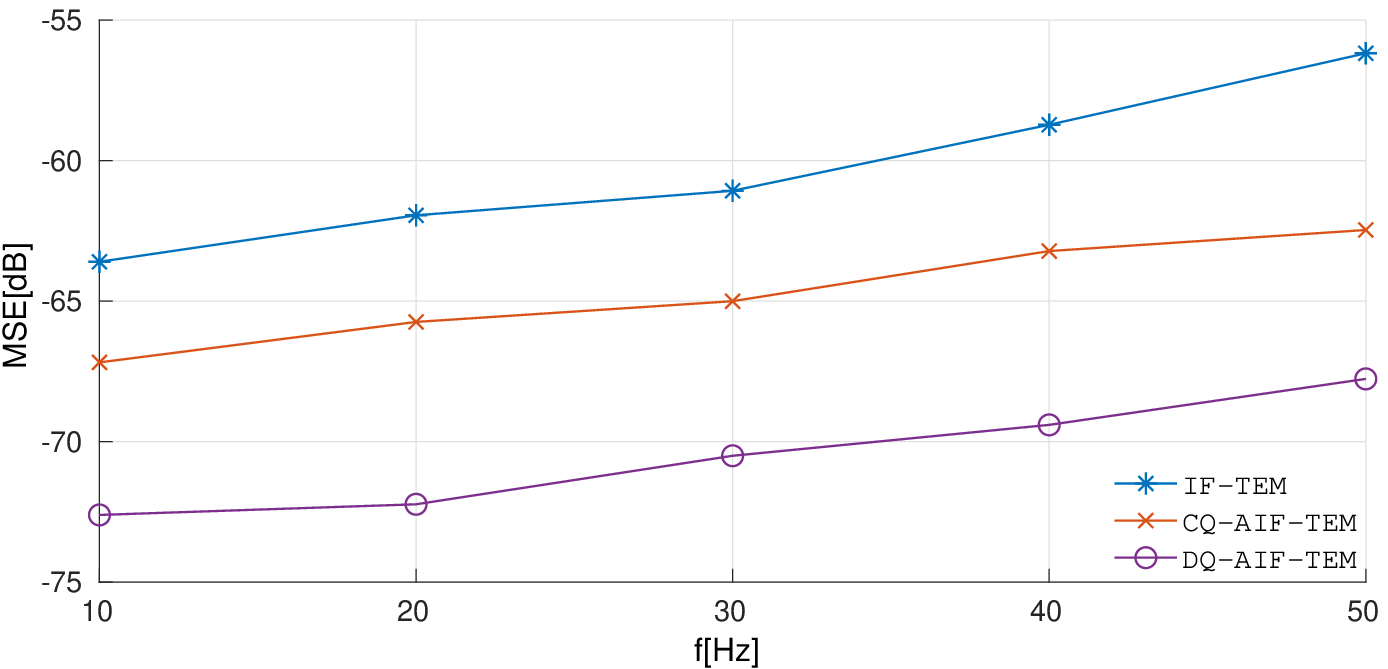}
    \caption{\small MSE of the quantization using 12 bits.}
    \label{fig:sub1}
  \end{subfigure}
  \hfill % adds horizontal space between the two subfigures
  \begin{subfigure}[c]{0.4\textwidth}
    \includegraphics[width=0.9\textwidth]{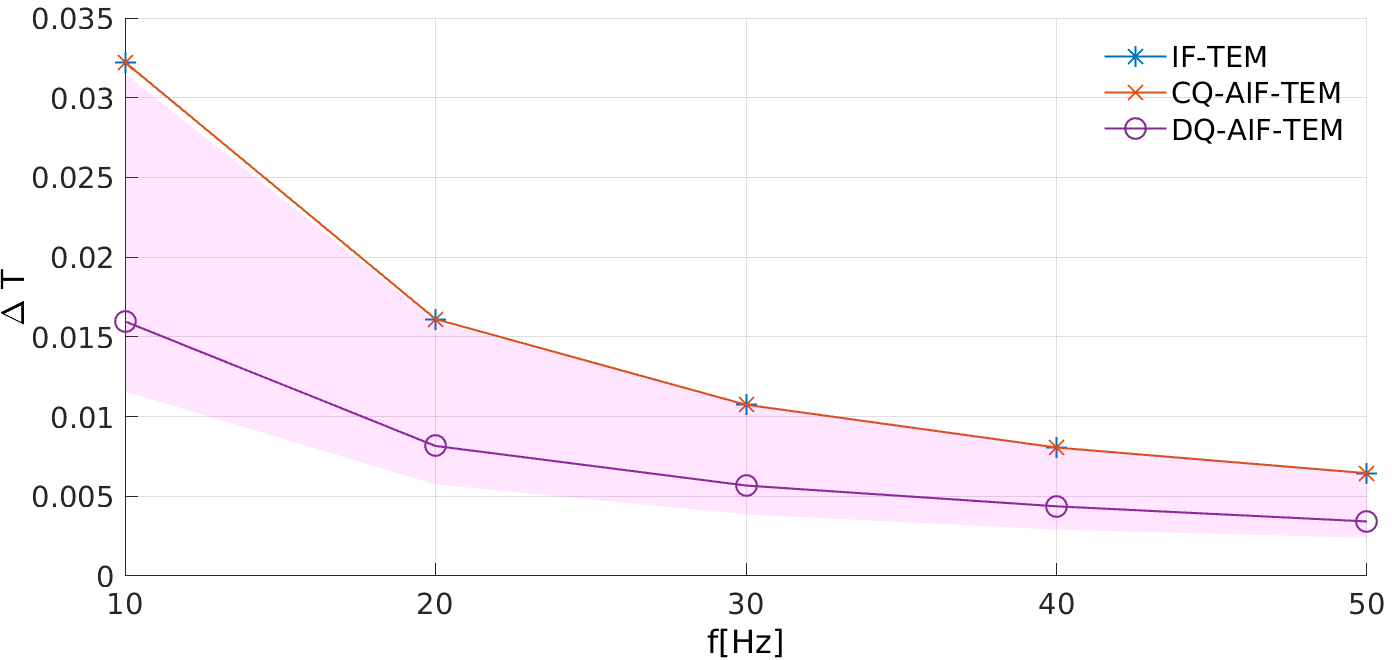}
    \caption{\small Dynamic range of time differences, $\Delta T = (T_{\text{max}} - T_{\text{min}})$. The purple line represents the average dynamic range of time differences for each segment $i$, expressed as $\Delta T_i = (T_{\text{max}}^{(i)} - T_{\text{min}}^{(i)})$, and the shaded region captures the range of $\Delta T_i$.}
    \label{fig:sub2}
  \end{subfigure}
  \caption{\small Performance comparison of quantization methods for IF-TEM, Classic Quantization for AIF-TEM ( CQ-AIF), and Dynamic Quantization for AIF-TEM (DQ-AIF).}
  \label{fig:MSE and time differnce, sin signal}
\vspace{-0.5cm}
\end{figure}

Figure~\ref{fig:MSE and time differnce, sin signal} illustrates the performance of classic quantization for IF-TEM and both classic and dynamic quantization methods for AIF-TEM.
Subfigure~\ref{fig:MSE and time differnce, sin signal}\subref{fig:sub1} displays the total MSE of quantization and reconstruction\off{sampling}, as defined in \eqref{MSE}.
Subfigure~\ref{fig:MSE and time differnce, sin signal}\subref{fig:sub2} shows the dynamic range of time differences denoted by $\Delta T = (T_{\text{max}} - T_{\text{min}})$.
The classic quantization performance for IF-TEM is indicated by the blue line, as described in Section~\ref{Classic Quantization}, which uses a fixed quantization step across all segments, given by equation~\eqref{fixed step size}. For AIF-TEM, the red line illustrates the performance of classic quantization, while the purple line represents dynamic quantization. The dynamic approach varies the quantization step for each segment based on the estimated maximum amplitude calculated per equation~\eqref{dynamic step size}. The shaded purple region marks the range of maximum time differences per segment $i$, as defined in equation~\eqref{dynamic range of Tn}, which directly impacts the quantization step size. Here, $T_{max}$ is calculated based on \eqref{max Tn}, with the estimated amplitude $\hat{c}_n$. 

Dynamic quantization significantly improves MSE performance, reducing it by at least 6 dB compared to classic quantization for AIF-TEM and by at least 10 dB compared to IF-TEM.

\begin{table}[ht]
\begin{center}  
    \footnotesize % Reduce the font size
    \begin{tabular}{|l||l|l|l||l|l|l||l|l|l|}
    \hline
    & \multicolumn{3}{c||}{\bf Segment A} & \multicolumn{3}{c||}{\bf Segment B} & \multicolumn{3}{c|}{\bf Segment C} \\
    \hline
    \textbf{Sam.}   & M    & B   & $\Delta$T  & M    & B   & $\Delta$T  & M   & B   & $\Delta$T    \\ \hline
    \textbf{IF.}      & -64  & -43 & 4.5  & -64 & -43 & 4.5 & -73 & -43 & 4.5 \\ \hline
    \textbf{AIF.}     & -63 & -59 & 3.1 & -69  &-62 & 3.1 &-75 &-69 & 3.1 \\ \hline
        \textbf{DAI.}     & -63 & -59 & 3.1 & -74 & -72 & 1.8 & -89 &  -82 & 1.3\\ \hline
    \end{tabular}
\end{center}
\vspace{-0.3cm}
\caption{\small\label{table:Sin Sig} Comparative performance with BL signal. The table represents the MSE in dB (M) of recovery for each quantizer (classic quantization for IF-TEM (IF), AIF-TEM (AIF), and dynamic quantization for AIF-TEM (DAI)), the MSE bound in dB (B), and the dynamic range of time differences ($\Delta T$) in sec ($10^{-2}$) for each segment of the signal.}\vspace{-0.4cm}
\label{table:Sin Sig quantization}
\end{table}

Table~\ref{table:Sin Sig quantization} illustrates the performance of classic quantization for IF-TEM and both classic and dynamic quantization methods for AIF-TEM. Using the same signal as in Figure~\ref{table:Sin Sig}, for each segment, the table represents the MSE of quantization and the MSE bound given in equation~\eqref{Quant bound 2} with the estimated amplitude for AIF-TEM and fixed bias and maximum amplitude $c_{\text{max}}$ for IF-TEM. Dynamic quantization improves MSE performance and reduces the quantization step size based on the amplitude of each segment.

\else
\section{Conclusion}
\off{
This paper introduces a novel Adaptive Integrate-and-Fire Time Encoding Machine (AIF-TEM) designed to dynamically adjust to changes in the amplitude and frequency of input signals. Moreover, it provides an analytical analysis of the oversampling, and \textcolor{red}{reconstruction}\off{sampling} and quantization rate-distortion.
The adaptive nature of AIF-TEM allows it to outperform traditional IF-TEM and periodic sampling methods, particularly in scenarios with varying signal characteristics. Our extensive analysis and empirical evaluations demonstrate that AIF-TEM achieves significant improvements in mean square error (MSE) and oversampling rates, highlighting its potential for practical applications in signal processing.

The performance of the MAP block, which estimates the \textcolor{red}{temporal} amplitude, is crucial for adaptive IF-TEM. Even a simple MAP block implementation shows good performance for AIF-TEM. Future improvements in implementations and estimators could further enhance the performance and optimization of AIF-TEM.
AIF-TEM offers a robust framework for adaptive sampling and quantization, paving the way for more efficient analog-to-digital conversion methods. As technology advances, AIF-TEM's adaptive approach could significantly impact next-generation signal processing systems.
}
This paper introduces an Adaptive Integrate-and-Fire Time Encoding Machine (AIF-TEM), designed to dynamically adjust to changes in input signal amplitude and frequency. Our analysis covers oversampling, sampling distortion, and quantization error. The performance of the MAP block, which estimates \textcolor{red}{temporal} amplitude, is crucial for the proposed adaptive approach.
AIF-TEM offers a robust framework for adaptive sampling and quantization, paving the way for more efficient analog-to-digital conversion methods.
As we demonstrate, an efficient adaptive design with a simple MAP block implementation can achieve significant reductions in MSE, oversampling rate, and total bit usage, as demonstrated in our evaluation results. Future work includes exploring new MAP block implementations and estimators that could further optimize AIF-TEM performance.

% or


%\appendices

\appendix

\subsection{Definitions}\label{Appendix a}
\begin{definition}\label{def: convolotion operator}
     The operator $\mathcal{G}$ maps an arbitrary function $x$ into a bandlimited function via $\mathcal{G}x = (g * x)$, where $*$ denotes convolution and $g(t) = \sin(\Omega t)/\pi t$, known as the sinc function.
\end{definition}

\begin{definition}
     The operator $\mathcal{A}^*$ is defined by 
    \begin{equation}\label{adjoint operator}
         \mathcal{A^*}x = \sum_{n \in \mathbb{Z}}x(\theta_n)\mathcal{G}1_{[t_{n-1}, t_n]},
    \end{equation}
    where $\mathcal{G}1_{[t_{n-1}, t_n]}$ applies the operator $\mathcal{G}$ to a pulse function defined over the interval $[t_{n-1}, t_n]$.
    \end{definition}
\begin{definition}\label{norm}
    Let $f: \mathbb{R} \to \mathbb{R}$ and and let $\mathcal{T}$ be a general time window defined by $\mathcal{T} = [t_{\text{start}}, t_{\text{end}}]$, with $t_{\text{start}}$ and $t_{\text{end}}$ being the start and end points of the window, respectively. The norm $\| f \|^2_{\mathcal{T}}$ is given by
\begin{equation*}
    \| f \|^2_{\mathcal{T}} = \int_{t_{\text{start}}}^{t_{\text{end}}} \abs{f(u)}^2du.
\end{equation*}
\end{definition}

\begin{definition}\label{product}
    Let $f: \mathbb{R} \to \mathbb{R}, g: \mathbb{R} \to \mathbb{R}$, the inner product $\langle f,g\rangle_{\mathcal{T}}$ is defined by 
    \begin{equation*}
    \langle f,g\rangle_{\mathcal{T}} = \int_{t_{\text{start}}}^{t_{\text{end}}} f(u)g(u)du.
\end{equation*}
\end{definition}
\vspace{-0.5cm}
\subsection{Windowed Bernstein's inequality}
\begin{lemma}[Windowed Bernstein's inequality]\label{lemma:Bernstein}
For a function  $f: \mathbb{R} \to \mathbb{R}$, bandlimited to $[-\Omega, \Omega]$, the following inequality holds within a given window $\mathcal{T}$
\[
\left\|  \frac{df}{dt} \right\| _{\mathcal{T}} \leq \Omega\| f\|_{\mathcal{T}},
\]
where the norm $\| f \|^2_{\mathcal{T}}$ is defined as per Definition~\ref{norm}.

\end{lemma}
\begin{proof}
First, we define a unit pulse function over the window $\mathcal{T}$, with length $|\mathcal{T}|=t_{\text{end}}-t_{\text{start}}$, as follows
\begin{equation}\label{diff:win}
 1_{\mathcal{T}}(t) =
\begin{cases}
1 & \text{if } t \in \mathcal{T} \\
0 & \text{otherwise}
\end{cases}.   
\end{equation}
Considering $f'$ as the derivative of $f$ and $\mathcal{F}\{\cdot\}$ as the Fourier transform, the proof unfolds with the following steps

\vspace{-0.4cm}
\ifshort
{\small
\begin{align*}
     & \| f'\|_{\mathcal{T}}^2 \overset{(a)}{=} \| f'\cdot1_{\mathcal{T}}(t)\|^2
    \overset{(b)}{=} \frac{1}{2\pi}\| \mathcal{F}\{f'\}
    *  \mathcal{F}\{1_{\mathcal{T}}(t)\}\|^2 \\
    & \overset{(c)}{=} \frac{1}{2\pi}\| jw\mathcal{F}\{f\}
    *  \mathcal{F}\{1_{\mathcal{T}}(t)\}\|^2
     \overset{(d)}{\leq} \frac{\Omega^2}{2\pi}\\
     &\| \mathcal{F}\{f\}
    *  \mathcal{F}\{1_{\mathcal{T}}(t)\}\|^2
     \overset{(e)}{=} \Omega^2 \| f \cdot 1_{\mathcal{T}}(t) \|^2
     = \Omega^2 \| f \|_{\mathcal{T}}^2,
\end{align*}}
\else
\begin{align*}
     \| f'\|_{\mathcal{T}}^2 & \overset{(a)}{=} \| f'\cdot1_{\mathcal{T}}(t)\|^2\\
    & \overset{(b)}{=} \frac{1}{2\pi}\| \mathcal{F}\{f'\}
    *  \mathcal{F}\{1_{\mathcal{T}}(t)\}\|^2 \\
    & \overset{(c)}{=} \frac{1}{2\pi}\| jw\mathcal{F}\{f\}
    *  \mathcal{F}\{1_{\mathcal{T}}(t)\}\|^2 \\
    &\overset{(d)}{\leq} \frac{\Omega^2}{2\pi}\| \mathcal{F}\{f\}
    *  \mathcal{F}\{1_{\mathcal{T}}(t)\}\|^2 \\
    & \overset{(e)}{=} \Omega^2 \| f \cdot 1_{\mathcal{T}}(t) \|^2 \\
    & = \Omega^2 \| f \|_{\mathcal{T}}^2,
\end{align*}
\fi
\hspace{-0.19cm}where (a) follows from expanding the norm over the entire real line and restricting the derivative $f'$ to the window $\mathcal{T}$ using the pulse function $1_{\mathcal{T}}(t)$. (b) holds by applying Parseval's Theorem \cite{papoulis1967limits}, considering that time-domain multiplication corresponds to convolution in the frequency domain. (c) follows from the property that the Fourier transform of $f'$ is $j\omega \mathcal{F}{f}$. (d) is given since $f(t)$ is bandlimited within $[-\Omega, \Omega]$, \blue{so we have $\abs{\omega} \leq \Omega$. By the convolution property in the frequency domain, define $G(\omega,\tau) = \mathcal{F}\{f\}(\tau) \mathcal{F}\{1_{\mathcal{T}(t)}\}(\omega - \tau)$. Taking the squared magnitude (i.e., the squared norm in the frequency domain), we obtain
%\scriptsize{
\begin{multline*}
\small
    \abs{jw\mathcal{F}\{f\} *  \mathcal{F}\{1_{\mathcal{T}}(t)\}}^2\\= \abs{\int_{-\infty}^{\infty} j\tau G(w,\tau)d\tau}^2 = \abs{\int_{-\infty}^{\infty} \tau G(w,\tau)d\tau}^2.
\end{multline*}}%}
\blue{Note that the factor $j$ has $\abs{j}=1$ and hence does not affect the magnitude. Since $f(t)$ is bandlimited, the integration variable $\tau$ in the convolution is confined to $\abs{\tau} \leq \Omega$. Thus, the integrand can be directly bounded by $\Omega$, allowing us to factor out $\Omega^2$ from the norm.} (e) results from reversing the convolution-multiplication relationship, transitioning back to the time domain.
This completes the proof of the lemma.
\end{proof}
\vspace{-0.5cm}
\subsection{The adjoint operator}
\begin{lemma}\label{lemma: adjoint in window}
Recall operators $\mathcal{A}$ and $\mathcal{A^*}$ as defined in \eqref{Operator A} \eqref{adjoint operator}, respectively. $\mathcal{A}$ and $\mathcal{A^*}$ are adjoint operators in window $W_i$. That is 
\[
\langle\mathcal{A}x,y\rangle_{W_i} = \langle x,\mathcal{A^*}y\rangle_{W_i},
\]
where the inner product $\langle\cdot,\cdot\rangle_{W_i}$ is as defined in Definition~\ref{product}.
\end{lemma}
\begin{proof}
\vspace{-0.2cm}
\ifshort
{\small
\begin{align*}
   & \langle\mathcal{A}x,y\rangle_{W_i} \overset{(a)}{=}\langle\mathcal{A}x,y\cdot1_{W_i}\rangle \\
   & \overset{(b)}{=} \left\langle\sum_{n\in S_{W_i}} \int_{t_{n-1}}^{t_n} x(u)du\ g(t-\theta_n),y\cdot1_{W_i} \right\rangle  \\
  & \overset{(c)}{=}\sum_{n\in S_{W_i}} \int_{t_{n-1}}^{t_n} x(u)du \left\langle g(t-\theta_n),y\cdot1_{W_i}\right\rangle \\
   & \overset{(d)}{=} \sum_{n\in S_{W_i}} \int_{t_{n-1}}^{t_n} x(u) y(\theta_n)du
     \overset{(e)}{=} \sum_{n\in S_{W_i}} \langle x,1_{[t_{n-1}, t_n]}\rangle  y(\theta_n)du\\
   & \overset{(f)}{=}  \left\langle x,\sum_{n\in S_{W_i}}1_{[t_{n-1}, t_n]}y(\theta_n)\right\rangle
    \overset{(g)}{=} \left\langle\mathcal{G}x,\sum_{n\in S_{W_i}}1_{[t_{n-1}, t_n]}y(\theta_n)\right\rangle \\
    & \overset{(h)}{=} \left\langle x,\sum_{n\in S_{W_i}}\mathcal{G}1_{[t_{n-1}, t_n]}y(\theta_n)\right\rangle  = \langle x,\mathcal{A^*}y\rangle_{W_i},
\end{align*}}
\else
\begin{align*}
   & \langle\mathcal{A}x,y\rangle_{W_i} \overset{(a)}{=}\langle\mathcal{A}x,y\cdot1_{W_i}\rangle \\
   & \overset{(b)}{=} \left\langle\sum_{n\in \mathcal{Z}} \int_{t_{n-1}}^{t_n} x(u)du\ g(t-\theta_n),y\cdot1_{W_i} \right\rangle  \\
  & \overset{(c)}{=}\sum_{n\in \mathcal{Z}} \int_{t_{n-1}}^{t_n} x(u)du \left\langle g(t-\theta_n),y\cdot1_{W_i}\right\rangle \\
   & \overset{(d)}{=} \sum_{n\in S_{W_i}} \int_{t_{n-1}}^{t_n} x(u) y(\theta_n)du\\
      & \overset{(e)}{=} \sum_{n\in S_{W_i}} \langle x,1_{[t_{n-1}, t_n]}\rangle  y(\theta_n)du\\
   & \overset{(f)}{=}  \left\langle x,\sum_{n\in S_{W_i}}1_{[t_{n-1}, t_n]}y(\theta_n)\right\rangle \\
   & \overset{(g)}{=} \left\langle\mathcal{G}x,\sum_{n\in S_{W_i}}1_{[t_{n-1}, t_n]}y(\theta_n)\right\rangle \\
    & \overset{(h)}{=} \left\langle x,\sum_{n\in S_{W_i}}\mathcal{G}1_{[t_{n-1}, t_n]}y(\theta_n)\right\rangle  = \langle x,\mathcal{A^*}y\rangle_{W_i},
\end{align*}
\fi
\hspace{-0.15cm}where (a) follows by restricting $y$ to the window $W_i$ using the pulse function $1_{W_i}$, as defined in \eqref{diff:win}. (b) applies the definition of $\mathcal{A}x$ from \eqref{Operator A}. (c) uses the linearity of the inner product.(d)  follows from approximating the inner product as a convolution with a delta function. (e) expresses the integral as an inner product between $x$ and $1_{[t_{n-1}, t_n]}$. (f) again uses the linearity of the inner product. (g) uses the property $\mathcal{G}x = x$ for $2\Omega$-bandlimited $x$, since $\mathcal{F}\{g\}$ is a unit-amplitude low-pass filter over the same bandwidth. (h) uses linearity and the definition of $\mathcal{A^*}$.

This completes the proof.
\end{proof}

\vspace{-0.6cm}
\subsection{Proof of Lemma~\ref{key_lemma}}\label{proof:lem1}
\begin{proof}[Proof of Lemma~\ref{key_lemma}]
Recall the operator $\mathcal{A}$, as defined in \eqref{Operator A}. To establish an upper bound on the norm of the discrepancy between $x$ and $\mathcal{A}x$ over a time window $W_i$, we use techniques similar to those in \cite[Appendix B]{lazar2004perfect} for classical TEM, and in \cite{benedetto1994theory} for irregular sampling and grounded in frame theory, as shown in \cite{duffin1952class}. However, the authors in \cite{lazar2004perfect} (as is common in classical recovery schemes, e.g., \cite{lazar2003time,lazar2004time}) provide proof for the norm's upper bound over $\mathbb{R}$, whereas here we focus on a specific $i$-th time window $W_i$, working in the finite regime.
Let $\mathcal{A}^*$ denote the adjoint operator of $\mathcal{A}$, as given in~\eqref{adjoint operator} and proved in Lemma~\ref{lemma: adjoint in window}.
\off{
First, we define the operator $\mathcal{G}$ which maps an arbitrary function $x$ into a bandlimited function through $\mathcal{G}x = (g * x)$, where  $g(t) = sin(\Omega t)/\pi t$. Next, we provide the adjoint operator of $\mathcal{A}$, $\mathcal{A}^*$, defined by 
\[
\mathcal{A^*}x = \sum_{k \in \mathbb{Z}}x(\theta_k)\mathcal{G}1_{[t_{k}, t_{k+1}]},
\]
where $\mathcal{G}1_{[t_k, t_{k+1}]} $ is applying the operator $\mathcal{G}$  to a pulse function defined on the interval $[t_k, t_{k+1}]$.}
Thus, for the finite norm  $\| \cdot \|^2_{W_i}$ (see Lemma~\ref{lemma:Bernstein}), we have

\vspace{-0.3cm}
\ifshort
{\small
\begin{align*}
& \| x-\mathcal{A^*}x\|^2_{W_i}  = \| x - \sum_{n\in S_{W_i}}x(\theta_n)\mathcal{G}1_{[t_{n-1}, t_n]}\|^2_{W_i} \\
& \overset{(a)}{=}\| \mathcal{G}x - \sum_{n\in S_{W_i}}x(\theta_n)\mathcal{G}1_{[t_{n-1}, t_n]}\|^2_{W_i} \\
& \overset{(b)}{\leq} \|  x -\sum_{n\in S_{W_i}}x(\theta_n)1_{[t_{n-1}, t_n]}\|^2_{W_i}\\
&\overset{(c)}{=} \| \sum_{n\in S_{W_i}}[x -x(\theta_n)]1_{[t_{n-1}, t_n]}\|^2_{W_i} \\
&\blue{\overset{(d)}{\leq} \int_{t_a}^{t_b}\abs{\sum_{n\in S_{W_i}} [x -x(\theta_n)]1_{[t_{n-1}, t_n]}} ^2} \\
&\blue{\overset{(e)}{=} \int_{t_a}^{t_b}\sum_{n\in S_{W_i}} \abs{x -x(\theta_n)}^21_{[t_{n-1}, t_n]}} \\
& \overset{(f)}{=} \sum_{n\in S_{W_i}} \int_{t_{n-1}}^{t_n}\mid x(u) - x(\theta_n)\mid^2du\\
& \overset{(g)}{\leq} \sum_{n\in S_{W_i}} \frac{4}{\pi^2}(\theta_n-t_{n-1})^2\int_{t_{n-1}}^{\theta_n}\mid x'(u)\mid^2du\\
&\quad\quad + \frac{4}{\pi^2}(t_n-\theta_n)^2 \int_{\theta_n}^{t_n}\mid x'(u)\mid^2du,
\end{align*}}
\else
\begin{align*}
& \| x-\mathcal{A^*}x\|^2_{W_i}  = \| x - \sum_{n \in \mathbb{Z}}x(\theta_n)\mathcal{G}1_{[t_{n-1}, t_n]}\|^2_{W_i} \\
& \overset{(a)}{=}\| \mathcal{G}x - \sum_{n \in \mathbb{Z}}x(\theta_n)\mathcal{G}1_{[t_{n-1}, t_n]}\|^2_{W_i} \\
& \overset{(b)}{\leq} \|  x -\sum_{n \in \mathbb{Z}}x(\theta_n)1_{[t_{n-1}, t_n]}\|^2_{W_i}\\
&\overset{(c)}{=} \| \sum_{n \in \mathbb{Z}}[x -x(\theta_n)]1_{[t_{n-1}, t_n]}\|^2_{W_i} \\
&\overset{(d)}{\leq} \int_{t_a}^{t_b}\abs{\sum_{n\in S_{W_i}} [x -x(\theta_n)]1_{[t_{n-1}, t_n]}} ^2 \\
&\overset{(e)}{=} \int_{t_a}^{t_b}\sum_{n\in S_{W_i}} \abs{x -x(\theta_n)}^21_{[t_{n-1}, t_n]} \\
& \overset{(f)}{=} \sum_{n\in S_{W_i}} \int_{t_{n-1}}^{t_n}\mid x(u) - x(\theta_n)\mid^2du\\
& \overset{(g)}{\leq} \sum_{n\in S_{W_i}} \frac{4}{\pi^2}(\theta_n-t_{n-1})^2\int_{t_{n-1}}^{\theta_n}\mid x'(u)\mid^2du\\
&\quad\quad + \frac{4}{\pi^2}(t_n-\theta_n)^2 \int_{\theta_n}^{t_n}\mid x'(u)\mid^2du,
\end{align*}
\fi
\hspace{-0.15cm}where, \blue{(a) follows since $x$ is bandlimited to $2\Omega$, and the operator $\mathcal{G}$ is defined in Appendix \ref{Appendix a}, Definition \ref{def: convolotion operator}. The function $g$, as given in Definition \ref{def: convolotion operator}, is the ideal low-pass filter with the same bandwidth. Consequently, the convolution $ \mathcal{G}x = (g * x)$ acts as an identity operator in this bandlimited space, resulting in  $\mathcal{G}x = x$.
(b) follows since $g$ is a smoothing low-pass filter, convolution with it attenuates high-frequency components, thereby reducing the norm. 
(c) follows from the definition of the norm over $W_i$, where the sum of the indicator functions $1_{[t_{n-1}, t_n]}$, over $n \in S_{W_i}$,  partitions the interval into disjoint segments, preserving the signal $x$ within the window $W_i$.
(d) follows from explicitly expressing the squared norm as an integral over the window $W_i$.
(e) follows by rewriting the squared absolute value term inside the summation. Since each term in the summation is multiplied by the indicator function $1_{[t_{n-1}, t_n]}$, which isolates contributions from non-overlapping time intervals, the squared absolute value can be applied individually to each term inside the summation.
(f) follows by interchanging the sum and the integral due to the linearity of integration. Additionally, the indicator function $1_{[t_{n-1}, t_n]}$ restricts the integration domain to disjoint intervals.
(g) holds by applying Wirtinger’s inequality \cite{lazar2004perfect}.} 

We do note that since for any $n \in S_{W_i}$
\begin{equation*}
(\theta_n-t_{n-1})^2 = (t_n-\theta_n)^2 =\frac{(t_n-t_{n-1})^2}{4},
\end{equation*}
we obtain
\ifshort
{\small
\begin{align*}
    & \| x-\mathcal{A^*}x\|^2_{W_i}  \\
    & \leq \sum_{n\in S_{W_i}}\frac{1}{\pi^2}(t_n-t_{n-1})^2\int_{t_{n-1}}^{t_n}\mid x'(u)\mid^2du\\
    & \leq \frac{1}{\pi^2}\left(\max_{n\in S_{W_i}} (T_n)\right)^2\sum_{n\in S_{W_i}} \int_{t_{n-1}}^{t_n}\mid x'(u)\mid^2du \\
    & =  \frac{1}{\pi^2}\left(\max_{n\in S_{W_i}} (T_n)\right)^2\| x' \|_{W_i}^2
     \overset{(h)}{\leq} \frac{1}{\pi^2}\left(\max_{n\in S_{W_i}}(T_n)\right)^2\Omega^2\| x\|_{W_i}^2, %= r^2_{w_i}\|{x}\|^2_{W_i},
\end{align*}}
\else
\begin{align*}
    & \| x-\mathcal{A^*}x\|^2_{W_i}  \\
    & \leq \sum_{n\in S_{W_i}}\frac{1}{\pi^2}(t_n-t_{n-1})^2\int_{t_{n-1}}^{t_n}\mid x'(u)\mid^2du\\
    & \leq \frac{1}{\pi^2}\left(\max_{n\in S_{W_i}} (T_n)\right)^2\sum_{n\in S_{W_i}} \int_{t_{n-1}}^{t_n}\mid x'(u)\mid^2du \\
    & =  \frac{1}{\pi^2}\left(\max_{n\in S_{W_i}} (T_n)\right)^2\| x' \|_{W_i}^2\\
    & \overset{(h)}{\leq} \frac{1}{\pi^2}\left(\max_{n\in S_{W_i}}(T_n)\right)^2\Omega^2\| x\|_{W_i}^2, %= r^2_{w_i}\|{x}\|^2_{W_i},
\end{align*}
\fi
%where (f) follows by expanding the norm of $x'$ over $\mathbb{R}$ and (g) by applying Bernstein's inequality \cite{}. 
where (h) holds by applying Windowed Bernstein's inequality as given on Lemma~\ref{lemma:Bernstein}.

Finally, substituting  $\max_{n\in S_{W_i}} (T_n)\leq r_{w_i}\frac{\pi}{\Omega}$, for $r_{w_i} = \max_{n \in  S_{W_i}}\{r_{a_n}\}$ and  $r_{a_n}$ as given in \eqref{r_a general}, we have
\begin{equation*}
\| x-\mathcal{A}x\|_{W_i} \ < r_{w_i}\|{x}\|_{W_i}.
\end{equation*}
This completed the lemma proof.
\end{proof}
\vspace{-0.5cm}
\subsection{Proof of Lemma~\ref{key_lemma2}}\label{proof:lem2}
\begin{proof}[Proof of Lemma~\ref{key_lemma2}]
The recovered signal for segment $W_i$, $x_{L_i}$, is defined by \eqref{x_l}, with $l=L_i$. We evaluate the norm difference between the original signal $x$ and the recovered signal $x_{L_i}$ over the segment $W_i$ as follows

\vspace{-0.3cm}
\ifshort
{\small
\begin{multline*}
\hspace{-0.4cm}\| x-x_{L_i}\|_{W_i} =\| \sum_{n \geq L_i+1}(I-\mathcal{A})^{n}\mathcal{A}x\|_{W_i}
= \| (I-\mathcal{A})^{L_i+1}\\\sum_{n \in \mathbb{N}}(I-\mathcal{A})^{n}\mathcal{A}x\|_{W_i}
= \| (I-\mathcal{A})^{L_i+1}\mathcal{A}^{-1}\mathcal{A}x\|_{W_i}\\
\quad = \| (I-\mathcal{A})\|_{W_i}^{L_i+1}\| x\|_{W_i}.
\end{multline*}}
\else
\begin{align*}
\| x-x_{L_i}\|_{W_i} & =\| \sum_{n \geq L_i+1}(I-\mathcal{A})^{n}\mathcal{A}x\|_{W_i} \\
&= \| (I-\mathcal{A})^{L_i+1}\sum_{n \in \mathbb{N}}(I-\mathcal{A})^{n}\mathcal{A}x\|_{W_i}\\
&= \| (I-\mathcal{A})^{L_i+1}\mathcal{A}^{-1}\mathcal{A}x\|_{W_i}\\
& = \| (I-\mathcal{A})\|_{W_i}^{L_i+1}\| x\|_{W_i}.
\end{align*}
\fi
This completed the lemma proof.
\end{proof}
\vspace{-0.5cm}
\subsection{Key Mathematical Tools}

 \subsubsection{Observation 1}\label{math_obs1} Assume the quantization error $d_n = (\tilde{T}_n - T_n)$ is a sequence of i.i.d. random variables in $[ -\Delta_i/2,  \Delta_i/2]$. Then $ \mathbb{E}[\epsilon_n\epsilon_m] =  (\frac{\kappa\delta}{T_n})^2 \frac{\Delta_i^2}{12}\delta_{n,m}$.
\begin{proof}
\ifshort
{\small
\begin{align*}
     \epsilon_n & =  (\kappa \delta -b_n\tilde{T}_n) - \int_{\tilde{t}_{n-1}}^{\tilde{t}_{n}} x(u)du\\
     & = \kappa \delta -b_nT_n + b_nT_n -b_n\tilde{T}_n - \int_{\tilde{t}_{n-1}}^{\tilde{t}_{n}} x(u)du\\
     &=\int_{t_{n-1}}^{t_{n}} x(u)du - \int_{\tilde{t}_{n-1}}^{\tilde{t}_{n}} x(u)du -b_n(\tilde{T}_n - T_n) \\
     & \overset{a}{=} x(\zeta_n)T_n - x(\hat{\zeta}_n)\tilde{T}_n-b_n(\tilde{T}_n - T_n) \\
     & \overset{b}{\simeq} (-x(\zeta_n)-b_n)(\tilde{T}_n - T_n)
      \overset{c}{=}  -\frac{\kappa\delta}{T_n}(\tilde{T}_n - T_n),
\end{align*}}
\else
\begin{align*}
     \epsilon_n & =  (\kappa \delta -b_n\tilde{T}_n) - \int_{\tilde{t}_{n-1}}^{\tilde{t}_{n}} x(u)du\\
     & = \kappa \delta -b_nT_n + b_nT_n -b_n\tilde{T}_n - \int_{\tilde{t}_{n-1}}^{\tilde{t}_{n}} x(u)du\\
     &=\int_{t_{n-1}}^{t_{n}} x(u)du - \int_{\tilde{t}_{n-1}}^{\tilde{t}_{n}} x(u)du -b_n(\tilde{T}_n - T_n) \\
     & \overset{a}{=} x(\zeta_n)T_n - x(\hat{\zeta}_n)\tilde{T}_n-b_n(\tilde{T}_n - T_n) \\
     & \overset{b}{\simeq} (-x(\zeta_n)-b_n)(\tilde{T}_n - T_n)\\
     & \overset{c}{=}  -\frac{\kappa\delta}{T_n}(\tilde{T}_n - T_n).
\end{align*}
\fi
\hspace{-0.2cm}where (a) follows from the mean value theorem. (b) follows from the approximation used in \cite[Appendix C]{lazar2004perfect}. Where  $\zeta_n \in (t_{n-1}, t_{n})$ and $\hat{\zeta}_n \in (\tilde{t}_{n-1}, \tilde{t}_n)$, for sufficiently small $\Delta_i$, we get $\hat{\zeta}_n \simeq \zeta_n $. (c) follows because $x(\zeta_n) =\frac{1}{T_n}\int_{t_{n-1}}^{t_n} x(u)du = \frac{\kappa\delta}{T_n}-b_n$.

Let $d_n =(\tilde{T}_n - T_n)$, this leads to
\[
\mathbb{E}[\epsilon_n\epsilon_m] = \frac{\kappa\delta}{T_n}\frac{\kappa\delta}{T_m}\mathbb{E}[d_nd_m] \overset{(a)}{=} \left(\frac{\kappa\delta}{T_n}\right)^2 \frac{\Delta_i^2}{12}\delta_{n,m},
\]
where (a) follows since the quantization error $d_n = (\tilde{T}_n - T_n)$ is a sequence of i.i.d. random variables uniformly distributed in $[ -\Delta_i/2,  \Delta_i/2]$. Therefore, $\mathbb{E}[d_nd_m] = \frac{\Delta_i^2}{12}\delta_{n,m}$.
\end{proof}
\off{
\subsubsection{Observation 2}: For dynamic quantization with quantization error $d_n =(\tilde{T}_n - T_n)$ is a sequence of i.i.d. random variables on $[ -\Delta_i/2,  \Delta_i/2]$ 
$ \mathbb{E}[\epsilon_n\epsilon_m] =  (\frac{\kappa\delta}{T_n})^2 \frac{\Delta_i^2}{12}\delta_{n,m}$
}
\subsubsection{Observation 2}\label{math_obs2} The integral $\int_{t_{\text{start}}}^{t_{\text{end}}}g^2(t-\eta)dt$, for any $\eta$,  over the interval $\mathcal{T} = [t_{\text{start}}, t_{\text{end}}]$ is upper bounded by $\frac{\Omega}{\pi}$.
 \begin{proof}
 \ifshort
 {\small
    \begin{align*}
    \int_{t_{\text{start}}}^{t_{\text{end}}} g^2(t - \eta)dt & \overset{(a)}{\leq} \int_{-\infty}^{\infty} g^2(t - \eta)dt\\
    &\overset{(b)}{=} \frac{1}{2\pi} \int_{-\infty}^{\infty} 1_{[-\Omega,\Omega]}dw =\frac{\Omega}{\pi},
    \end{align*}}
 \else
    \begin{align*}
    \int_{t_{\text{start}}}^{t_{\text{end}}} g^2(t - \eta)dt & \overset{(a)}{\leq} \int_{-\infty}^{\infty} g^2(t - \eta)dt\\
    &\overset{(b)}{=} \frac{1}{2\pi} \int_{-\infty}^{\infty} 1_{[-\Omega,\Omega]}dw =\frac{\Omega}{\pi},
    \end{align*}
\fi
\hspace{-0.2cm}where (a) follows from the fact that $g^2$ is a non-negative function, and equality (b) is derived by applying Parseval's theorem. Note that if $\eta \in \mathcal{T}$, then this bound is tighter.
 \end{proof}
\subsubsection{Observation 3}\label{math_obs3}: For any positive numbers $\{T_i\}_{i=1}^{L_i}$
 \[
 \frac{1}{\overline{T}_n}  \overline{\left(\frac{1}{T_n}\right)^2} \leq \overline{\left(\frac{1}{T_n}\right)^3}
 \]
 \begin{proof}
     Note that
  \[
    \Bigl(\sum_{i=1}^L T_i\Bigr)\Bigl(\sum_{j=1}^L \tfrac{1}{T_j^3}\Bigr)
    \;=\;
    \sum_{i=1}^L \sum_{j=1}^L \frac{T_i}{T_j^3}
    \;\;\ge\;\;
    \sum_{i=1}^L \frac{T_i}{T_i^3}
    \;=\;
    \sum_{i=1}^L \frac{1}{T_i^2},
  \]
  the inequality is because each inner sum over $j$ contains the term $j = i$. Dividing by
  $(\sum_{i=1}^L T_i)\,L_i$ completes the proof. 
 \end{proof}
\fi

\section{Conclusion}

\textcolor{black}{This paper introduces an Adaptive Integrate-and-Fire Time Encoding Machine (AIF-TEM) that extends classical IF-TEM by dynamically adapting the bias according to temporal signal amplitude variations. The proposed design addresses the oversampling inefficiency of fixed-bias IF-TEM, where the bias is independent of the temporal signal amplitude\off{envelope} and may lead to\off{ redundant firing activity} excessive sampling density when the input amplitude is low.
We established analytical conditions for correct adaptive operation, including the MAP condition (i.e., maintaining the adaptive bias slightly above, yet close to, the temporal maximum amplitude of the input signal) and the segment-wise Nyquist constraint (i.e., maximum Nyquist ratio smaller than one), ensuring stable encoding and reconstruction. Closed-form bounds were derived for reconstruction distortion and quantization error, characterizing their dependence on the adaptive bias, sampling density, and segment length.
The analysis and simulations demonstrate that the proposed adaptive strategy significantly reduces oversampling rate, mean-squared reconstruction error, and total bit usage compared to fixed-bias IF-TEM and periodic sampling under identical signal conditions. These results show that a simple causal MAP implementation is sufficient to achieve consistent rate–distortion improvements while preserving reconstruction guarantees for bandlimited inputs.
Future work will consider alternative bias estimation strategies and extensions to broader signal classes and hardware-aware implementations, including mechanisms that further enforce the correct MAP operation condition under model uncertainty or transient estimation errors.}

\bibliographystyle{IEEEtran}
\bibliography{Ref_short}

@article{tarnopolsky2022compressed,
  title={Compressed IF-TEM: Time Encoding Analog-To-Digital Compression},
  author={Tarnopolsky, Saar and Naaman, Hila and Eldar, Yonina C and Cohen, Alejandro},
  journal={arXiv preprint arXiv:2210.17544},
  year={2022}
}

@inproceedings{koscielnik2007designing,
  title={Designing time-to-digital converter for asynchronous ADCs},
  author={Koscielnik, Dariusz and Miskowicz, Marek},
  booktitle={2007 IEEE Design and Diagnostics of Electronic Circuits and Systems},
  pages={1--6},
  year={2007},
  organization={IEEE}
}

@article{lazar2004perfect,
  title={Perfect recovery and sensitivity analysis of time encoded bandlimited signals},
  author={Lazar, Aurel A and T{\'o}th, L{\'a}szl{\'o} T},
  journal={IEEE Transactions on Circuits and Systems I: Regular Papers},
  volume={51},
  number={10},
  pages={2060--2073},
  year={2004},
  publisher={IEEE}
}

@article{koscielnik2015sample,
  title={Sample-and-hold asynchronous sigma-delta time encoding machine},
  author={Ko{\'s}cielnik, Dariusz and Rzepka, Dominik and Szyduczy{\'n}ski, Jakub},
  journal={IEEE Transactions on Circuits and Systems II: Express Briefs},
  volume={63},
  number={4},
  pages={366--370},
  year={2015},
  publisher={IEEE}
}

@inproceedings{lazar2005time,
  title={Time encoding of bandlimited signals, an overview},
  author={Lazar, Aurel A and Simonyi, Erno K and T{\'o}th, L{\'a}szl{\'o} T},
  booktitle={Proceedings of conference on telecommunication systems, modeling and analysis},
  year={2005},
  organization={Citeseer}
}

@article{lazar2004time,
  title={Time encoding with an integrate-and-fire neuron with a refractory period},
  author={Lazar, Aurel A},
  journal={Neurocomputing},
  volume={58},
  pages={53--58},
  year={2004},
  publisher={Elsevier}
}

@article{naaman2021time,
  title={Time-Based Quantization for {FRI} and Bandlimited signals},
  author={Naaman, Hila and Mulleti, Satish and Eldar, Yonina C and Cohen, Alejandro},
  journal={in 2022 30th European Signal Processing Conference (EUSIPCO)},
  pages={2241--2245},
  year={2022},
  organization={IEEE}
}

@inproceedings{rastogi2011integrate,
  title={Integrate and fire circuit as an ADC replacement},
  author={Rastogi, Manu and Alvarado, Alexander Singh and Harris, John G and Principe, Jose C},
  booktitle={2011 IEEE International Symposium of Circuits and Systems (ISCAS)},
  pages={2421--2424},
  year={2011},
  organization={IEEE}
}

@article{papoulis1967limits,
  title={Limits on bandlimited signals},
  author={Papoulis, Athanasios},
  journal={Proceedings of the IEEE},
  volume={55},
  number={10},
  pages={1677--1686},
  year={1967},
  publisher={IEEE}
}

@article{andrew2003spiking,
  title={Spiking neuron models: Single neurons, populations, plasticity},
  author={Andrew, Alex M},
  journal={Kybernetes},
  year={2003},
  publisher={Emerald Group Publishing Limited}
}

@inproceedings{lazar2003time,
  title={Time encoding and perfect recovery of bandlimited signals},
  author={Lazar, Aurel A and T{\'o}th, L{\'a}szl{\'o} T},
  booktitle={2003 IEEE International Conference on Acoustics, Speech, and Signal Processing, 2003. Proceedings.(ICASSP'03).},
  volume={6},
  pages={VI--709},
  year={2003},
  organization={IEEE}
}

@book{miskowicz2018event,
  title={Event-based control and signal processing},
  author={Miskowicz, Marek},
  year={2018},
  publisher={CRC press}
}

@article{duffin1952class,
  title={A class of nonharmonic Fourier series},
  author={Duffin, Richard J and Schaeffer, Albert C},
  journal={Transactions of the American Mathematical Society},
  volume={72},
  number={2},
  pages={341--366},
  year={1952}
}

@inproceedings{kinget2005robustness,
  title={On the robustness of an analog VLSI implementation of a time encoding machine},
  author={Kinget, Peter R and Lazar, Aurel A and T{\'o}th, L{\'a}szl{\'o} T},
  booktitle={2005 IEEE International Symposium on Circuits and Systems},
  pages={4221--4224},
  year={2005},
  organization={IEEE}
}

@inproceedings{wei2006asynchronous,
  title={An asynchronous delta-sigma converter implementation},
  author={Wei, Dazhi and Garg, Vaibhav and Harris, John G},
  booktitle={2006 IEEE International Symposium on Circuits and Systems (ISCAS)},
  pages={4--pp},
  year={2006},
  organization={IEEE}
}

@book{antoniou2006digital,
  title={Digital signal processing},
  author={Antoniou, Andreas},
  year={2006},
  publisher={McGraw-Hill}
}

@article{ryu2021time,
  title={A time-based pipelined ADC using integrate-and-fire multiplying-DAC},
  author={Ryu, Sigang and Park, Chan Young and Kim, Wooryeol and Son, Seuk and Kim, Jaeha},
  journal={IEEE Transactions on Circuits and Systems I: Regular Papers},
  volume={68},
  number={7},
  pages={2876--2889},
  year={2021},
  publisher={IEEE}
}

@inproceedings{rudresh2020time,
  title={A time-based sampling framework for finite-rate-of-innovation signals},
  author={Rudresh, Sunil and Kamath, Abijith Jagannath and Seelamantula, Chandra Sekhar},
  booktitle={ICASSP 2020-2020 IEEE International Conference on Acoustics, Speech and Signal Processing (ICASSP)},
  pages={5585--5589},
  year={2020},
  organization={IEEE}
}

@article{adam2020sampling,
  title={Sampling and reconstruction of bandlimited signals with multi-channel time encoding},
  author={Adam, Karen and Scholefield, Adam and Vetterli, Martin},
  journal={IEEE Transactions on Signal Processing},
  volume={68},
  pages={1105--1119},
  year={2020},
  publisher={IEEE}
}

@inproceedings{chen2006asynchronous,
  title={Asynchronous biphasic pulse signal coding and its CMOS realization},
  author={Chen, Du and Li, Yuan and Xu, Dongming and Harris, John G and Principe, Jos{\'e} Carlos},
  booktitle={2006 IEEE International Symposium on Circuits and Systems (ISCAS)},
  pages={4--pp},
  year={2006},
  organization={IEEE}
}

@article{florescu2022time,
  title={Time encoding via unlimited sampling: Theory, algorithms and hardware validation},
  author={Florescu, Dorian and Bhandari, Ayush},
  journal={IEEE Transactions on Signal Processing},
  volume={70},
  pages={4912--4924},
  year={2022},
  publisher={IEEE}
}

@article{nyquist1928certain,
  title={Certain topics in telegraph transmission theory},
  author={Nyquist, Harry},
  journal={Transactions of the American Institute of Electrical Engineers},
  volume={47},
  number={2},
  pages={617--644},
  year={1928},
  publisher={IEEE}
}

@article{welford1962note,
  title={Note on a method for calculating corrected sums of squares and products},
  author={Welford, BP},
  journal={Technometrics},
  volume={4},
  number={3},
  pages={419--420},
  year={1962},
  publisher={Taylor \& Francis}
}

@inproceedings{omar2024adaptive,
  title={{Adaptive Integrate-and-Fire Time Encoding Machine}},
  author={Omar, Aseel and Cohen, Alejandro},
  booktitle={2024 32nd European Signal Processing Conference (EUSIPCO)},
  pages={2442--2446},
  year={2024},
  organization={IEEE}
}

@incollection{benedetto1994theory,
  title={Theory and practice of irregular sampling},
  author={Benedetto, J and Frazier, M},
  booktitle={Wavelets: Mathematics and Applications},
  pages={305--363},
  year={1994},
  publisher={CRC}
}

@article{florescu2015novel,
  title={A novel reconstruction framework for time-encoded signals with integrate-and-fire neurons},
  author={Florescu, Dorian and Coca, Daniel},
  journal={Neural computation},
  volume={27},
  number={9},
  pages={1872--1898},
  year={2015},
  publisher={MIT Press}
}

@article{gontier2014sampling,
  title={Sampling based on timing: Time encoding machines on shift-invariant subspaces},
  author={Gontier, David and Vetterli, Martin},
  journal={Applied and Computational Harmonic Analysis},
  volume={36},
  number={1},
  pages={63--78},
  year={2014},
  publisher={Elsevier}
}

@article{naaman2022fri,
  title={{FRI-TEM}: Time encoding sampling of finite-rate-of-innovation signals},
  author={Naaman, Hila and Mulleti, Satish and Eldar, Yonina C},
  journal={IEEE Transactions on Signal Processing},
  volume={70},
  pages={2267--2279},
  year={2022},
  publisher={IEEE}
}

@inproceedings{lazar2006real,
  title={A real-time algorithm for time decoding machines},
  author={Lazar, Aurel A and Simonyi, Ern{\H{o}} K and T{\'o}th, L{\'a}szl{\'o} T},
  booktitle={2006 14th European Signal Processing Conference},
  pages={1--5},
  year={2006},
  organization={IEEE}
}

@article{alexandru2019reconstructing,
  title={Reconstructing classes of non-bandlimited signals from time encoded information},
  author={Alexandru, Roxana and Dragotti, Pier Luigi},
  journal={IEEE Transactions on Signal Processing},
  volume={68},
  pages={747--763},
  year={2019},
  publisher={IEEE}
}

@article{thao2023bandlimited,
  title={Bandlimited signal reconstruction from leaky integrate-and-fire encoding using POCS},
  author={Thao, Nguyen T and Rzepka, Dominik and Mi{\'s}kowicz, Marek},
  journal={IEEE Transactions on Signal Processing},
  volume={71},
  pages={1464--1479},
  year={2023},
  publisher={IEEE}
}

@article{lazar2008overcomplete,
  title={An overcomplete stitching algorithm for time decoding machines},
  author={Lazar, Aurel A and Simonyi, Ern{\"o} K and T{\'o}th, L{\'a}szl{\'o} T},
  journal={IEEE Transactions on Circuits and Systems I: Regular Papers},
  volume={55},
  number={9},
  pages={2619--2630},
  year={2008},
  publisher={IEEE}
}

@inproceedings{lazar2005fast,
  title={Fast recovery algorithms for time encoded bandlimited signals},
  author={Lazar, Aurel A and Simonyi, Ern{\"o} K and T{\'o}th, L{\'a}szl{\'o} T},
  booktitle={Proceedings.(ICASSP'05). IEEE International Conference on Acoustics, Speech, and Signal Processing, 2005.},
  volume={4},
  pages={iv--237},
  year={2005},
  organization={IEEE}
}

@article{liu2001spike,
  title={Spike-frequency adaptation of a generalized leaky integrate-and-fire model neuron},
  author={Liu, Ying-Hui and Wang, Xiao-Jing},
  journal={Journal of computational neuroscience},
  volume={10},
  pages={25--45},
  year={2001},
  publisher={Springer}
}

@article{brette2005adaptive,
  title={Adaptive exponential integrate-and-fire model as an effective description of neuronal activity},
  author={Brette, Romain and Gerstner, Wulfram},
  journal={Journal of neurophysiology},
  volume={94},
  number={5},
  pages={3637--3642},
  year={2005},
  publisher={American Physiological Society}
}

@article{kwon2021low,
  title={{Low-power adaptive integrate-and-fire neuron circuit using positive feedback FET Co-Integrated with CMOS}},
  author={Kwon, Min-Woo and Park, Kyungchul and Park, Byung-Gook},
  journal={IEEE Access},
  volume={9},
  pages={159925--159932},
  year={2021},
  publisher={IEEE}
}

@article{millner2010vlsi,
  title={{A VLSI implementation of the adaptive exponential integrate-and-fire neuron model}},
  author={Millner, Sebastian and Gr{\"u}bl, Andreas and Meier, Karlheinz and Schemmel, Johannes and Schwartz, Marc-Olivier},
  journal={Advances in neural information processing systems},
  volume={23},
  year={2010}
}

@article{fourcaud2003spike,
  title={How spike generation mechanisms determine the neuronal response to fluctuating inputs},
  author={Fourcaud-Trocm{\'e}, Nicolas and Hansel, David and Van Vreeswijk, Carl and Brunel, Nicolas},
  journal={Journal of neuroscience},
  volume={23},
  number={37},
  pages={11628--11640},
  year={2003},
  publisher={Society for Neuroscience}
}

@article{kikkert1975asynchronous,
  title={ASYNCHRONOUS DELTA SIGMA MODULATION.},
  author={Kikkert, CJ and PJ, MILLER},
  year={1975}
}

@inproceedings{ozols2013amplitude,
  title={Amplitude adaptive asynchronous Sigma-delta modulator},
  author={Ozols, Kaspars and Shavelis, Rolands and Greitans, Modris},
  booktitle={2013 8th International Symposium on Image and Signal Processing and Analysis (ISPA)},
  pages={467--470},
  year={2013},
  organization={IEEE}
}

@inproceedings{ozols2016amplitude,
  title={Amplitude adaptive {ASDM} without envelope encoding},
  author={Ozols, Kaspars and Shavelis, Rolands},
  booktitle={2016 24th European Signal Processing Conference (EUSIPCO)},
  pages={165--169},
  year={2016},
  organization={IEEE}
}

@inproceedings{shavelis2017amplitude,
  title={Amplitude adaptive {ASDM} circuit},
  author={Shavelis, Rolands and Ozols, Kaspars and Greitans, Modris},
  booktitle={2017 3rd International Conference on Event-Based Control, Communication and Signal Processing (EBCCSP)},
  pages={1--8},
  year={2017},
  organization={IEEE}
}

@article{xie2021robust,
  title={Robust autoregression with exogenous input model for system identification and predicting},
  author={Xie, Jiaxin and Li, Cunbo and Li, Ning and Li, Peiyang and Wang, Xurui and Gao, Dongrui and Yao, Dezhong and Xu, Peng and Yin, Gang and Li, Fali},
  journal={Electronics},
  volume={10},
  number={6},
  pages={755},
  year={2021},
  publisher={MDPI}
}

@book{chui2017kalman,
  title={Kalman filtering},
  author={Chui, Charles K and Chen, Guanrong and others},
  year={2017},
  publisher={Springer}
}

@article{cioffi1984fast,
  title={Fast, recursive-least-squares transversal filters for adaptive filtering},
  author={Cioffi, John and Kailath, Thomas},
  journal={IEEE Transactions on Acoustics, Speech, and Signal Processing},
  volume={32},
  number={2},
  pages={304--337},
  year={1984},
  publisher={IEEE}
}

@inproceedings{koscielnik2011natural,
  title={Natural compression and expansion characteristics of asynchronous sigma-delta {ADC}},
  author={Ko{\'s}cielnik, Dariusz and Mi{\'s}kowicz, Marek and Jab{\l}eka, Marek},
  booktitle={2011 5th International Conference on Signal Processing and Communication Systems (ICSPCS)},
  pages={1--8},
  year={2011},
  organization={IEEE}
}

@inproceedings{florescu2023model,
  title={Model-Driven Quantization for Time Encoding Machines},
  author={Florescu, Dorian},
  booktitle={2023 International Conference on Sampling Theory and Applications (SampTA)},
  pages={1--5},
  year={2023},
  organization={IEEE}
}

@article{kozmin2008level,
  title={Level-crossing {ADC} performance evaluation toward ultrasound application},
  author={Kozmin, Kirill and Johansson, Jonny and Delsing, Jerker},
  journal={IEEE Transactions on Circuits and Systems I: Regular Papers},
  volume={56},
  number={8},
  pages={1708--1719},
  year={2008},
  publisher={IEEE}
}

@article{saeed2021evaluation,
  title={{Evaluation of level-crossing ADCs for event-driven ECG classification}},
  author={Saeed, Maryam and Wang, Qingyuan and M{\"a}rtens, Olev and Larras, Benoit and Frapp{\'e}, Antoine and Cardiff, Barry and John, Deepu},
  journal={IEEE Transactions on Biomedical Circuits and Systems},
  volume={15},
  number={6},
  pages={1129--1139},
  year={2021},
  publisher={IEEE}
}

@inproceedings{akopyan2006level,
  title={A level-crossing flash asynchronous analog-to-digital converter},
  author={Akopyan, Filipp and Manohar, Rajit and Apsel, Alyssa B},
  booktitle={12th IEEE International Symposium on Asynchronous Circuits and Systems (ASYNC'06)},
  pages={11--pp},
  year={2006},
  organization={IEEE}
}

@inproceedings{guo2024ed,
  title={{ED-FreEst: Event-Driven Frequency Estimation}},
  author={Guo, Ruiming and Bhandari, Ayush},
  booktitle={2024 32nd European Signal Processing Conference (EUSIPCO)},
  pages={867--871},
  year={2024},
  organization={IEEE}
}

@inproceedings{florescu2023time,
  title={Time encoding of sparse signals with flexible filters},
  author={Florescu, Dorian and Bhandari, Ayush},
  booktitle={2023 International Conference on Sampling Theory and Applications (SampTA)},
  pages={1--5},
  year={2023},
  organization={IEEE}
}

@book{oppenheim1999discrete,
  title={Discrete-time signal processing},
  author={Oppenheim, Alan V},
  year={1999},
  publisher={Pearson Education India}
}

@article{harris2005use,
  title={On the use of windows for harmonic analysis with the discrete Fourier transform},
  author={Harris, Fredric J},
  journal={Proceedings of the IEEE},
  volume={66},
  number={1},
  pages={51--83},
  year={2005},
  publisher={IEEE}
}

@article{davies2018loihi,
  title={Loihi: A neuromorphic manycore processor with on-chip learning},
  author={Davies, Mike and Srinivasa, Narayan and Lin, Tsung-Han and Chinya, Gautham and Cao, Yongqiang and Choday, Sri Harsha and Dimou, Georgios and Joshi, Prasad and Imam, Nabil and Jain, Shweta and others},
  journal={IEEE {M}icro},
  volume={38},
  number={1},
  pages={82--99},
  year={2018},
  publisher={IEEE}
}

@article{stein1965theoretical,
  title={A theoretical analysis of neuronal variability},
  author={Stein, Richard B},
  journal={Biophysical journal},
  volume={5},
  number={2},
  pages={173--194},
  year={1965},
  publisher={Elsevier}
}

@article{mekel2025self,
  title={Self-Calibrating Integrate-and-Fire Time Encoding Machine},
  author={Mekel, Maya and Karp, Vered and Mulleti, Satish and Cohen, Alejandro},
  journal={Accepted for publication in IEEE ICASSP 2026. ArXiv preprint arXiv:2509.10831},
  year={2025}
}

@article{sankar2007analysis,
  title={Analysis of integrator nonlinearity in a class of continuous-time delta--sigma modulators},
  author={Sankar, Prabu and Pavan, Shanthi},
  journal={IEEE Transactions on Circuits and Systems II: Express Briefs},
  volume={54},
  number={12},
  pages={1125--1129},
  year={2007},
  publisher={IEEE}
}

@book{hardy1952inequalities,
  title={Inequalities},
  author={Hardy, Godfrey Harold and Littlewood, John Edensor and P{\'o}lya, George},
  year={1952},
  publisher={Cambridge university press}
}

\ifpagelimit
%\clearpage
\appendix
\ifpagelimit\else
\appendices
\fi

\ifpagelimit
\subsection{Definitions}\label{Appendix a}
\begin{definition}\label{def: convolotion operator}
     The operator $\mathcal{G}$ maps an arbitrary function $x$ into a bandlimited function via $\mathcal{G}x = (g * x)$, where $*$ denotes convolution and $g(t) = \sin(\Omega t)/\pi t$, known as the sinc function.
\end{definition}

\begin{definition}
     The operator $\mathcal{A}^*$ is defined by 
    \begin{equation}\label{adjoint operator}
         \mathcal{A^*}x = \sum_{n \in \mathbb{Z}}x(\theta_n)\mathcal{G}1_{[t_{n-1}, t_n]},
    \end{equation}
    where $\mathcal{G}1_{[t_{n-1}, t_n]}$ applies the operator $\mathcal{G}$ to a pulse function defined over the interval $[t_{n-1}, t_n]$.
    \end{definition}
\begin{definition}\label{norm}
    Let $f: \mathbb{R} \to \mathbb{R}$ and and let $\mathcal{T}$ be a general time window defined by $\mathcal{T} = [t_{\text{start}}, t_{\text{end}}]$, with $t_{\text{start}}$ and $t_{\text{end}}$ being the start and end points of the window, respectively. The norm $\| f \|^2_{\mathcal{T}}$ is given by
\begin{equation*}
    \| f \|^2_{\mathcal{T}} = \int_{t_{\text{start}}}^{t_{\text{end}}} \abs{f(u)}^2du.
\end{equation*}
\end{definition}

\begin{definition}\label{product}
    Let $f: \mathbb{R} \to \mathbb{R}, g: \mathbb{R} \to \mathbb{R}$, the inner product $\langle f,g\rangle_{\mathcal{T}}$ is defined by 
    \begin{equation*}
    \langle f,g\rangle_{\mathcal{T}} = \int_{t_{\text{start}}}^{t_{\text{end}}} f(u)g(u)du.
\end{equation*}
\end{definition}
\vspace{-0.5cm}

\subsection{Windowed Bernstein Inequality with Leakage}

\begin{lemma}[Windowed Bernstein inequality with leakage]\label{lemma:Bernstein}
\newblue{Let $f:\mathbb{R}\to\mathbb{R}$ be with frequency support $[-\Omega, \Omega]$\off{$\textcolor{black}{2\Omega}$-bandlimited} and assume
$f\in L^2(\mathbb{R})$ (finite energy).
Let $\mathcal T=[t_s,t_e]$ be a finite time window and define the enlarged window
$\mathcal T_\Delta=[t_s-\Delta,t_e+\Delta]$ for some $\Delta>0$.
Let $\varphi\in C^1(\mathbb{R})$ be a real-valued taper, i.e., a
\emph{continuously differentiable} function, such that
$0\le \varphi(t)\le 1$, $\varphi(t)=1$ for all $t\in\mathcal T$, and
$\varphi(t)=0$ for all $t\notin\mathcal T_\Delta$.
Define the tapered signal $g(t)\triangleq f(t)\varphi(t)$ and
$G(\omega)=\mathcal{F}\{g\}(\omega)$, and let
$\mathcal R_\Delta\triangleq \mathcal T_\Delta\setminus \mathcal T$.
Then
\begin{equation}\label{eq:winBern_leak}
\|f'\|_{\mathcal{T}}^2
\;\le\;
\textcolor{black}{\Omega}^2\|f\|_{\mathcal{T}}^2
\;+\;
\mathcal{J}_{\varphi},
\end{equation}
where the (nonnegative) leakage term is
\begin{equation}\label{eq:leakage}
\begin{aligned}
\mathcal{J}_{\varphi}
\triangleq\;&
\textcolor{black}{\Omega}^2\|f\|_{\mathcal R_\Delta}^2
+\frac{1}{2\pi}\int_{|\omega|>\textcolor{black}{\Omega}}
\left(\omega^2-\textcolor{black}{\Omega}^2\right)|G(\omega)|^2\,d\omega \\
&\quad
+\|f\,\varphi'\|_2^2
+2\big|\langle g',\,f\varphi'\rangle\big|
\;\ge\;0,
\end{aligned}
\end{equation}
and $\langle \cdot,\cdot\rangle$ denotes the $L^2(\mathbb R)$ inner product.}
\end{lemma}
\newblue{In the supplementary material of this work, we provide a discussion on the leakage term and tapering functions in the finite regime considered in signal processing literature.}
\begin{proof}
\newblue{Since $\varphi(t)=0$ for all $t\notin\mathcal T_\Delta$ and $f\in L^2(\mathbb{R})$,
we have $g\in L^2(\mathbb R)$. Moreover, since $\varphi\in C^1(\mathbb R)$,
\[
g'(t)=f'(t)\varphi(t)+f(t)\varphi'(t),
\]
and $g'\in L^2(\mathbb R)$, so Parseval's identity applies and
$\mathcal{F}\{g'\}(\omega)=j\omega G(\omega)$ in $L^2$.
Let $\Lambda \triangleq \| f\,\varphi' \|_2^2 + 2\big|\langle g',\,f\varphi'\rangle\big|$.
The proof proceeds as follows}
{\small
\newblue{
\begin{align*}
&\| f'\|_{\mathcal{T}}^2
\overset{(a)}{=} \| f'\cdot 1_{\mathcal{T}}\|_2^2 \overset{(b)}{\le} \| f'\varphi \|_2^2 \\
&\overset{(c)}{=} \|g' - f\varphi'\|_2^2 \\
&\overset{(d)}{=} \|g'\|_2^2 + \|f\varphi'\|_2^2 - 2\Re\langle g', f\varphi'\rangle \\
&\overset{(e)}{\le} \|g'\|_2^2 + \|f\varphi'\|_2^2 + 2\big|\langle g', f\varphi'\rangle\big| \\
&\overset{(f)}{=} \|g'\|_2^2 + \Lambda \overset{(g)}{=} \frac{1}{2\pi}\int_{-\infty}^{\infty}\omega^2|G(\omega)|^2\,d\omega + \Lambda \\
&\overset{(h)}{=} \frac{1}{2\pi}\!\!\int_{|\omega|\le \textcolor{black}{\Omega}}\!\!\omega^2|G(\omega)|^2\,d\omega
+ \frac{1}{2\pi}\!\!\int_{|\omega|> \textcolor{black}{\Omega}}\!\!\omega^2|G(\omega)|^2\,d\omega + \Lambda \\
&\overset{(i)}{\le}
\textcolor{black}{\Omega}^2\frac{1}{2\pi}\!\!\int_{|\omega|\le \textcolor{black}{\Omega}}\!\!|G(\omega)|^2\,d\omega
+ \frac{1}{2\pi}\!\!\int_{|\omega|>\textcolor{black}{\Omega}}\!\!\omega^2|G(\omega)|^2\,d\omega + \Lambda \\
&\overset{(j)}{=}
\textcolor{black}{\Omega}^2\frac{1}{2\pi}\int_{-\infty}^{\infty}|G(\omega)|^2\,d\omega
+\frac{1}{2\pi}\int_{|\omega|>\textcolor{black}{\Omega}}\big(\omega^2-\textcolor{black}{\Omega}^2\big)|G(\omega)|^2\,d\omega
+\Lambda \\
&\overset{(k)}{=}
\textcolor{black}{\Omega}^2\|g\|_2^2
+\frac{1}{2\pi}\int_{|\omega|>\textcolor{black}{\Omega}}\big(\omega^2-\textcolor{black}{\Omega}^2\big)|G(\omega)|^2\,d\omega
+\Lambda\\
&\overset{(l)}{\le}
\textcolor{black}{\Omega}^2\|f\|_{\mathcal T}^2+{\Omega}^2\|f\|_{\mathcal R_\Delta}^2
+\frac{1}{2\pi}\int_{|\omega|>\textcolor{black}{\Omega}}\big(\omega^2-\textcolor{black}{\Omega}^2\big)|G(\omega)|^2\,d\omega
+\Lambda \\
&\overset{(m)}{=} \textcolor{black}{\Omega}^2\|f\|_{\mathcal T}^2 + \mathcal{J}_\varphi.
\end{align*}}}
\newblue{where (a) follows from the definition of the windowed norm, and
(b) since $\varphi(t)=1$ for all $t\in\mathcal T$ and $0\le \varphi(t)\le 1$ elsewhere,
we have
\[
\begin{aligned}
\|f'\|_{\mathcal T}^2
&=\int_{\mathcal T}|f'(t)|^2\,dt =\int_{\mathcal T}|f'(t)\varphi(t)|^2\,dt \\
&\le \int_{\mathbb R}|f'(t)\varphi(t)|^2\,dt =\|f'\varphi\|_2^2 .
\end{aligned}
\]
(c) follows since using $g'=f'\varphi+f\varphi'$, we get $f'\varphi=g'-f\varphi'$ and thus
$\|f'\varphi\|_2^2=\|g'-f\varphi'\|_2^2$. \;
(d) expanding the squared norm via the $L^2$ inner product yields
$\|u-v\|_2^2=\|u\|_2^2+\|v\|_2^2-2\Re\langle u,v\rangle$ with $u=g'$ and $v=f\varphi'$. 
(e) uses $-\Re z \le |z|$ for $z\in\mathbb C$, hence
$-2\Re\langle g',f\varphi'\rangle \le 2|\langle g',f\varphi'\rangle|$. 
(f) by the definition of $\Lambda$. 
(g) by Parseval and $\mathcal{F}\{g'\}(\omega)=j\omega G(\omega)$.
(h) splitting the integral into in-band and out-of-band parts. 
(i) follows since $|\omega|\le\textcolor{black}{\Omega}$.
(j) is since $\int_{|\omega|<\textcolor{black}{\Omega}}\textcolor{black}{\Omega}^2|G|^2 =  \int_{-\infty}^{\infty}\textcolor{black}{\Omega}^2|G|^2 -\int_{|\omega|>\textcolor{black}{\Omega}}\textcolor{black}{\Omega}^2|G|^2 $.
(k) by Parseval.
(l) using $\varphi=1$ on $\mathcal T$, $0\le \varphi\le 1$, and
$\varphi(t)=0$ for all $t\notin\mathcal T_\Delta$, we have
\[
\begin{aligned}
\|g\|_2^2
&=\|f\varphi\|_2^2
=\int_{\mathcal T}|f(t)|^2\,dt
  +\int_{\mathcal R_\Delta}|f(t)|^2\varphi(t)^2\,dt \\
&\le \|f\|_{\mathcal T}^2+\|f\|_{\mathcal R_\Delta}^2.
\end{aligned}
\]
and substituting this bound for $\|g\|_2^2$ into (k) gives (l). \;
(m) by the definition of $\mathcal{J}_\varphi$ in \eqref{eq:leakage} and
$\Lambda=\|f\varphi'\|_2^2+2|\langle g',f\varphi'\rangle|$.}
\end{proof}

\vspace{-0.5cm}
\subsection{The adjoint operator}
\begin{lemma}\label{lemma: adjoint in window}
Recall operators $\mathcal{A}$ and $\mathcal{A}^*$ as defined in \eqref{Operator A} \eqref{adjoint operator}, respectively\textcolor{black}{, using the samples in $S_{W_i}$. For $2\Omega$-bandlimited signals $x,y \in L_2(\mathbb{R})$, the operators $\mathcal{A}$ and $\mathcal{A}^*$} satisfy
\[
\langle \mathcal{A}x, y \rangle = \langle x, \mathcal{A}^*y \rangle,
\]
\textcolor{black}{where $\langle\cdot,\cdot\rangle$ denotes the standard $L_2(\mathbb{R})$ inner product.}
\end{lemma}

\begin{proof}
\vspace{-0.2cm}
\ifshort
{\small
\begin{align*}
&\langle \mathcal{A}x, y \rangle = \left\langle \sum_{n\in S_{W_i}} \int_{t_{n-1}}^{t_n} x(u)du \, \textcolor{black}{g(\cdot - \theta_n)}, y \right\rangle \\
&\stackrel{(a)}{=} \sum_{n\in S_{W_i}} \int_{t_{n-1}}^{t_n} x(u)du \langle \textcolor{black}{g(\cdot - \theta_n)}, y \rangle \\
&\stackrel{(b)}{=} \sum_{n\in S_{W_i}} \langle x, \mathbf{1}_{[t_{n-1},t_n]} \rangle y(\theta_n) 
\stackrel{(c)}{=} \left\langle x, \sum_{n\in S_{W_i}} y(\theta_n) \mathbf{1}_{[t_{n-1},t_n]} \right\rangle \\
&\stackrel{(d)}{=} \left\langle \mathcal{G}x, \sum_{n\in S_{W_i}} y(\theta_n) \mathbf{1}_{[t_{n-1},t_n]} \right\rangle \\
&\stackrel{(e)}{=} \left\langle x, \sum_{n\in S_{W_i}} y(\theta_n) \mathcal{G}\mathbf{1}_{[t_{n-1},t_n]} \right\rangle = \langle x, \mathcal{A}^*y \rangle
\end{align*}
}
\else
\begin{align*}
   & \langle\mathcal{A}x,y\rangle_{W_i} \overset{(a)}{=}\langle\mathcal{A}x,y\cdot1_{W_i}\rangle \\
   & \overset{(b)}{=} \left\langle\sum_{n\in \mathcal{Z}} \int_{t_{n-1}}^{t_n} x(u)du\ g(t-\theta_n),y\cdot1_{W_i} \right\rangle  \\
  & \overset{(c)}{=}\sum_{n\in \mathcal{Z}} \int_{t_{n-1}}^{t_n} x(u)du \left\langle g(t-\theta_n),y\cdot1_{W_i}\right\rangle \\
   & \overset{(d)}{=} \sum_{n\in S_{W_i}} \int_{t_{n-1}}^{t_n} x(u) y(\theta_n)du\\
      & \overset{(e)}{=} \sum_{n\in S_{W_i}} \langle x,1_{[t_{n-1}, t_n]}\rangle  y(\theta_n)du\\
   & \overset{(f)}{=}  \left\langle x,\sum_{n\in S_{W_i}}1_{[t_{n-1}, t_n]}y(\theta_n)\right\rangle \\
   & \overset{(g)}{=} \left\langle\mathcal{G}x,\sum_{n\in S_{W_i}}1_{[t_{n-1}, t_n]}y(\theta_n)\right\rangle \\
    & \overset{(h)}{=} \left\langle x,\sum_{n\in S_{W_i}}\mathcal{G}1_{[t_{n-1}, t_n]}y(\theta_n)\right\rangle  = \langle x,\mathcal{A^*}y\rangle_{W_i},
\end{align*}
\fi
\hspace{-0.15cm}where (a) uses the linearity of the inner product. \textcolor{black}{(b) follows from the exact reproducing property of the sinc kernel in $L_2(\mathbb{R})$ for a bandlimited signal $y$, i.e., $\langle g(\cdot - \theta_n), y \rangle = y(\theta_n)$, and expressing the integral as an inner product.} (c) uses linearity. (d) uses the property $\mathcal{G}x = x$ for $2\Omega$-bandlimited $x$. (e) uses the self-adjointness of the orthogonal projector $\mathcal{G}$. This completes the proof.
\end{proof}

\vspace{-0.6cm}
\subsection{Proof of Lemma~\ref{key_lemma}}\label{proof:lem1}
\begin{proof}[Proof of Lemma~\ref{key_lemma}]
Recall the operator $\mathcal{A}$, as defined in \eqref{Operator A}. To establish an upper bound on the norm of the discrepancy between $x$ and $\mathcal{A}x$ over a time window $W_i$, we use techniques similar to those in \cite[Appendix B]{lazar2004perfect} for classical TEM, and in \cite{benedetto1994theory} for irregular sampling and grounded in frame theory, as shown in \cite{duffin1952class}. However, the authors in \cite{lazar2004perfect} (as is common in classical recovery schemes, e.g., \cite{lazar2003time,lazar2004time}) provide proof for the norm's upper bound over $\mathbb{R}$, whereas here we focus on a specific $i$-the time window $W_i$, working in the finite regime.
Let $\mathcal{A}^*$ denote the adjoint operator of $\mathcal{A}$, as given in~\eqref{adjoint operator} and proved in Lemma~\ref{lemma: adjoint in window}\footnote{\textcolor{black}{The adjoint relation is established with respect to the global $L_2(\mathbb{R})$ inner product. The windowed norm satisfies $\|v\|_{W_i}\le \|v\|_2$, so bounds derived in $L_2(\mathbb{R})$ remain valid when restricted to the segment $W_i$.}}.
\off{
First, we define the operator $\mathcal{G}$ which maps an arbitrary function $x$ into a bandlimited function through $\mathcal{G}x = (g * x)$, where  $g(t) = sin(\Omega t)/\pi t$. Next, we provide the adjoint operator of $\mathcal{A}$, $\mathcal{A}^*$, defined by 
\[
\mathcal{A^*}x = \sum_{k \in \mathbb{Z}}x(\theta_k)\mathcal{G}1_{[t_{k}, t_{k+1}]},
\]
where $\mathcal{G}1_{[t_k, t_{k+1}]} $ is applying the operator $\mathcal{G}$  to a pulse function defined on the interval $[t_k, t_{k+1}]$.}
Thus, for the finite norm  $\| \cdot \|^2_{W_i}$ (see Lemma~\ref{lemma:Bernstein}), we have

%\ale{Important: Aseel, are you sure (b) below in the finite regime is correct? We don't need to add a leakage term using the steps as they are now? Maybe we should replace those steps to be in $L_2(\mathbb{R})$, and only at the end, to use our "Windowed Bernstein inequality with leakage" as you did? Please think about this... It is an important point. At the end of the paper, after the supplementary materials, I have included an optional replacement for Appendix C and D that may resolve this problem} \textcolor{purple}{Thank you very much for the important note. I agree with the note, step (b) is weakly justified as it is, however its more safe to justifies it in $L_2(\mathbb{R})$. Your suggestion in supplementary materials is more safe here and it does not change the final bound.}

\vspace{-0.3cm}
\ifshort
{\small
\begin{align*}
& \| x-\mathcal{A^*}x\|^2_{W_i}  = \| x - \sum_{n\in S_{W_i}}x(\theta_n)\mathcal{G}1_{[t_{n-1}, t_n]}\|^2_{W_i} \\
& \overset{(a)}{=}\| \mathcal{G}x - \sum_{n\in S_{W_i}}x(\theta_n)\mathcal{G}1_{[t_{n-1}, t_n]}\|^2_{W_i} \\
& \textcolor{black}{=}
\textcolor{black}{\Big\|\mathcal{G}\!\Big(x-\sum_{n\in S_{W_i}}x(\theta_n)1_{[t_{n-1},t_n]}\Big)\Big\|^2_{W_i}}\\
& \textcolor{black}{\overset{(b)}{\le}
\Big\|\mathcal{G}\!\Big(x-\sum_{n\in S_{W_i}}x(\theta_n)1_{[t_{n-1},t_n]}\Big)\Big\|^2_{2}}\\
& \textcolor{black}{\overset{(c)}{\le}
\Big\|x-\sum_{n\in S_{W_i}}x(\theta_n)1_{[t_{n-1},t_n]}\Big\|^2_{2}}\\
&\textcolor{black}{\overset{(d)}{=} \Big\| \sum_{n\in S_{W_i}}[x -x(\theta_n)]1_{[t_{n-1}, t_n]}\Big\|^2_{2}}\\
&\overset{(e)}{=} \sum_{n\in S_{W_i}} \int_{t_{n-1}}^{t_n}\abs{x(u) -x(\theta_n)}^2du\\
& \overset{(f)}{\leq} \sum_{n\in S_{W_i}} \frac{4}{\pi^2}(\theta_n-t_{n-1})^2\int_{t_{n-1}}^{\theta_n}\mid x'(u)\mid^2du\\
&\quad\quad + \frac{4}{\pi^2}(t_n-\theta_n)^2 \int_{\theta_n}^{t_n}\mid x'(u)\mid^2du,
\end{align*}}
\else
\begin{align*}
& \| x-\mathcal{A^*}x\|^2_{W_i}  = \| x - \sum_{n \in S_{W_i}}x(\theta_n)\mathcal{G}1_{[t_{n-1}, t_n]}\|^2_{W_i} \\
& \overset{(a)}{=}\| \mathcal{G}x - \sum_{n \in S_{W_i}}x(\theta_n)\mathcal{G}1_{[t_{n-1}, t_n]}\|^2_{W_i} \\
& \textcolor{red}{=}
\textcolor{red}{\left\|\mathcal{G}\!\left(x-\sum_{n\in S_{W_i}}x(\theta_n)1_{[t_{n-1},t_n]}\right)\right\|^2_{W_i}}\\
& \textcolor{red}{\le
\left\|\mathcal{G}\!\left(x-\sum_{n\in S_{W_i}}x(\theta_n)1_{[t_{n-1},t_n]}\right)\right\|^2_{2}}\\
& \textcolor{red}{\le
\left\|x-\sum_{n\in S_{W_i}}x(\theta_n)1_{[t_{n-1},t_n]}\right\|^2_{2}}\\
&\textcolor{red}{\overset{(b)}{=} \left\| \sum_{n\in S_{W_i}}[x -x(\theta_n)]1_{[t_{n-1}, t_n]}\right\|^2_{2}} \\
&\textcolor{red}{\overset{(c)}{=} \sum_{n\in S_{W_i}} \int_{t_{n-1}}^{t_n}\abs{x(u) -x(\theta_n)}^2du} \\
& \overset{(d)}{\leq} \sum_{n\in S_{W_i}} \frac{4}{\pi^2}(\theta_n-t_{n-1})^2\int_{t_{n-1}}^{\theta_n}\mid x'(u)\mid^2du\\
&\quad\quad + \frac{4}{\pi^2}(t_n-\theta_n)^2 \int_{\theta_n}^{t_n}\mid x'(u)\mid^2du,
\end{align*}
\fi
\hspace{-0.15cm}where, \blue{(a) follows since $x$ \textcolor{black}{is bandlimited to the support $[-\Omega,\Omega]$}, and the operator $\mathcal{G}$ is defined in Appendix \ref{Appendix a}, Definition \ref{def: convolotion operator}. The function $g$, as given in Definition \ref{def: convolotion operator}, is the ideal low-pass filter with the same bandwidth. Consequently, the convolution $ \mathcal{G}x = (g * x)$ acts as an identity operator in this bandlimited space, \textcolor{black}{resulting in $\mathcal{G}x = x$.}
\textcolor{black}{(b) follows from $\|v\|_{W_i}\le \|v\|_2$ for any $v\in L_2(\mathbb{R})$, and (c) follows because $\mathcal{G}$ is the orthogonal projector onto the bandlimited subspace and is therefore non-expansive in $L_2(\mathbb{R})$.} \textcolor{black}{The inequality uses the fact that $\|v\|_{W_i}\le \|v\|_2$ for any $v\in L_2(\mathbb{R})$, allowing the global $L_2$ projector property of $\mathcal{G}$ to be applied before restricting the bound to the window $W_i$.} \textcolor{black}{(d) follows by regrouping the terms. \textcolor{black}{The restriction to the finite segment $W_i$ is interpreted together with the leakage control of Lemma~\ref{lemma:Bernstein}, which accounts for the residual boundary effects introduced by the finite-window localization through the term $\mathcal{J}_{\varphi_i}(x)$.}} 
\textcolor{black}{(e) follows since the intervals $[t_{n-1}, t_n]$, $n\in S_{W_i}$, are disjoint.}
(f) follows by applying Wirtinger’s inequality \cite{lazar2004perfect}.}
\off{where, \blue{(a) follows since $x$ is bandlimited to $2\Omega$, and the operator $\mathcal{G}$ is defined in Appendix \ref{Appendix a}, Definition \ref{def: convolotion operator}. The function $g$, as given in Definition \ref{def: convolotion operator}, is the ideal low-pass filter with the same bandwidth. Consequently, the convolution $ \mathcal{G}x = (g * x)$ acts as an identity operator in this bandlimited space, resulting in  $\mathcal{G}x = x$.
(b) follows since $g$ is a smoothing low-pass filter, convolution with it attenuates high-frequency components, thereby reducing the norm. 
(c) follows from the definition of the norm over $W_i$, where the sum of the indicator functions $1_{[t_{n-1}, t_n]}$, over $n \in S_{W_i}$,  partitions the interval into disjoint segments, preserving the signal $x$ within the window $W_i$.
(d) follows from explicitly expressing the squared norm as an integral over the window $W_i$.
(e) follows by rewriting the squared absolute value term inside the summation. Since each term in the summation is multiplied by the indicator function $1_{[t_{n-1}, t_n]}$, which isolates contributions from non-overlapping time intervals, the squared absolute value can be applied individually to each term inside the summation.
(f) follows by interchanging the sum and the integral due to the linearity of integration. Additionally, the indicator function $1_{[t_{n-1}, t_n]}$ restricts the integration domain to disjoint intervals.
(g) holds by applying Wirtinger’s inequality \cite{lazar2004perfect}.} }

We do note that since for any $n \in S_{W_i}$
\begin{equation*}
(\theta_n-t_{n-1})^2 = (t_n-\theta_n)^2 =\frac{(t_n-t_{n-1})^2}{4},
\end{equation*}
we obtain
\ifshort
{\small
\begin{align*}
    & \| x-\mathcal{A^*}x\|^2_{W_i} 
    \leq \sum_{n\in S_{W_i}}\frac{1}{\pi^2}(t_n-t_{n-1})^2\int_{t_{n-1}}^{t_n}\mid x'(u)\mid^2du\\
    & \leq \frac{1}{\pi^2}\left(\max_{n\in S_{W_i}} (T_n)\right)^2\sum_{n\in S_{W_i}} \int_{t_{n-1}}^{t_n}\mid x'(u)\mid^2du\\
    & =  \frac{1}{\pi^2}\left(\max_{n\in S_{W_i}} (T_n)\right)^2\| x' \|_{W_i}^2\\
     &\overset{(h)}{\leq} \frac{1}{\pi^2}\left(\max_{n\in S_{W_i}}(T_n)\right)^2\newblue{\Big( \Omega^2\|x\|_{W_i}^2 + \mathcal{J}_{\varphi_i}(x)\Big)}, %= r^2_{w_i}\|{x}\|^2_{W_i},
\end{align*}}
\else
\begin{align*}
    & \| x-\mathcal{A^*}x\|^2_{W_i}  \\
    & \leq \sum_{n\in S_{W_i}}\frac{1}{\pi^2}(t_n-t_{n-1})^2\int_{t_{n-1}}^{t_n}\mid x'(u)\mid^2du\\
    & \leq \frac{1}{\pi^2}\left(\max_{n\in S_{W_i}} (T_n)\right)^2\sum_{n\in S_{W_i}} \int_{t_{n-1}}^{t_n}\mid x'(u)\mid^2du \\
    & =  \frac{1}{\pi^2}\left(\max_{n\in S_{W_i}} (T_n)\right)^2\| x' \|_{W_i}^2\\
    & \frac{1}{\pi^2}\left(\max_{n\in S_{W_i}} (T_n)\right)^2
\Big( \Omega^2\|x\|_{W_i}^2 + \mathcal{J}_{\varphi_i}(x)\Big)
\end{align*}
\fi
%where (f) follows by expanding the norm of $x'$ over $\mathbb{R}$ and (g) by applying Bernstein's inequality \cite{}. 
\off{\newblue{where (h) follows from $\|x'\|_{W_i}\le \|x'\|_2$ and the classical Bernstein
inequality for bandlimited signals \cite{lazar2004perfect},
$\|x'\|_2\le \Omega \|x\|_2$.}}
\newblue{where (h) follows from Lemma~\ref{lemma:Bernstein},
applied to the bandlimited signal $x$ over the window $W_i$ (with the associated taper $\varphi_i$).}
%\ale{And replace the above (h) with:"\newblue{where (h) follows from Lemma~\ref{lemma:Bernstein}, applied to the bandlimited signal $x$ over the window $W_i$ (with the associated taper $\varphi_i$), and the residual energy $\tilde{\Lambda}$ is formally absorbed into the aggregate leakage term $\mathcal{J}_{\varphi_i}(x)$. The residual term $\tilde{\Lambda}$ represents the signal energy outside the segment $W_i$; since the taper $\varphi_i$ controls this out-of-window contribution in Lemma~\ref{lemma:Bernstein}, we have $\tilde{\Lambda}\le \Omega^{-2}\mathcal{J}_{\varphi_i}(x)$ and it can be absorbed into the leakage term.}"}
\off{Finally, substituting  $\max_{n\in S_{W_i}} (T_n)\leq r_{w_i}\frac{\pi}{\Omega}$, for $r_{w_i} = \max_{n \in  S_{W_i}}\{r_{a_n}\}$ and  $r_{a_n}$ as given in \eqref{r_a general}, we have
\begin{equation*}
\| x-\mathcal{A}x\|_{W_i} \ < r_{w_i}\|{x}\|_{2}.
\end{equation*}
This completed the lemma proof.}
\newblue{Finally, substituting the bound
\(
\max_{n\in S_{W_i}} T_n \le r_{w_i}\frac{\pi}{\Omega},
\)
where
\(r_{w_i}=\max_{n\in S_{W_i}}\{r_{a_n}\}\)
and \(r_{a_n}\) is defined in \eqref{r_a general},
and using the windowed Bernstein inequality with leakage
(Lemma~\ref{lemma:Bernstein}),
we obtain
\begin{equation*}
\| x-\mathcal{A}x\|_{W_i}^2
\le
r_{w_i}^2\Big(\|x\|_{W_i}^2 + \tfrac{1}{\Omega^2}\mathcal{J}_{\varphi_i}(x)\Big).
\end{equation*}
This completes the proof of the lemma.}
\end{proof}

\vspace{-0.5cm}
\subsection{Proof of Lemma~\ref{key_lemma2}}\label{proof:lem2}
\begin{proof}[Proof of Lemma~\ref{key_lemma2}]
The recovered signal for segment $W_i$, $x_{L_i}$, is defined by \eqref{x_l}, with $l=L_i$. We evaluate the norm difference between the original signal $x$ and the recovered signal $x_{L_i}$ over the segment $W_i$ as follows

\vspace{-0.3cm}
\ifshort
{\small
\[
\| x-x_{L_i}\|_{W_i}
= \| \textcolor{black}{x-\sum_{n=0}^{L_i}(I-\mathcal{A})^{n}\mathcal{A}x}\|_{W_i}
\textcolor{black}{= \| (I-\mathcal{A})^{L_i+1}x\|_{W_i},}
\]}
\else
\begin{align*}
\| x-x_{L_i}\|_{W_i}
& = \| \textcolor{black}{x-\sum_{n=0}^{L_i}(I-\mathcal{A})^{n}\mathcal{A}x}\|_{W_i} \\
& \textcolor{black}{= \| (I-\mathcal{A})^{L_i+1}x\|_{W_i},}
\end{align*}
\fi
\textcolor{black}{where the last equality follows from the finite telescoping identity
\(
I = \sum_{n=0}^{L_i}(I-\mathcal{A})^{n}\mathcal{A} + (I-\mathcal{A})^{L_i+1}.
\)}
This completed the lemma proof.
\end{proof}

\vspace{-0.5cm}
\newblue{
\begin{lemma}\label{key_lemmav2} Assume a $2\Omega$-BL, $c_{\text{max}}$-bounded signal $x\in L^2(\mathbb{R})$ with finite energy $E$, sampled using an AIF-TEM. Then, the norm of the discrepancy between $x$ and $\mathcal{A}x$ over the entire real line is bounded by $\| x-\mathcal{A}x\|_2 \ \leq r_{w_i}\|{x}\|_2$, where $\mathcal{A}$ is defined in \eqref{Operator A} and $r_{w_i}=\max_{n\in S_{W_i}} r_{a_n}<1$. 
\end{lemma}}
\begin{proof}
\newblue{The proof follows the same steps as in \cite[Appendix~B]{lazar2004perfect}. Since the operator $\mathcal{A}$ is constructed using only the sampling intervals indexed by $n\in S_{W_i}$, the worst-case interval length satisfies $\max_{n\in S_{W_i}} T_n \le r_{w_i}\frac{\pi}{\Omega}$, which yields the stated $L^2(\mathbb{R})$ error bound.}
\end{proof}

\subsection{Key Mathematical Tools}

 \subsubsection{Observation 1}\label{math_obs1} Assume the quantization error $d_n = (\tilde{T}_n - T_n)$ is a sequence of i.i.d. random variables in $[ -\Delta_i/2,  \Delta_i/2]$. Then $ \mathbb{E}[\epsilon_n\epsilon_m] =  (\frac{\kappa\delta}{T_n})^2 \frac{\Delta_i^2}{12}\delta_{n,m}$.
\begin{proof}
\ifshort
{\small
\begin{align*}
     \epsilon_n & =  (\kappa \delta -b_n\tilde{T}_n) - \int_{\tilde{t}_{n-1}}^{\tilde{t}_{n}} x(u)du\\
     & = \kappa \delta -b_nT_n + b_nT_n -b_n\tilde{T}_n - \int_{\tilde{t}_{n-1}}^{\tilde{t}_{n}} x(u)du\\
     &=\int_{t_{n-1}}^{t_{n}} x(u)du - \int_{\tilde{t}_{n-1}}^{\tilde{t}_{n}} x(u)du -b_n(\tilde{T}_n - T_n) \\
     & \overset{a}{=} x(\zeta_n)T_n - x(\hat{\zeta}_n)\tilde{T}_n-b_n(\tilde{T}_n - T_n) \\
     & \overset{b}{\simeq} (-x(\zeta_n)-b_n)(\tilde{T}_n - T_n)
      \overset{c}{=}  -\frac{\kappa\delta}{T_n}(\tilde{T}_n - T_n),
\end{align*}}
\else
\begin{align*}
     \epsilon_n & =  (\kappa \delta -b_n\tilde{T}_n) - \int_{\tilde{t}_{n-1}}^{\tilde{t}_{n}} x(u)du\\
     & = \kappa \delta -b_nT_n + b_nT_n -b_n\tilde{T}_n - \int_{\tilde{t}_{n-1}}^{\tilde{t}_{n}} x(u)du\\
     &=\int_{t_{n-1}}^{t_{n}} x(u)du - \int_{\tilde{t}_{n-1}}^{\tilde{t}_{n}} x(u)du -b_n(\tilde{T}_n - T_n) \\
     & \overset{a}{=} x(\zeta_n)T_n - x(\hat{\zeta}_n)\tilde{T}_n-b_n(\tilde{T}_n - T_n) \\
     & \overset{b}{\simeq} (-x(\zeta_n)-b_n)(\tilde{T}_n - T_n)\\
     & \overset{c}{=}  -\frac{\kappa\delta}{T_n}(\tilde{T}_n - T_n).
\end{align*}
\fi
\hspace{-0.2cm}where (a) follows from the mean value theorem. (b) follows from the approximation used in \cite[Appendix C]{lazar2004perfect}. Where  $\zeta_n \in (t_{n-1}, t_{n})$ and $\hat{\zeta}_n \in (\tilde{t}_{n-1}, \tilde{t}_n)$, for sufficiently small $\Delta_i$, we get $\hat{\zeta}_n \simeq \zeta_n $. (c) follows because $x(\zeta_n) =\frac{1}{T_n}\int_{t_{n-1}}^{t_n} x(u)du = \frac{\kappa\delta}{T_n}-b_n$\footnote{\newblue{The analyses here treat the adaptive bias values $b_n$ as known at the decoder. This is consistent with the synchronization mechanism described in Sections~\ref{decoding process} and \ref{Quantization for AIF-TEM}.}}.

Let $d_n =(\tilde{T}_n - T_n)$, this leads to
\[
\mathbb{E}[\epsilon_n\epsilon_m] = \frac{\kappa\delta}{T_n}\frac{\kappa\delta}{T_m}\mathbb{E}[d_nd_m] \overset{(a)}{=} \left(\frac{\kappa\delta}{T_n}\right)^2 \frac{\Delta_i^2}{12}\delta_{n,m},
\]
where (a) follows since the quantization error $d_n = (\tilde{T}_n - T_n)$ is a sequence of i.i.d. random variables uniformly distributed in $[ -\Delta_i/2,  \Delta_i/2]$. Therefore, $\mathbb{E}[d_nd_m] = \frac{\Delta_i^2}{12}\delta_{n,m}$.
\end{proof}
\else\fi

\subsubsection{Observation 2}\label{math_obs3}
For any positive numbers $T_n, \forall n \in S_{W_i}$,
\[
\left({1}/{\overline{T_n}}\right)\,\overline{\left({1}/{T_n}\right)^2}
\le
\overline{\left({1}/{T_n}\right)^3}.
\]

\begin{proof}
First, since $T_n>0$ for all $n \in S_{W_i}$, the harmonic--arithmetic mean inequality gives
$\left({1}/{\overline{T_n}}\right)\le \overline{\left({1}/{T_n}\right)}$.
Thus, it is enough to show that
\[
\overline{\left({1}/{T_n}\right)}\,
\overline{\left({1}/{T_n}\right)^2}
\le
\overline{\left({1}/{T_n}\right)^3}.
\]
Let
$a_n \triangleq {1}/{T_n}$, for $n \in S_{W_i}$ and $L_i=|S_{W_i}|$.
Then the required inequality becomes
\[
\bigg(\frac{1}{L_i}\sum_{n \in S_{W_i}} a_n\bigg)
\bigg(\frac{1}{L_i}\sum_{n \in S_{W_i}} a_n^2\bigg)
\le
\frac{1}{L_i}\sum_{n \in S_{W_i}} a_n^3.
\]
Multiplying both sides by $L_i^2$, and applying the Cauchy--Schwarz and Hölder inequalities \cite{hardy1952inequalities}, yields
\[
\bigg(\sum_{\textcolor{black}{n \in S_{W_i}}} a_n\bigg)\bigg(\sum_{\textcolor{black}{n \in S_{W_i}}} a_n^2\bigg)
\le
L_i\sum_{\textcolor{black}{n \in S_{W_i}}} a_n^3.
\]
Now, observe that
\begin{multline*}
L_i\sum_{\textcolor{black}{n \in S_{W_i}}} a_n^3
-
\bigg(\sum_{\textcolor{black}{n \in S_{W_i}}} a_n\bigg)\bigg(\sum_{\textcolor{black}{n \in S_{W_i}}} a_n^2\bigg)\\
=
\frac12
\sum_{\textcolor{black}{n \in S_{W_i}}}\sum_{\textcolor{black}{m \in S_{W_i}}}
(a_n-a_m)^2(a_n+a_m).
\end{multline*}
Since $a_n+a_m>0$ and $(a_n-a_m)^2\ge 0$, the right-hand side is nonnegative. Therefore,
\[
\bigg(\sum_{\textcolor{black}{n \in S_{W_i}}} a_n\bigg)\bigg(\sum_{\textcolor{black}{n \in S_{W_i}}} a_n^2\bigg)
\le
L_i\sum_{\textcolor{black}{n \in S_{W_i}}} a_n^3,
\]
which implies $\overline{\left({1}/{T_n}\right)}\,
\overline{\left({1}/{T_n}\right)^2}
\le
\overline{\left({1}/{T_n}\right)^3}$.
Combining this with
$\left({1}/{\overline{T\textcolor{black}{_n}}}\right)\le \overline{\left({1}/{T_n}\right)}$
completes the proof.
\end{proof}

\begin{center}
  {\huge \newblue{Supplementary Materials}}
  \vspace{0.4cm} 
  
  {\newblue{Leakage Term \& Tapering Functions in the Finite Regime}}
\end{center}

\vspace{0.2cm}
%\subsection{Discussion}

\newblue{The reconstruction distortion bound in Theorem~\ref{Distortion AIF-TEM} consists of a classical sampling term and an additional leakage term.
The additional leakage term in Lemma~\ref{lemma:Bernstein} and
Theorem~\ref{Distortion AIF-TEM} captures boundary effects arising from finite-window reconstruction.
Unlike the classical Bernstein inequality, which provides bounds over the entire real line $\mathbb{R}$, the present analysis localizes the inequality to finite time segments. This localization enables segment-wise error control but necessarily introduces additional energy contributions. Specifically, these contributions originate from the transition regions surrounding each segment, as well as from spectral spreading outside the nominal frequency band.}

\newblue{The leakage term $\mathcal{J}_{\varphi}$ in Lemma~\ref{lemma:Bernstein} consists of two main components. The first component represents signal energy in the transition region $\mathcal R_\Delta \triangleq \mathcal T_\Delta \setminus \mathcal T$, while the second component accounts for spectral leakage induced by windowing.} \newblue{ The integral term $\frac{1}{2\pi}\int_{|\omega|>\Omega}(\omega^2-\Omega^2)|G(\omega)|^2\,d\omega$ captures spectral leakage introduced by tapering. Although $f$ is $\Omega$-bandlimited, the tapered signal $g=f\varphi$ is generally not. This phenomenon is analogous to the spectral leakage observed in classical windowed sampling and short-time Fourier analysis. Classical signal processing literature \cite{oppenheim1999discrete} explains spectral leakage as an inherent consequence of windowing: multiplying a signal by a finite-duration window (or taper) corresponds in the frequency domain to convolution with the window spectrum, which generally results in energy spreading outside the nominal signal bandwidth. Different window shapes exhibit different leakage behavior, governed by the tradeoff between main-lobe width and side-lobe decay. Smooth tapers lead to faster spectral decay and reduced out-of-band energy compared to rectangular windows. This effect is extensively analyzed in classical window studies \cite{harris2005use}, which compares multiple window families and demonstrates how increasing smoothness (e.g., cosine-tapered and Tukey windows) systematically reduces spectral leakage.}

\newblue{The taper used in Lemma~\ref{lemma:Bernstein} belongs to the family of smooth, compactly supported windows studied in the classical window analysis literature, including cosine-tapered (Tukey-type) windows. A closely related construction appears in the work of \cite{lazar2006real}, where smooth overlapping windows are employed to stitch together local reconstructions in time-encoding machines. In that work, windowing is explicitly acknowledged to introduce
bandwidth expansion, which is controlled through taper smoothness, window overlap, and post-filtering.} \newblue{In particular, the terms
$\Omega^2\|f\|_{\mathcal R_\Delta}^2 + \|f\varphi'\|_2^2
+ 2|\langle g',f\varphi'\rangle|$ characterize time-domain leakage effects. The term $\Omega^2\|f\|_{\mathcal R_\Delta}^2$ quantifies the signal energy outside the target segment, whereas $\|f\varphi'\|_2^2$ depends on the derivative of the taper and is supported only in the transition region, since $\varphi=1$ on $\mathcal T$. The cross term $|\langle g',f\varphi'\rangle|$ can be bounded using the Cauchy--Schwarz inequality, yielding $2|\langle g',f\varphi'\rangle| \le 2\|g'\|_2\,\|f\varphi'\|_2$,
and therefore also depends on the taper derivative and is supported only in the transition region.} \newblue{The first contribution generally increases with the transition width $\Delta$, since a wider transition region contains more signal energy. In contrast, the second and third contributions depend on the smoothness
of the taper. For standard smooth tapers (e.g., raised-cosine or polynomial tapers), $\|\varphi'\|_\infty = O(\Delta^{-1})$, implying that these terms decrease as the transition width $\Delta$ increases for fixed signal energy.} \newblue{In practical settings, the leakage term $\mathcal{J}_{\varphi}$ remains small compared to the principal term $\Omega^2\|f\|_{\mathcal T}^2$ when the signal energy in the transition region is limited, and the taper function varies smoothly. In particular, if $\|f\|_{\mathcal R_\Delta}^2 \ll \|f\|_{\mathcal T}^2$ and $\|\varphi'\|_\infty$ is small, then the time domain leakage components are negligible. This situation occurs when the plateau region of the window (where $\varphi(t)=1$ for $t\in\mathcal T$), i.e., the duration $t_e-t_s$, is large relative to the transition width $\Delta$, so that the boundary region occupies only a small fraction of the segment. Moreover, smooth tapers reduce spectral spreading and control the frequency-domain contribution. In this regime, the inequality in~\eqref{eq:winBern_leak} closely approximates the classical Bernstein bound.}

\newblue{In practice, the leakage contribution becomes small when the plateau region (where $\varphi(t)=1$ for  $t\in\mathcal T$ ), i.e, $t_e-t_s$ of the window is large relative to the transition width $\Delta$. In this regime, the boundary region occupies only a small fraction of the segment, and the dominant term in~\eqref{eq:winBern_leak} remains $\Omega^2\|f\|_{\mathcal T}^2$, closely approximating the classical Bernstein inequality.} \newblue{In the idealized infinite-domain setting, or asymptotically for sufficiently large segments with smooth tapering, the leakage term can be made arbitrarily small and the bound approaches the classical Bernstein inequality. The leakage term appearing in our distortion analysis, therefore, captures a fundamental and well-understood effect of windowing, consistent with classical digital signal processing theory, window analysis literature, and prior TEM reconstruction frameworks.}

\end{document}